\documentclass{CSML}

\def\dOi{13(1:10)2017}
\lmcsheading%
{\dOi}
{1--51}
{}
{}
{Jan.~15, 2016}
{Mar.~17, 2017}
{}

\usepackage{amssymb,amsmath,stmaryrd}
\usepackage{mathabx}
\usepackage{color}
\usepackage{xcolor}
\usepackage{graphicx}
\usepackage{epstopdf}

\usepackage{placeins}

\usepackage{hyperref}
\usepackage{url}

\theoremstyle{plain}

\newcommand{\mi}[1]{\mathit{#1}}
\newcommand{\mrm}[1]{\mathrm{#1}}

\newcommand{\Inference}[2]{\begin{array}{@{}c@{}}#1\\[0em]\hline\\[-0.9em]#2\\ \end{array}}

\newcommand{\MATCH}[2]{\mathit{match}(#1,#2)}
\newcommand{\subst}[2]{{#1}/{#2}}

\newcommand{\MUNION}{\uplus}

\newcommand{\UL}{\underline}
\newcommand{\LST}{\widetilde}
\newcommand{\LSTStruct}[2]{\langle\!\!|~\!{#1}~\!|\!\!\rangle_{{#2}}}
\newcommand{\LSTParenL}{{\langle\!\!|}}
\newcommand{\LSTParenR}{{|\!\!\rangle}}

\newcommand{\Comment}[1]{}

\newcommand{\lloc}{\ell}
\newcommand{\llocConst}{l}

\newcommand{\ttloc}{{\tt l}}

\newcommand{\NIL}{{\sf nil}}
\newcommand{\NILN}{{\sf nil}_{\mrm{N}}}
\newcommand{\NILP}{{\sf nil}_{\mrm{P}}}
\newcommand{\PREF}[2]{{#1}.{#2}}
\newcommand{\PAR}[2]{{#1}|{#2}}
\newcommand{\CALL}[2]{#1(#2)}

\newcommand{\FOREACH}[5]{{\sf foreach}({#1},{#2},{#3},{#4}):{#5}}
\newcommand{\ordr}{\preceq}
\newcommand{\SEQ}[2]{{#1};{#2}}
\newcommand{\OUT}[2]{{\sf out} (#1)@{#2}}
\newcommand{\IN}[2]{{\sf in} (#1)@{#2}}
\newcommand{\READ}[2]{{\sf read} (#1)@{#2}}
\newcommand{\EVAL}[2]{{\sf eval} (#1)@{#2}}
\newcommand{\INS}[3]{{\sf insert}({#2},{#1}@{#3})}
\newcommand{\DEL}[4]{{\sf delete}({#1}@{#4},{#2},{#3})}
\newcommand{\EVALP}[2]{{\sf eval}({#1})@{#2}}

\newcommand{\CREATENEW}[2]{{\sf create}({#1},{#2})}
\newcommand{\DROP}[2]{{\sf drop}({#1}@{#2})}

\newcommand{\UPDATE}[5]{{\sf update}({#1}@{#5},{#2},{#3},{#4})}

\newcommand{\AGGR}[6]{{\sf aggr}({#1}@{#6},{#2},{#3},{#4},{#5})}
\newcommand{\PRED}{\psi}
\newcommand{\TRUE}{{\sf true}}

\newcommand{\PARNET}[2]{{#1} || {#2}}
\newcommand{\RESNET}[2]{(\nu {#1}){#2}}
\newcommand{\LOCATED}[3]{{#1}\!::_{#2}\!{#3}}

\newcommand{\TBV}{\mi{tbv}}
\newcommand{\CALTB}{\mi{TB}}
\newcommand{\TBID}{\mi{tid}}
\newcommand{\TBSK}{\mi{sk}}

\newcommand{\LOCS}{\mathcal{L}}
\newcommand{\VLOCS}{\mathcal{U}}

\newcommand{\BV}[1]{\mi{bv}({#1})}
\newcommand{\FV}[1]{\mi{fv}({#1})}
\newcommand{\FL}[1]{\mi{fl}({#1})}

\newcommand{\str}{\mi{str}}
\newcommand{\num}{\mi{num}}

\newcommand{\aop}{\mi{aop}}
\newcommand{\cop}{\mi{cop}}

\newcommand{\strcon}{~\!\cdot~\!}

\newcommand{\ErrNet}{\mi{ERR}}

\newcommand{\SEL}[5]{{\sf select}({#1},{#2},{#3},{#4},{#5})}

\newcommand{\TBDst}{\mi{TB}}

\newcommand{\avar}{\mi{var}}

\newcommand{\bd}{W}

\newcommand{\SQLINSERT}{\mathtt{INSERT}}
\newcommand{\SQLINTO}{\mathtt{INTO}}
\newcommand{\SQLVALUES}{\mathtt{VALUES}}
\newcommand{\SQLSELECT}{\mathtt{SELECT}}
\newcommand{\SQLFROM}{\mathtt{FROM}}
\newcommand{\SQLAS}{\mathtt{AS}}
\newcommand{\SQLWHERE}{\mathtt{WHERE}}
\newcommand{\SQLUPDATE}{\mathtt{UPDATE}}
\newcommand{\SQLSET}{\mathtt{SET}}
\newcommand{\SQLCREATE}{\mathtt{CREATE}}
\newcommand{\SQLTABLE}{\mathtt{TABLE}}
\newcommand{\SQLDROP}{\mathtt{DROP}}
\newcommand{\SQLDELETE}{\mathtt{DELETE}}

\newcommand{\val}{\mi{val}}

\newcommand{\flattensk}{\mi{flatten}_{\mrm{s}}}
\newcommand{\flattendt}{\mi{flatten}_{\mrm{d}}}
\newcommand{\prodSK}{\otimes_{\mrm{sk}}}
\newcommand{\prodR}{\otimes_{\mrm{R}}}

\newcommand{\Err}{\mi{err}}
\newcommand{\EVALT}[1]{\mathcal{E}\llbracket {#1} \rrbracket}
\newcommand{\EVALPred}[1]{\mathcal{E}\llbracket {#1} \rrbracket}

\newcommand{\sattsk}{\Vdash}
\newcommand{\nsattsk}{\Vdash\!\!\!\!\!{/}~\!}

\newcommand{\REDINS}{(\mrm{INS})}
\newcommand{\REDDEL}{(\mrm{DEL})}
\newcommand{\REDSEL}{(\mrm{SEL})}

\newcommand{\REDUPD}{(\mrm{UPD})}
\newcommand{\REDAGGR}{(\mrm{AGR})}
\newcommand{\REDCREATE}{(\mrm{CRT})}
\newcommand{\REDDROP}{(\mrm{DRP})}
\newcommand{\REDPAR}{(\mrm{PAR})}
\newcommand{\REDRES}{(\mrm{RES})}
\newcommand{\REDEQUIV}{(\mrm{EQUIV})}

\newcommand{\REDDELBND}{(\mrm{DEL})}
\newcommand{\REDEVALP}{(\mrm{EVL})}

\newcommand{\SUM}{\mi{sum}}
\newcommand{\AVG}{\mi{avg}}

\newcommand{\REDFORTT}{(\mrm{FOR^{tt}})}
\newcommand{\REDFORFF}{(\mrm{FOR^{ff}})}

\newcommand{\REDSEQTT}{(\mrm{SEQ^{tt}})}
\newcommand{\REDSEQFF}{(\mrm{SEQ^{ff}})}

\newcommand{\LIDN}[1]{\mi{Lid}(#1)}
\newcommand{\LIDC}[2]{\mi{Lid}(#1,#2)}

\newcommand{\TRANS}[2]{{#1}\rightarrow{#2}}
\newcommand{\NOREP}[1]{\mi{no\_rep}(#1)}

\newcommand{\TT}{\mi{tt}}
\newcommand{\FF}{\mi{ff}}

\newcommand{\EVALPredb}[3]{{\cal E}_{#1}\llbracket {#2} \rrbracket_{#3}} 
\newcommand{\MATCHb}[3]{\mi{match}_{#1}({#2},{#3})} 

\newcommand{\minimal}{\mi{Minimal}}

\newcommand{\envnet}{N_{\mrm{env}}}

\newcommand{\oknet}{\mi{ok}}

\newcommand{\dom}[1]{\mathbf{D}_{#1}}
\newcommand{\itfc}{\nabla}

\newcommand{\ext}[2]{{#1},{#2}}
\newcommand{\emptyTPEnv}{[]}

\newcommand{\jdgtE}[3]{{#1}\vdash {#2} {~\!\smalltriangleright~\!~\!\!} {#3}}
\newcommand{\jdgtpred}[2]{{#1}\vdash {#2} {~\!\smalltriangleright~\!~\!\!} {\booltp}}
\newcommand{\jdgtt}[3]{{#1}\vdash {#2} {~\!\smalltriangleright~\!~\!\!} {#3}}
\newcommand{\jdgtT}[3]{{#1}\vdash {#2} {~\!\smalltriangleright~\!~\!\!} {#3}}
\newcommand{\jdgta}[3]{{#1}\vdash {#2} {~\!\smalltriangleright~\!~\!\!} {#3}}
\newcommand{\jdgtP}[2]{{#1}\vdash {#2}}
\newcommand{\jdgtC}[2]{{#1}\vdash {#2}}
\newcommand{\jdgtN}[2]{{#1}\vdash {#2}}

\newcommand{\skproj}[3]{{#1\!\!}\downarrow^{#2}_{#3}}

\newcommand{\type}{\tau}
\newcommand{\rcdcon}{~\!,~\!}

\newcommand{\msett}[1]{{#1}~\mi{mset}}

\newcommand{\inttp}{\mi{Int}}
\newcommand{\strtp}{\mi{String}}
\newcommand{\loctp}{\mi{Loc}} 
\newcommand{\idtp}{\mi{Id}}
\newcommand{\booltp}{\mi{Bool}}

\newcommand{\sktp}{\tau}
\newcommand{\dttp}{\tau_{\mrm{d}}}
\newcommand{\msettp}{\tau_{\mrm{m}}}
\newcommand{\prodtp}{\tau_{\mrm{p}}}

\newcommand{\size}[1]{\mi{size}({#1})}
\newcommand{\sz}[1]{\mi{sz}\_\mi{asc}({#1})}

\newcommand{\restbid}{{\sf SSResult}}

\newcommand{\dcnet}{N_{\mrm{DC}}}
\newcommand{\dcnab}{\nabla_{\mrm{DC}}}

\newcommand{\CASE}{{\bf{Case}~}}
\newcommand{\IH}{{induction hypothesis }}
\newcommand{\STEP}[3]{{#1}\vdash {#2} \rightarrow {#3}}

\begin{document}

\title[A Coordination Language for Databases]{A Coordination Language for Databases}

\author[X.~Li]{Ximeng Li\rsuper a}
\address{{\lsuper{a,c,d,e}}Technical University of Denmark}
\email{\{ximl, albl, fnie, hrni\}@dtu.dk}

\author[X.~Wu]{Xi Wu\rsuper b}
\address{{\lsuper{b}}The University of Queensland\\
East China Normal University}
\email{xi.wu@uq.edu.au}

\author[A.~L.~Lafuente]{Alberto Lluch Lafuente\rsuper c}
\address{\vspace{-18 pt}}

\author[F.~Nielson]{Flemming Nielson\rsuper d}
\address{\vspace{-18 pt}}

\author[H,~R.~Nielson]{Hanne Riis Nielson\rsuper e}
\address{\vspace{-18 pt}}

% \titlecomment{OPTIONAL comment concerning the title, \eg, if a variant
% or an extended abstract of the paper has appeared elsewehere}
%%%%%%%%%%%%%%%%%%%%%%%%%%%%%%%%%%%%%%%%%%%%%%%%%%%%%%%%%%%%%%%%%%%%%%%%%%

%% the abstract has to PRECEED the command \maketitle:
%% be sure not to issue the \maketitle command twice!
\begin{abstract}
We present a coordination language for the modeling of distributed database applications.
The language, baptized Klaim-DB, borrows the concepts of localities and nets of the coordination language Klaim
	but re-incarnates the tuple spaces of Klaim as databases.
It provides high-level abstractions and primitives for the access and manipulation of structured data,
	with integrity and atomicity considerations.
We present the formal semantics of Klaim-DB and develop a type system that avoids potential runtime errors
	such as certain evaluation errors and mismatches of data format in tables, which are monitored in the semantics.
The use of the language is illustrated in a scenario where the sales from different branches of a chain of department stores
	are aggregated from their local databases.
Raising the abstraction level and encapsulating integrity checks in the language primitives
	have benefited the modeling task considerably.
\end{abstract}
\maketitle

%!TEX root = ./paper.tex

\section*{Introduction}

Today's data-intensive applications are becoming increasingly distributed.
Multi-national collaborations on science, economics, military etc.,
  require the communication and aggregation of data extracted from databases that are geographically dispersed.
Distributed applications including websites frequently adopt the Model-View-Controller (MVC) design pattern
  in which the ``model'' layer is a database.
Fault tolerance and recovery in databases also favors the employment of distribution and replication.
The programmers of distributed database applications are faced with not only the challenge of writing good queries,
  but also that of dealing with the coordination of widely distributed databases.
It is commonly accepted in the formal methods community that the \emph{modeling} of complex systems in design
  can reduce implementation errors considerably~\cite{WoodcockLBF09_FM,Abrial07_FMTheoryPractice,Newcombe14_Amazon}.

Klaim~\cite{Nicola98_Klaim} is a kernel language facilitating the specification of distributed and coordinated processes.
In Klaim, processes and information repositories exist at different localities.
The information repositories are in the form of tuple spaces, that can hold data and code.
The processes can read tuples from (resp. write tuples to) local or remote tuple spaces,
  or spawn other processes to be executed at certain localities.
Many distributed programming paradigms can be modeled in Klaim~\cite{Nicola98_Klaim}.

Klaim thus provides an ideal ground for the modeling of networked applications in general.
However, the unstructured nature of tuple spaces and fine-grained operations mostly targeting individual tuples create difficulties
	in the description of the data-manipulation tasks that are usually performed using a high-level language such as SQL.
A non-exhaustive list of the disadvantages of staying at the original abstraction level of Klaim are given below.
\begin{itemize}
  \item A considerable amount of meta-data needed by databases has to be maintained as
        particular tuples or components of tuples.
  \item The sanity checks associated with database operations have to be borne in mind and
        performed manually by the programmer.
  \item Atomicity guarantees are difficult to deliver when batch operations are performed.
\end{itemize}

To support the modeling of applications operating on distributed, structured data,
	we propose the coordination language Klaim-DB.
The language is inspired by Klaim in which it allows the distribution of \emph{nodes}, and remote operations on data.
Its succinct syntax facilitates the precise formulation of a structural operational semantics,
	giving rigor to high-level formal specification and reasoning of distributed database applications.
The language also provides structured data organized as databases and tables,
	and high-level actions that accomplish the data-definition and data-manipulation tasks ubiquitous in these applications.

Since it is desirable to be able to predict and rule out certain classes of runtime errors statically,
we develop a type system that scrutinizes specifications for
	potential evaluation errors and format mismatches in manipulating structured data.
These errors are monitored in the semantics, and the safety of the type system (in the sense of subject reduction~\cite{Pierce02_Types}) and
	its soundness with respect to the monitoring are proved.
The efficiency of type checking is also evidenced by a formal result.

We will use database operations involved in the management of a large-scale chain of department stores as our running example.
Each individual store in the chain has its local database
  containing information about the current stock and sales of each kind of product.
The semantic rules for the core database operations will be illustrated by the local query and maintenance of these individual databases,
  and our final case study will be concerned with data aggregation across multiple local databases
  needed to generate statistics on the overall product sales.

Compared to the conference version~\cite{WuLLNN15_KlaimDB}, our new contribution is twofold.
\begin{enumerate}
  \item The language design is made more succinct and uniform, and
	is augmented with two important features: code mobility and join operations on tables
	(which were only discussed as a potential extension previously).
  \item The type system mentioned above and its theoretical properties were not developed in \cite{WuLLNN15_KlaimDB}.
  	To clarify the guarantees of well-typedness,
  		we have also adapted the semantics of \cite{WuLLNN15_KlaimDB} to signal run-time errors explicitly,
  		rather than get stuck, in case certain classes of abnormalities arise.
\end{enumerate}

\noindent This paper is structured as follows.
In Section~\ref{sec:background},
we briefly review the common database operations and the Klaim language,
and describe the design of Klaim-DB.
In Section~\ref{sec:syntax}, the syntax of Klaim-DB is presented,
which is followed by the structural operational semantics specified in Section~\ref{sec:semantics}.
In Section~\ref{sec:typing}, we present the type system and establish its theoretical properties.
Our case study is then presented in Section~\ref{sec:case_study}.
We discuss potential alternatives at several points in our development and related work in Section~\ref{sec:related}, and
conclude in Section~\ref{sec:conclusion}.
Multiset notations are explained in Appendix~\ref{app:multiset} and
the proofs of theoretical results can be found in Appendix~\ref{app:proofs}.

%%% Local Variables:
%%% mode: latex
%%% TeX-master: "paper"
%%% End:

%!TEX root = ./paper.tex

\section{Background}\label{sec:background}

In this section we give a brief overview of common database operations and of the coordination language Klaim.
On this basis, we introduce our ideas about the design of Klaim-DB.

Our discussion hereafter will be related to the following scenario.
Consider the management of a large-scale chain of department stores,
in which the head office can manage the sales of different imaginary brands (e.g., KLD, SH,...) in its individual branches,
as shown in Figure~\ref{case_fig}.
The localities $\ttloc_0$, ..., $\ttloc_n$ represent the distinct places where the database systems of the
head office and the branches are located.

\begin{figure}[!htb]
\caption{Running Example}\label{case_fig}
\begin{center}
  \includegraphics[
%  width=0.7\textwidth,
scale=0.4
  % natwidth=635,natheight=317
  ]{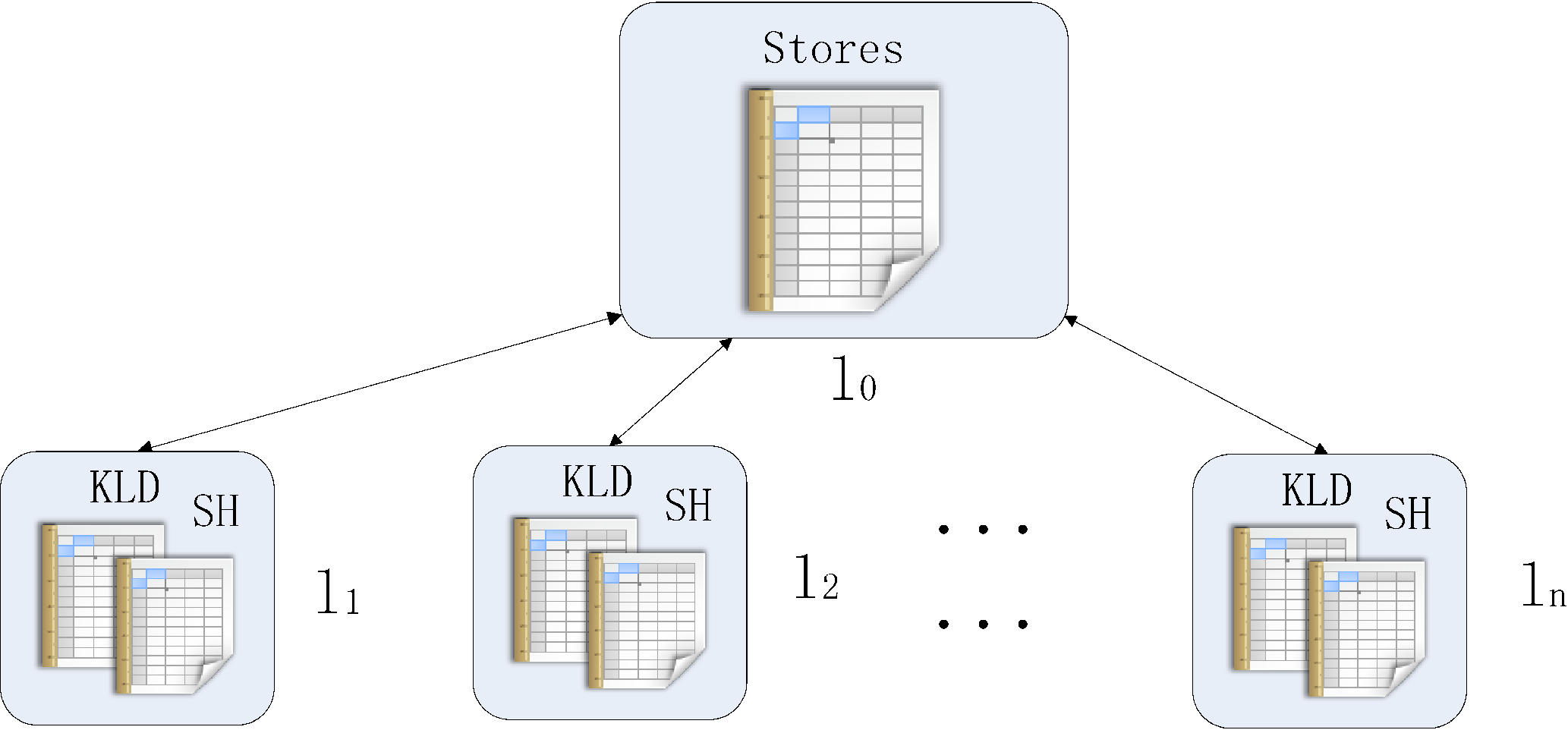}\\
\end{center}
\end{figure}

Suppose a table {\sf Stores} exists in the database of the head office (Figure~\ref{fig:Stores_table}).
Each row is a record of the information of one branch:
the city of the branch, its address and name, the brands sold in the branch and
the locality corresponding to the branch.

\begin{figure}[ht!]
  \caption{The Table {\sf Stores}}\label{fig:Stores_table}
  \centering
\begin{tabular}{|c|c|c|c|c|}
  \hline
  $\mi{City}$ & $\mi{Address}$ & $\mi{Shop\_Name}$ & $\mi{Brands}$ & $\mi{Locality}$\\
  \hline\hline
  CPH & ~ABC DEF 2, 1050~  & ~Shop1~ & ~$\{{\sf KLD},{\sf SH},...\}$~ & $\ttloc_1$ \\
  \hline
  CPH & ~DEF HIJ 13, 2800~ & ~Shop2~ & ~$\{{\sf KLD},{\sf SH},...\}$~ & $\ttloc_2$\\
  \hline
  CPH & ~HIJ KLM 26, 1750~ & ~Shop3~ & ~$\{{\sf KLD},{\sf SH},...\}$~ & $\ttloc_3$\\
  \hline
  AAL & ~KLM NOP 3, 3570~ & ~Shop4~ & ~$\{{\sf LAM},{\sf IMK},...\}$~ & $\ttloc_4$\\
  \hline
  AAL & ~NOP QUW 18, 4500~ & ~Shop5~ & ~$\{{\sf LAM},{\sf IMK},...\}$~ & $\ttloc_5$\\
  \hline
  ... & ... & ... & ... & ...\\
\end{tabular}
\end{figure}

Similarly, each database of a branch has several tables identified by brand names.
Each such table records the stock and sales of all types of shoes of its corresponding brand.
The table {\sf KLD} in one of the branches, about the shoe brand KLD, is shown in Figure~\ref{fig:KLD_table},
where ``HB'' stands for ``High Boot'' and ``SB'' for ``Short Boot''.

\begin{figure}[ht!]
  \caption{The Table {\sf KLD} in one branch}\label{fig:KLD_table}\normalsize
  \centering\small
\begin{tabular}{|c|c|c|c|c|c|c|}
  \hline
  $\mi{Shoe\_ID}$ & $\mi{Shoe\_type}$ & $\mi{Year}$ & $\mi{Color}$ & $\mi{Size}$ & $\mi{In}\hbox{-}\mi{stock}$ & $\mi{Sales}$\\
  \hline\hline
  ~001~ & ~HB~  & ~2015~ & ~red~ & ~38~& ~5~ & ~2~ \\
  \hline
  ~001~ & ~HB~  & ~2015~ & ~red~ & ~37~& ~8~ & ~5~ \\
  \hline
  ~001~ & ~HB~  & ~2015~ & ~red~ & ~36~& ~3~ & ~1~ \\
  \hline
  ~001~ & ~HB~  & ~2015~ & ~black~ & ~38~& ~3~ & ~2~ \\
  \hline
  ~001~ & ~HB~  & ~2015~ & ~black~ & ~37~& ~5~ & ~2~ \\
  \hline
  ~002~ & ~SB~  & ~2015~ & ~green~ & ~38~& ~2~ & ~0~ \\
  \hline
  ~002~ & ~SB~  & ~2015~ & ~brown~ & ~37~& ~4~ & ~3~ \\
  \hline
  ... & ... & ... & ... & ...& ... & ...\\
\end{tabular}
\end{figure}

\subsection*{A Brief Overview of Database Operations}

We now sketch the typical database operations involved
in the management of the chain of department stores.

\paragraph{Insertion}
Suppose a new shoe type of the brand KLD is added to stock, or a new branch has been opened up.
Then a new row is supposed to be inserted into either the table {\sf KLD} or the table {\sf Stores}.
\paragraph{Deletion}
When a shoe type has become outdated and removed from stock,
a row corresponding to this shoe type needs to be deleted from the table {\sf KLD}.
\paragraph{Selection}
From time to time, a manager at the head office may want to know the stock and sales figures
of certain types of KLD shoes.
Selection operations are well-suited to the manager's needs.
A concrete example is ``select the color, size, and sales of the high boots that are not red''.
\paragraph{Update}
After each pair of KLD shoes of some type is sold,
the row in table {\sf KLD} corresponding to that shoe type should be updated,
deducing one from the \textit{In-stock} field.
After a new brand is introduced into the inventory at certain branches,
the rows corresponding to these branches need to be updated,
adding the brand name to the \textit{Brand} field.
\paragraph{Aggregation}
Generating sales statistics is an important task, especially when the amount of data gets large.
Tasks such as totaling the number of KLD shoes colored black sold can be performed using
aggregation operations.
In general, the aggregation operation plays a key role in developing business analytics.
\paragraph{Table Creation}
After a new brand is introduced into the inventory at certain branches,
we need not only to update the records for these branches in the table {\sf Stores},
but also to create a new table for the brand at each of these branches.
\paragraph{Table Deletion}
If the company producing a shoe brand has gone out of business,
and all the shoes of that brand has been sold or disposed of at some branch,
then it would be desirable to be able to remove the table corresponding to that brand at that branch.

\subsection*{The Essentials of Klaim}
Before touching upon the design rationale of Klaim-DB,
we discuss the essential elements of the coordination language Klaim,
which is a fundamental inspiration for Klaim-DB.

Klaim~\cite{Nicola98_Klaim}
(\underline{K}ernel \underline{L}anguage for
 \underline{A}gents \underline{I}nteraction and \underline{M}obility)
is a core language derived from the Linda language~\cite{GelernterB82_GlobalBuf}.
Like Linda, Klaim features \emph{generative communication}~\cite{Gelernter85_Generative} between processes:
processes share ``tuple spaces'' and can generate tuples in them
for other processes to retrieve via pattern matching.
This allows a higher level of abstraction and expressiveness compared with
communication via shared variables or message passing.

Klaim supports programming with \emph{explicit use of localities}.
The action $\OUT{t}{l}$ generates the tuple $t$ in the tuple space \emph{located at $l$}.
The action $\IN{T}{l}$ looks for tuples $t$ matching the template $T$ that potentially contains formal fields,
in the tuple space \emph{located at $l$}.
If such tuples exist, the formal fields of $T$ are bound to the actual counterparts in $t$ for later use,
and $t$ is subsequently removed from the tuple space.
The action $\READ{T}{l}$ is a non-destructive version of $\IN{T}{l}$ ---
it does not remove matched tuples.

As an alternative to addressing a remote locality $l$ directly,
a process $P$ can be spawned at $l$ using the action $\EVAL{P}{l}$ to perform the tasks locally ---
this is a form of \emph{code mobility}.

\subsection*{The Design of Klaim-DB}

In Klaim-DB, we conceive tuple spaces as being structured into tables.
Each located tuple space in Klaim is now a located database consisting of a collection of tables.
A table has the structure $(I,R)$,
where $I$ is an \emph{interface} that contains the table identifier and describes the data format of the table,
and $R$ is a multiset containing the data records.

Correspondingly, the basic operations of Klaim-DB are the high-level database operations such as
insertion, selection, aggregation, etc., targeting whole tables,
rather than plain outputs and inputs of tuples.
Each of these database operations is executed in an \emph{atomic} step
and makes sure when necessary that the data involved
conforms to the format specified in the interface of the table targeted.
An example on atomic operations is an update operation
concerned with the $\mi{Brands}$ field of the table {\sf Stores},
that adds the brand {\sf PKL} to the set of brands sold in all shops.
This operation is performed as one holistic step, rather than in a record-by-record fashion.
Hence if it is performed simultaneously with another update
that removes the brand {\sf PKL} from all shops,
we will either find {\sf PKL} among the brands sold in all the shops, or none of them,
after both operations are completed.

In the operations of Klaim-DB, we often use Klaim-style pattern matching, combined with predicates,
to provide considerable expressive power.
Suppose we need to select from the table {\sf KLD} (of Figure~\ref{fig:KLD_table})
the color, size and sales of the types of high boots with ID ``001'', which are not red.
We can achieve this with the {\sf select} action of Klaim-DB,
using the template $(!\mi{id},!\mi{tp},!\mi{yr},!\mi{cr},!\mi{sz},!\mi{is},!\mi{ss})$ to match against
the 7-field rows of the table,
specifying the criterion with the predicate $\mi{id}=``001"\land \mi{tp}=``\mi{HB}"\land \mi{cr}\neq ``red"$,
and signaling that the only fields needed in the result are for ``color'', ``size'' and ``sales'' with
the tuple $(\mi{cr},\mi{sz},\mi{ss})$.

Joining tables is supported in Klaim-DB,
and we can in fact select data from the join of tables located at different places.
In addition, the different kinds of operations can be combined to form more powerful queries.
An example is guiding selections with the statistics generated in aggregation operations ---
consider the selection of shoes whose sales figures are above average,
where the average number is obtained via a preceding aggregation.

Klaim-DB inherits the code mobility primitive from Klaim;
so instead of operating on remote tables directly,
processes can be spawned onto remote localities, and
collect and manipulate data from tables at different places. Code mobility can yield typical parallel processing patterns, for example combining the {\sf eval} primitive with  {\sf foreach} loops to process all tuples in a table in parallel. 

\subsubsection*{Klaim-DB versus Standard Klaim}
One may be prompted to think about whether standard Klaim is already
sufficiently expressive to model the databases and their corresponding
management issues in the aforementioned scenario.
The answer is positive, but without defeating the purpose of Klaim-DB:
\begin{enumerate}
\item By providing data abstractions like tables,
  and language primitives such as selection and aggregation,
  Klaim-DB makes a designer \emph{think} in terms of the management of distributed databases,
  rather than the manipulation of individual tuples.
  For example, data selection is naturally a batch operation in Klaim-DB,
  where the system designer can specify complicated criteria via predicates,
  rather than a replication/recursion in which basic input operations are performed,
  as in Klaim.
\item Related to (1), Klaim-DB enables us to maintain certain meta-data
  in a more natural and efficient way.
  In standard Klaim, the membership of the tuple $(001,\mi{HB}, 2015,$ $\mi{red}, 38, 5, 2)$
  in the table {\sf KLD} can be represented by augmenting the tuple with the component {\sf KLD},
  e.g., as $({\sf KLD}, 001,\mi{HB}, 2015, \mi{red}, 38, 5, 2)$.
  Such a mixture of structural information and content information is not beneficial for
  addressing the meta-data by specified fields of language primitives ---
  instead of adhering to one single convention imposed by a language,
  one has to follow potentially different conventions each specific to a particular system.
  In addition, meta-data such as the table identifier need to be repeated in all tuples,
  reducing space efficiency.
\item Klaim-DB actions are atomic, which avoids race conditions existing in situations like
  updating a table over which a selection is being performed, or
  concurrently updating the same table, and reduces unexpected
  behaviors under concurrency in general.
  The ability to have fine-grained interleaving of tuple-by-tuple operations can be
  recovered using {\sf foreach} loops, \emph{only when efficiency is important}.
\item Access control types are developed in standard Klaim and
  datatypes are used in X-Klaim.
  We need \emph{dedicated} typing disciplines to ease database-oriented thinking and design.
  Supposing an action inserting the tuple $(001, \mi{HB}, 2015)$ into the table {\sf KLD}
  is placed in a piece of specification,
  it would be desirable to raise a typing error.
  In Klaim-DB tables,
  the separation of structural information in interfaces $I$
  and content information in data sets $R$
  facilitates maintaining the types of \emph{all} the rows of each table in its interface,
  thereby facilitating the detection of typing errors including but
  are not limited to the aforementioned kind.
\end{enumerate}

%%% Local Variables:
%%% mode: latex
%%% TeX-master: "paper"
%%% End:

%!TEX root = ./paper.tex

\section{Syntax}\label{sec:syntax}

This section focuses on the syntactic features of Klaim-DB. We start presenting the syntax of Klaim-DB in Section~\ref{sec:syntax-only}. We then present in Section~\ref{sec:syntax-types} a first description of the types that we shall consider in our type system. We also provide some additional remarks on Klaim (Section~\ref{sec:syntax-klaim}) and examples based on our case study  (Section~\ref{sec:syntax-example}).

\subsection{Klaim-DB syntax}\label{sec:syntax-only}
The syntax of Klaim-DB is presented in Figure~\ref{fig:syntax}. In the following we discuss all syntactic ingredients of the language following bottom-up order: from the bottom-most syntactic category $e$ for expressions to the top-most syntactic category $N$ for networks.

\begin{figure}[ht!]
  \centering
  \caption{The Syntax of Klaim-DB}\label{fig:syntax}
  \small
$$
\begin{array}{rcll}
\hline \\[-1.8ex]
N & ::= & \NILN \mid \ErrNet \mid \PARNET{N_1}{N_2} \mid \RESNET{\llocConst}{N} \mid \LOCATED{\llocConst}{}{C} & (\mi{networks}) \\[1ex]
C & ::= & P \mid (I,R) \mid \PAR{C_1}{C_2} & (\mi{components})\\[1ex]
P & ::= & \NILP \mid \PREF{a}{P} \mid \CALL{A}{\tilde{e}} \mid \FOREACH{\CALTB}{T}{\psi}{\ordr}{P} \mid
          \SEQ{P_{1}}{P_{2}} & (\mi{processes})\\[1ex]
a & ::= & \INS{\TBID}{t}{\lloc} \mid \DEL{\TBID}{T}{\psi}{\ell} \mid \SEL{\LST{\CALTB}}{T}{\psi}{t}{!\TBV} \mid & (\mi{actions})\\[1ex]
  &     & \UPDATE{\TBID}{T}{\PRED}{t}{\ell} \mid \AGGR{\TBID}{T}{\PRED}{f}{T'}{\ell}\mid \\[.5ex]
  &     & \CREATENEW{\TBID@\lloc}{\TBSK} \mid \DROP{\TBID}{\ell} \mid \\[.5ex]
  &     & \EVALP{P}{\ell}\\[1ex]
\CALTB & ::= & \TBID@\lloc \mid \TBV \mid (I,R) & (\mi{tables}) \\[1ex]
% \CALTB & ::= & \TBV \mid (I,R) \\[1ex]
I & ::= & (\TBID, \TBSK) & (\mi{interfaces})\\[1ex]
T & ::= & \bd_1,...,\bd_n \quad (n>0) & (\mi{templates}) \\[1ex]
\bd & ::= & !x \mid{} !u & (\mi{fields}) \\[1ex]
t & ::= & e_1,...,e_n \quad (n>0) & (\mi{tuples}) \\[1ex]
\psi & ::= & \TRUE \mid e_1~\!\cop~\!e_2 \mid e_1\in e_2 \mid
             \lnot \psi_1 \mid \psi_1\land \psi_2 & (\mi{predicates}) \\[1ex]
  e & ::= & \num \mid \str \mid \TBID \mid \llocConst \mid x \mid u \mid e_1 \strcon e_2 \mid
  e_1 ~\!\aop~\! e_2 \mid \{e_1,...,e_n\}\quad(n\ge 0) & (\mi{expressions}) \\\\[-2ex]
\hline
\end{array}
$$\normalsize
\end{figure}

% \mid \NEWTAB{u}{t} \mid \DROP{u} temporarily omitted

\paragraph{\bf Expressions.}
We assume a set $\LOCS$ of localities, and a set $\VLOCS$ of locality variables.
An \emph{expression} $e$ can be an integer $\num$, a string $\str$, a locality $\llocConst\in\LOCS$,
a table identifier $\TBID$, a variable for data ($x$) or for locality ($u\in\VLOCS$),
the concatenation $e_1\strcon e_2$ of (string) expressions $e_1$ and $e_2$,
an arithmetic operation $\aop$ applied on sub-expressions $e_1$ and $e_2$, or
a multiset $\{e_1,...,e_n\}$ of expressions $e_1$, ..., $e_n$ which are themselves multiset-free.
We will follow \cite{Nicola98_Klaim} to write $\ell$ for either a locality $l$ or a locality variable $u$.

\paragraph{\bf Predicates.}
A \emph{predicate} $\psi$ can be $\TRUE$ for logical truth,
a comparison $e_1~\!\cop\!~e_2$ of two expressions $e_1$ and $e_2$,
a membership test of expression $e_1$ in $e_2$,
the negation of some predicate $\psi_1$, or
a conjunction of two predicates $\psi_1$ and $\psi_2$.
For a comparison $e_1~\! \cop~\! e_2$,
$\cop$ ranges over equality between integers, strings, or localities,
numerical ordering of integers, and alphabetical ordering of strings.
For the membership test $e_1\in e_2$ to be carried out, $e_2$ needs to evaluate to a multiset $M$.

\paragraph{\bf Tuples and Templates.}
We distinguish between \emph{tuples} $t$ and \emph{templates} $T$.
The components of tuples are arbitrary expressions.
On the other hand,
a template $T$ contains only formal fields $!x$ and $!u$
where $x$ is a variable that can be bound to data, and $u$ to localities as introduced above.
We assume templates $T$ are linear, i.e., each bound variable can only appear once in $T$.

\paragraph{\bf Interfaces.}
An \emph{interface} $I$ of a table has the structure $(\TBID,\TBSK)$ where
$\TBID$ is the table identifier and $\TBSK$ is a schema (or a type)
describing the data format of the table.
We will use the syntax $I.\TBID$ and $I.\TBSK$ to refer to the corresponding components in $I$,
i.e., $(\TBID',\TBSK').\TBID=\TBID'$ and $(\TBID',\TBSK').\TBSK=\TBSK'$.

\paragraph{\bf Tables.}
\emph{Tables} of different forms constitute the syntactical category $\CALTB$.
In more detail, a table can be a reference $\TBID@\lloc$, with identifier $\TBID$ and locality $\lloc$,
a table variable $\TBV$, or
a concrete table $(I,R)$ with interface $I$ and data set $R$.
For a concrete table $(I,R)$, the data set $R$ is a multiset of tuples whose components are constant values.
Each tuple in $R$ corresponds to one row of data in the table,
and is supposed to have $I.\TBSK$ as its type.

\paragraph{\bf Actions.}
There are eight different kinds of \emph{actions}, which are explained in detail below.
To help the readers familiar with the SQL language, we will discuss when put in a SQL-like syntax,
what the DB-oriented actions would be.
These \emph{SQL-like} queries used for the analogy differ from standard SQL queries in two ways:
located table references are used instead of plain table references, and
bound variables of templates are used (in predicates and tuples) instead of field names (or attribute names).
The motivation for this deviation is that a detailed encoding of
Klaim-DB actions into SQL would require, for example, to deal with the
conversion from position-based indexing in Klaim-DB schemas to name-based indexing in SQL.
Since such encoding would hamper readability and
the aim of the SQL analogy is to provide intuition rather than formal results,
we prefer to keep a lighter notation.

\paragraph{\textbf{Insertion:}} The action $\INS{\TBID}{t}{\lloc}$ is used to insert a new row $t$
into a table with identifier $\TBID$ inside the database at $\lloc$.
In SQL-like syntax, the insertion action can be written as
$\SQLINSERT~\SQLINTO~\TBID@\lloc~\SQLVALUES~t$.
\begin{exa}[Adding New Shoes]\label{ex:act_ins}
  The following action inserts into the table {\sf KLD} at $\ttloc_1$
  an entry for KLD high boots of a new color, white, sized ``37'', produced in 2015, with 6 in stock:
  \small$$\INS{{\sf KLD}}{(``001", ``\mi{HB}", ``2015", ``white", ``37", 6, 0)}{\ttloc_1}.$$\normalsize
\end{exa}
\paragraph{\textbf{Deletion}:}
The action $\DEL{\TBID}{T}{\psi}{\lloc}$ deletes all rows matching the template $T$,
resulting in bindings that satisfy the predicate $\psi$,
from the table referenced by $\TBID$ in the database located at $\lloc$.
This action corresponds to the SQL-like query $\SQLDELETE~\SQLFROM~\TBID@\lloc~\SQLAS~T~\SQLWHERE~\PRED$.\footnote{Readers familiar with SQL will notice our abuse of notation using SQL's $\SQLAS$ clause in combination with a template $T$. As we explained above, we prefer to keep a lighter notation and resort to the reader's intuition. This remark applies to all subsequent references to SQL-like operations, where we perform similar abuses of notation.}
\begin{exa}[Deleting Existing Shoes]\label{ex:act_del}
  The following action deletes all entries for white high boots of the brand KLD, sized ``37'',
  from the table ${\sf KLD}$ at $\ttloc_1$:
  \small
  $$
  \begin{array}{l}
  {\sf delete}({\sf KLD}@\ttloc_1,(!\mi{id},!\mi{tp},!\mi{yr},!\mi{cr},!\mi{sz},!\mi{is},!\mi{ss}),
    \mi{tp}=``\mi{HB}"\land\mi{cr}=``\mi{white}"\land\mi{sz}=``37").
  \end{array}
  $$
  \normalsize
\end{exa}

\paragraph{\textbf{Selection:}}
The action $\SEL{\LST{\CALTB}}{T}{\psi}{t}{!\TBV}$ picks from the Cartesian product of the data sets
of tables identified by the list elements of $\LST{\CALTB}$ all rows that match the template $T$, and
resulting in bindings that satisfy the predicate $\psi$.
The instantiations of $t$ with these bindings form a new table that is bound to $\TBV$.
Forming Cartesian products and specifying constraints with $\psi$ capture certain join operations of
the Relational Algebra~\cite{Codd70_Relational} and SQL.
The selection action corresponds to the SQL-like query
$\SQLSELECT~t~\SQLINTO~!\TBV~\SQLFROM~\CALTB_1,...,\CALTB_n~\SQLAS~T$ $\SQLWHERE~\PRED$.
\begin{exa}[Selection of Shoes in a Certain Color]\label{ex:act_sel}
  The following action selects the color, size, and sales of
  the types of high boots of the brand KLD that are not red,
  from the database at $\ttloc_1$, with the resulting table substituted into $\TBV$:
  \small
  $$
  \begin{array}{l}
    {\sf select}({\sf KLD}@~\!\ttloc_1,(!\mi{id},!\mi{tp},!\mi{yr},!\mi{cr},!\mi{sz},!\mi{is},!\mi{ss}), \\
    \qquad\qquad\qquad\!\! \mi{id}=``001"\land\mi{tp}=``\mi{HB}"\land\mi{cr}\neq ``red",
    (\mi{cr},\mi{sz},\mi{ss}), !\TBV).
  \end{array}
  $$
  \normalsize
\end{exa}

\paragraph{\textbf{Update:}}
The action $\UPDATE{\TBID}{T}{\PRED}{t}{\lloc}$ replaces each row matching $T$ and satisfying $\psi$
in table $\TBID$ (at $\lloc$) with a new row $t$,
while leaving the rest of the rows unchanged.
This action corresponds to the SQL-like query $\SQLUPDATE~\TBID@\lloc~\SQLAS~T~\SQLSET~t~\SQLWHERE~\PRED$.
\begin{exa}[Update of Shoes Information]\label{ex:act_upd}
  Suppose two more red KLD high boots sized $37$ are sold.
  The following action informs the database at $\ttloc_1$ of the corresponding changes
  in the stock and sales figures.
  \small$$
  \begin{array}{r@{~\!}l}
  {\sf update}({\sf KLD}@{\ttloc_1},
  & (!\mi{id},!\mi{tp},!\mi{yr},!\mi{cr},!\mi{sz},!\mi{is},!\mi{ss}),
    \mi{tp}=``\mi{HB}"\land \mi{cr}=``red"\land \mi{sz}=``37",\\
  & (\mi{id},\mi{tp},\mi{yr},\mi{cr},\mi{sz}, \mi{is} - 2, \mi{ss} + 2)).
  \end{array}
  $$\normalsize
\end{exa}

\paragraph{\textbf{Aggregation:}}
The action $\AGGR{\TBID}{T}{\psi}{f}{T'}{\lloc}$ applies the aggregator function $f$
on the multiset of all rows matching $T$ and satisfying $\psi$ in table $\TBID$ (at $\lloc$) and
binds the aggregation result to the pattern $T'$.
In SQL, aggregation is achieved via selection, and our aggregation action can be compared to
the SQL-like query $\SQLSELECT~f(t)~\SQLINTO~T'~\SQLFROM~\TBID@\lloc~\SQLAS~T~\SQLWHERE~\PRED$.

\begin{exa}[Aggregation of Sales Figures]\label{ex:act_aggr}
The following action produces the total sales (substituted into $\mi{res}$)
of shoes of the brand KLD with ID ``001'' at $\ttloc_1$:
\small
$$
{\sf aggr}({\sf KLD}@{\ttloc_1},(!\mi{id},!\mi{tp},!\mi{yr},!\mi{cr},!\mi{sz},!\mi{is},!\mi{ss}),
\mi{id}=``001", \SUM_7, !\mi{res}),
$$
\normalsize
where $\SUM_7=\lambda R.(\SUM(\{v_7|(v_1,...,v_7)\in R\}))$,
i.e., $\SUM_7$ is a function from multisets $R$
to unary tuples containing the summation results of the 7-th components of the tuples in $R$.
\end{exa}

\paragraph{\textbf{Table Creation and Removal}:}
The action $\CREATENEW{\TBID@\lloc}{\TBSK}$ creates a table with interface $(\TBID,\TBSK)$
at locality $\lloc$.
The SQL-like version is $\SQLCREATE~\SQLTABLE~\TBID@\lloc~(\TBSK)$.
The action $\DROP{\TBID}{\lloc}$ drops the table with identifier $\TBID$ at locality $\lloc$.
In a SQL-like syntax this could be written $\SQLDROP~\SQLTABLE~\TBID@\lloc$.

\begin{exa}[Creating Table for Turnovers]\label{ex:act_crt}
The following action creates a table with identifier {\sf Turnover} at the head office,
to record the yearly turnovers of the whole chain of department stores:
\small
$$
\CREATENEW{{\sf Turnover}@\ttloc_0}{\strtp\times \inttp}.
$$
\normalsize
The table only has two columns ---
one for the year, represented by strings, and the other for the amount, represented by integers.
\end{exa}

\paragraph{\textbf{Code Mobility}:}
As in Klaim, the action $\EVAL{P}{\lloc}$ spawns the process $P$ at the locality $\lloc$.
For our chain of department stores,
it can be useful to create a mobile process from the system of the head office
and to make it migrate among certain relevant branches to accomplish tasks including but
are not limited to the collection of data.
In general, the use of $\EVAL{P}{\lloc}$ enhances the ability to
parallelize and distribute the operations on databases, in particular when combined with the construct to perform iterations, as we shall see later on. \\[-2ex]

We will use $\FV{-}$ and $\BV{-}$, respectively,
to refer to the set of free variables and the set of bound variables of their arguments,
which can be expressions, predicates, tables, tuples, templates or processes.
We will also use $\FL{N}$ to refer to the set of free localities of net $N$.
We globally assume that bound variables and bound localities all have distinct names. \\[-2ex]

\paragraph{\bf Processes.}
A \emph{process} $P$ can be
an inert process $\NILP$,
an action-prefixed process $\PREF{a}{P}$,
a parameterized procedure invocation $\CALL{A}{\tilde{e}}$,
a looping process $\FOREACH{\CALTB}{T}{\psi}{\ordr}{P}$,
or a sequential composition $\SEQ{P_{1}}{P_{2}}$.
Looping is introduced in addition to recursion via process invocation,
to ease the task of traversing tables or data selection results in a natural way.
In more detail, the process $\FOREACH{\CALTB}{T}{\psi}{\ordr}{P}$ traverses the rows of $\CALTB$
that match the template $T$,
resulting in bindings that satisfy the predicate $\psi$,
and execute the process $P$ once for each such binding produced.
The parameter $\ordr$ is a partial order on tuples:
for $t_1$ and $t_2$ as rows of $\CALTB$ such that $t_1\ordr t_2$ and $t_1\neq t_2$,
$t_1$ is processed before $t_2$ is.
Sequential composition is allowed in addition to action-prefixing,
to support imperative-programming-style thinking.
We have chosen that prefixing binds more tightly than sequential composition,
i.e., $\PREF{a}{P_1};{P_2}$ represents $(\PREF{a}{P_1});{P_2}$, and that
in $P_1;P_2$, the scopes of the variables declared in $P_1$ are local to $P_1$.
Note that the explicit bracketing $\PREF{a}{(P_1;P_2)}$ allows variables bound in $a$
to be used in both processes $P_1$ and $P_2$,
unlike what is the case with $(\PREF{a}{P_1});{P_2}$.
Each procedure $A$ that is invoked with the list $\LST{e}$ of arguments
needs to be defined with $A(\avar_1:\tau_1,...,\avar_n:\tau_n)\triangleq P$,
where $\avar_1:\tau_1,...,\avar_n:\tau_n$ is a list of formal parameters $x$ or $u$,
associated with their types
(the types admitted in our language will be explained shortly),
and $P$ is the process acting as the body of the procedure.

\paragraph{\bf Components.}
A \emph{component} $C$ can be a process $P$,
a table of the form $(I,R)$,
or a parallel composition $\PAR{C_1}{C_2}$ of two components.

\paragraph{\bf Networks.}
A \emph{net} $N$ models a distributed database system
that may contain several data\-bases situated at different localities.
An empty net is represented by $\NILN$.
The construct $\llocConst::C$ represents a node of the net,
which captures the ensemble $C$ of processes and tables at $\llocConst$.
All the tables at a locality constitute the database at that locality.
The special net $\ErrNet$ signals that an error has happened in the execution.
The parallel composition of different nodes is represented using the $||$ operator.
With the restriction operator $\RESNET{\llocConst}{N}$, the scope of $\llocConst$ is restricted to $N$.

\paragraph{\bf Systems.}
% To explicate the definitions of procedures,
The specification of a Klaim-DB system takes the following form:
\small
$$
\begin{array}{l}
  {\sf let}~A_1(\avar_{11}:\tau_{11},...,\avar_{1k_1}:\tau_{1k_1})\triangleq P_1 \\
  \qquad\qquad ... \\
  \qquad\!\!\!\! A_n(\avar_{n1}\!:\tau_{n1},...,\avar_{nk_n}\!:\!\tau_{nk_n})\triangleq P_n \\
  {\sf in}~N%_\star
\end{array}
$$
\normalsize

\noindent
where each procedure used in $N$ has its own definition. We focus on \emph{closed} systems: i.e. systems where
all variables to be used in any execution must have been previously bound, and
there are no unknown components in the systems.

\subsection{A taste of types}\label{sec:syntax-types}
We introduce the types that our language admits ---
they have already appeared together with the formal parameters of procedures,
and will play a fundamental role in our type system to be developed in Section~\ref{sec:typing}.
The following types are allowed:
\small
$$
\begin{array}{r@{~\!~\!}l}
\type_{\mrm{d}} ::= & \inttp \mid \strtp \mid \loctp \mid \idtp \\[1ex]
\type_{\mrm{m}} ::= & \msett{\tau_{\mrm{d}}} \mid \type_{\mrm{d}} \\[1ex]
\type_{\mrm{p}} ::= & \tau^1_{\mrm{m}} \times ... \times \tau^n_{\mrm{m}} \\[1ex]
\type ::= & \type_{\mrm{p}} \mid \booltp \mid \msett{\tau_{\mrm{p}}}\rightarrow \type_{\mrm{p}}.
\end{array}
$$
\normalsize

A datatype $\dttp$ can be $\inttp$ for integers, $\strtp$ for strings,
$\idtp$ for table identifiers, and $\loctp$ for localities.
A multiset type $\msettp$ is either of the form $\msett{\dttp}$
for multisets containing elements of type $\dttp$, or just a data type.
A product type $\prodtp$ is of the form $\msettp^1\times...\times\msettp^n$,
and is the type for tuples.
Finally, a type $\tau$ can be a product type $\prodtp$ for strcutured data in tables,
a boolean type $\booltp$ for predicates,
or an arrow type $\msett{\prodtp}\rightarrow\prodtp$ for aggregator functions.
Note that multisets are involved in two different ways.
A component of a tuple can be a multiset,
which is the case for the fourth component of all the tuples in Figure~\ref{fig:Stores_table}.
In this situation, we restrict the types of the elements to be data types $\dttp$.
On the other hand, each data set of a table is a multiset of tuples;
hence the elements have product types $\prodtp$,
as is reflected in the signature of aggregator functions.

Throughout our presentation, a list of objects $o_1,...,o_n$ will be denoted by $\LST{o}$,
  or $\LSTStruct{o(i)}{i\le n}$ when the objects are structured.
We also admit $\LSTStruct{o(i)}{i}$ as a simpler version of the latter when the length is unimportant.
For example, a list of formal parameters and types for a procedure definition can be
$\LSTStruct{\avar_i:\tau_i}{i}$.
A list $\LST{\CALTB}$ of tables is a shorthand for $\LSTStruct{\CALTB_i}{i}$.
For an object $o$, $|o|$ will denote its size ---
  the number of elements of $o$ when $o$ is a list, and
  the number of components of $o$ when $o$ is a tuple or a template, etc.

\subsection{On the syntax of Klaim and Klaim-DB}\label{sec:syntax-klaim}

We comment here on two additional aspects in which Klaim-DB deviates from standard Klaim.

%\begin{rem}
First, in Klaim, there is a two-layer stratification of localities --- logical ones and physical ones.
If a logical locality is thought of as the URL of a site,
then its corresponding physical locality can be thought of as the IP address.
Allocation environments are needed at each physical site to record
which logical localities are locally mapped to which physical ones.
Compared to the conference version (\cite{WuLLNN15_KlaimDB}) of this paper,
we have replaced this two-layer structure with a single notion of ``locality'',
to be able to focus on the exposition of the semantics and typing of
database management constructs and primitives.
Recovering this missing mechanism would not create any technical
difficulties for the development.
%\end{rem}

%\begin{rem}
Second, in \cite{Nicola98_Klaim}, tuples can contain both formal fields and actual ones.
The distinction between templates that can contain formal fields and tuples that cannot,
as made in \cite{NicolaGHNNPP10_FLTP}, \cite{NicolaLPT14_SCEL} and the current paper,
has the effect of explicating certain constraints such as formal fields are not allowed
in the data to be output into a tuple space or inserted into a table,
when a tuple is used instead of a template.
The further constraint used in the current development,
that templates can only contain formal fields, is no real restriction ---
combining templates with \emph{predicates} provides more flexibility and expressiveness
than the matching of tuples/templates in
\cite{Nicola98_Klaim}, \cite{NicolaGHNNPP10_FLTP} and \cite{NicolaLPT14_SCEL}.
%\end{rem}

\subsection{Modeling the Chain of Department Stores} \label{sec:syntax-example}

Going back to the example with a chain of department stores introduced in Section~\ref{sec:background},
we now formally model the network of databases in Klaim-DB.
Recall that the head office maintains a database at locality $\ttloc_0$,
and the branches maintain their own databases at localities $\ttloc_1$ to $\ttloc_n$.

The table {\sf Stores} at the head office (Figure~\ref{fig:Stores_table}) can be represented as $(I_0,R_0)$,
where $I_0.\TBID={\sf Stores}$.
The header describes the purpose of each column of the table and
suggests its appropriate type that shows up in $I_0.\TBSK$, and the subsequent rows
constitute the multiset $R_0$, contain information of the different branches.
We have $I_0.\TBSK=\strtp\times \strtp\times \strtp\times (\msett{\idtp})\times \loctp$, where
each component type describes one column of the table.
Similarly, the table {\sf KLD} at locality $\ttloc_j$ ($1\le j\le n$) can be represented as $(I_j,R_j)$,
where $I_j.\TBID={\sf KLD}$ and
$I_j.\TBSK=\strtp\times \strtp\times \strtp\times \strtp\times \strtp\times \inttp\times \inttp$.
We assume that the data stored in these tables indeed have the types specified in the schemas
(quotation marks around strings are omitted in Figure~\ref{fig:Stores_table} and Figure~\ref{fig:KLD_table}).

To sum up, the databases and operating processes used by the chain of department stores constitute
the following net $\dcnet$:
\small
$$
\ttloc_0::((I_0,R_0) | C'_0) ~||~\ttloc_1::((I_1,R_1) | C'_1) ~||~ ... ~||~
\ttloc_n::((I_n,R_n) | C'_n),
$$
\normalsize
where for $j\in\{1,...,n\}$, $(I_j,R_j)$ describes the local table for the brand KLD inside its database at $\ttloc_j$,
and for each $k\in\{0,...,n\}$, $C'_k$ stands for the remaining
processes and tables at $\ttloc_j$.

\begin{rem}
The reader may have found that there is no formal counterpart of the field names (City, Address, etc.)
in the headers of the tables in Figure~\ref{fig:Stores_table} and Figure~\ref{fig:KLD_table}.
These field names were part of the schemas in the conference version~\cite{WuLLNN15_KlaimDB},
and getting rid of them in the current development is a deliberate choice ---
all the operations on schemas can be realized \emph{positionally},
and the existence of field names in the schemas would be purely symbolic.
\end{rem}

%%% Local Variables:
%%% mode: latex
%%% TeX-master: "paper"
%%% End:

%!TEX root = ./paper.tex

\section{Semantics}\label{sec:semantics}

We devise a structural operational semantics~\cite{PlotkinSOS} for Klaim-DB. 
The semantics is defined with the help of a structural congruence --- 
  the smallest congruence relation satisfying the rules in Figure~\ref{fig:congruence}, 
  where the $\alpha$-equivalence of $N$ and $N'$ is denoted by $N\equiv_\alpha N'$. 

\begin{figure}[ht!]
  \centering
  \caption{The Structural Congruence}\label{fig:congruence}\small
  \begin{tabular}{l l}
    \hline \\[-2ex]
    $\PARNET{N_1}{N_2}\equiv \PARNET{N_2}{N_1}$ &~~ $\RESNET{l_1}{\RESNET{l_2}{N}}\equiv \RESNET{l_2}{\RESNET{l_1}{N}}$\\
    $\PARNET{(\PARNET{N_1}{N_2})}{N_3}\equiv \PARNET{N_1}{(\PARNET{N_2}{N_3})}$ &~~ $\PARNET{N_1}{\RESNET{l}{N_2}}\equiv \RESNET{l}{(\PARNET{N_1}{N_2})}~~(\mrm{if}~l\not\in \FL{N_1})$  \\
    \Comment{$\PARNET{N}{\NIL}\equiv N$ &~~ $\LOCATED{l}{}{C}\equiv \LOCATED{l}{}{\PAR{C}{\NIL}}$\\}
    $N\equiv N'~~(\mrm{if}~N\equiv_{\alpha}N')$ &~~ $\LOCATED{l}{}{(\PAR{C_1}{C_2})}\equiv \LOCATED{l}{}{C_1}~||~\LOCATED{l}{}{C_2}$ \\
    $N||~\!\NILN\equiv N$ & ~~ $\llocConst::(P|\NILP) \equiv \llocConst:: P$ \\[.5ex]
    \hline \\ 
  \end{tabular}\vspace{-3ex}\normalsize
\end{figure}

We start by covering some preliminaries.
The evaluation of expressions, predicates and tuples using the function $\EVALT{-}$, 
  is inductively defined in Figure~\ref{fig:eval_ett}.
The evaluation result of an expression can be an integer, a string, a table identifier, a locality, or $\Err$,
the evaluation result of a predicate is a value in $\{\TT,\FF,\Err\}$, and
the evaluation result of a tuple is an evaluated tuple or $\Err$, 
where $\Err$ signals an evaluation error in all three cases. 
We globally assume that arithmetic operations are ideal; thus there are no overflow errors. 
As a simplification, we also assume that division by zero yields the integer zero.

\begin{figure}[ht!]
\caption{The Evaluation of Expressions, Predicates and Tuples}\label{fig:eval_ett}
\hrule~\\~\\[-2ex]\small
\hspace{-12cm}\framebox[1.2\width]{Expressions}\\[-2ex]
$$
\begin{array}{rl}
  \EVALT{\num}&=\num \\
  \EVALT{\str}&=\str \\
  \EVALT{\TBID} &= \TBID \\
  \EVALT{l}&=l \\
  \EVALT{e_1\strcon e_2}&=
                          \begin{cases}
                            \EVALT{e_1}\strcon \EVALT{e_2} & 
                            \mrm{if}~\EVALT{e_1}\mrm{~and~}\EVALT{e_2}\mrm{~are~strings}\\
                            \Err & \mrm{otherwise}
                          \end{cases} \\\\[-2ex]
  \EVALT{e_1~\!\aop~\! e_2}&=
                             \begin{cases}
                               \EVALT{e_1}~\!\aop~\!\EVALT{e_2} &
                               \mrm{if}~\EVALT{e_1} \mrm{~and~}\EVALT{e_2} \mrm{~are~integers} \\
                               \Err & \mrm{otherwise}
                             \end{cases} \\\\[-2ex]
  \EVALT{\{e_1,...,e_n\}} &=
                            \begin{cases}
                              \{\EVALT{e_1},...,\EVALT{e_n}\} &
                              \mrm{if}~\mrm{for~all~}j, ~\EVALT{e_j}~\mrm{are~}
                              \mrm{values~of~the~same~type~\dttp}
                            \\
                            \Err & \mrm{otherwise}
                            \end{cases}
\end{array}
$$
\hspace{-12cm}\framebox[1.2\width]{Predicates}\\[-2ex]
$$
\begin{array}{rl}
\EVALPred{\TRUE} & =\TT \\[.5ex]
\EVALPred{e_1~\!\cop~\! e_2} & =
  \begin{cases}
  \EVALT{e_1}~\!\cop~\!\EVALT{e_2} & 
    \mrm{if}~\EVALT{e_1} \mrm{~and~} \EVALT{e_2} 
    \mrm{~are~values~of~the~same~type~}\dttp \\
  \Err & \mrm{otherwise}
  \end{cases}
 \\\\[-2ex]
\EVALPred{e_1\in e_2} & = 
  \begin{cases}
  \EVALT{e_1} \in \EVALT{e_2} & 
    \mrm{if}~\EVALT{e_2}~\mrm{is~a~multiset} 
    \mrm{~of~values~of~the~type~of~}\EVALT{e_1} \\
  \Err & \mrm{otherwise}
  \end{cases} 
 \\\\[-2ex]
\EVALPred{\lnot\psi_1} & = 
  \begin{cases}
    \FF & \mrm{if}~\EVALPred{\psi_1}=\TT \\
    \TT & \mrm{if}~\EVALPred{\psi_1}=\FF \\
    \Err & \mrm{if}~\EVALPred{\psi_1}=\Err 
  \end{cases}
 \\\\[-2ex]
\EVALPred{\psi_1\land\psi_2} & = 
  \begin{cases}
    \TT & \mrm{if}~\EVALPred{\psi_1}=\TT \mrm{~and~}\EVALPred{\psi_2}=\TT \\
    \Err & \mrm{if}~\EVALPred{\psi_1}=\Err \mrm{~or~} \EVALPred{\psi_2}=\Err \\
    \FF & \mrm{otherwise}
  \end{cases}
\end{array}
$$
\Comment{
  \\\\[-2ex]
\EVALPredb{b}{\psi} = 
  \begin{cases}
  \EVALPred{\psi} & \mrm{if}~\EVALPred{\psi} \neq \Err \\
  b & \mrm{otherwise}
  \end{cases}
}
\hspace{-12cm}\framebox[1.2\width]{Tuples}\\[-3ex]
$$
\EVALT{e_1\rcdcon ...\rcdcon e_n} =
  \begin{cases}
  \EVALT{e_1} \rcdcon ...\rcdcon \EVALT{e_n} & 
    \mrm{if}~\forall j\in\{1,...,n\}: \EVALT{e_j}\neq\Err \\
  \Err & \mrm{otherwise}
  \end{cases}
$$
\Comment{
\\\\[-1ex]\hspace{-10cm}\framebox[1.2\width]{Templates}\\[-2ex]
$$
\begin{array}{rl}
\EVALT{!x} & =~\! !x \\
\EVALT{!u} & =~\! !u \\
\EVALT{T_1\rcdcon T_2} & = 
  \begin{cases}
  \EVALT{T_1}\rcdcon \EVALT{T_2} & \mrm{if}~
    \EVALT{T_1}\neq\Err\mrm{~and~}\EVALT{T_2} \neq\Err\\
  \Err & \mrm{otherwise}
  \end{cases}
\end{array}\\
$$
}
\smallskip\normalsize
\hrule
\end{figure}

Pattern matching of evaluated tuples $et$ against templates $T$, denoted $et/T$, 
is an operation that results in a substitution or $\Err$. 
The detailed definition is given in Figure~\ref{fig:match_pred}, which contains several rule schemas
  that need to be instantiated by replacing $\val$ with one of $\num$, $\str$ and $\TBID$. 
The replacement of $\val_j$ in the same multiset needs to be performed consistently, 
  e.g., $\{\num_1,\num_2\}$ is allowed, but not $\{\num_1,\str_2\}$. 
In case the format of $et$ does not match against that of $T$, the value $\Err$ is resulted. 
For example, we have 
\small
\begin{align*}
(5,7)/(!x,!y)= &~\! [5 / x, 7 / y]\\
(5,7)/(!x,!y,!z)= &~\! \Err
\end{align*}
\normalsize  
We also write $\psi\sigma$ for the result of applying the substitution $\sigma$ to the formula $\psi$.

\begin{figure}[ht!]
\caption{The Substitution $\subst{et}{T}$}\label{fig:match_pred}
\centering
\hrule
\small
\begin{tabular}{c}{~}\\[-1ex]
$\val/!x=[\val/x]$ \quad
$\{\val_1,...,\val_n\}/!x=[\{\val_1,...,\val_n\}/x]$ \\[1ex]
$\{l_1,...,l_n\}/!x=[\{l_1,...,l_n\}/x]$ \quad
$l/!u=[l/u]$ 
\\[1ex]
$l/!x=\Err$\quad
$\val/!u=\Err$\quad
$\{\val_1,...,\val_n\}/!u=\Err$ 
\\[1ex]
$
\subst{(e_1,...,e_n)}{(\bd_1,...,\bd_m)}=
\begin{cases}
  \sigma_1...\sigma_m & \mrm{if}~m=n\land \forall j:\subst{e_j}{\bd_j}=\sigma_j\land\sigma_j\neq\Err \\
  \Err & \mrm{otherwise}
\end{cases}
$
\Comment{
\\ \\[-2ex]
$\MATCHb{b}{eT}{et}=
  \begin{cases}
    \MATCH{eT}{et} & \mrm{if}~\MATCH{eT}{et}\neq \Err \\
    b & \mrm{otherwise}
  \end{cases}
$
}
\\\\[-1.5ex]
\end{tabular}
\normalsize
\hrule
\end{figure}

The well-sortedness of tuples and templates is specified in Figure~\ref{fig:sat_sk_etpl}.
The judgments $t\sattsk\prodtp$ and $T\sattsk\prodtp$
represent that the tuple $t$ and the template $T$, respectively,
are well-sorted under the sort $\type_{\mrm{p}}$, which is essentially a product type. 
In the case of templates, the different kinds of binders $!x$ and $!u$ are treated separately --- 
an $!x$ cannot have the sort $\loctp$, which is the sort of any $!u$. 

We will use $\uplus$, $\cap$ and $\setminus$ to represent the union, intersection and subtraction  
  of multisets, and the detailed definitions are given in Appendix~\ref{app:multiset}.

\begin{figure}[ht!]
\caption{Well-Sortedness of Evaluated Tuples and Templates}\label{fig:sat_sk_etpl}
\centering
\small
\hrule\smallskip\smallskip
\begin{tabular}{c}
$\num\sattsk \inttp$\quad $\str\sattsk \strtp$\quad $\TBID\sattsk \idtp$\quad 
$\Inference{e_1\sattsk \dttp~...~e_n\sattsk \dttp}{\{e_1,...,e_n\}\sattsk \msett{\dttp}}$ \quad 
$\llocConst\sattsk \loctp$ \quad
$\Inference{e_1\sattsk \sktp_1 ~~...~~ e_n\sattsk \sktp_n}
{e_1\rcdcon... \rcdcon e_n \sattsk \sktp_1\times ...\times \sktp_n}$ \\\\[-2ex]
$!x\sattsk \inttp$\quad $!x\sattsk \strtp$\quad $!x\sattsk \idtp$ \quad $!x\sattsk \msett{\dttp}$\quad
$!u\sattsk \loctp$ \quad 
$\Inference{\bd_1\sattsk \sktp_1~~...~~ \bd_n\sattsk \sktp_n}
{\bd_1\rcdcon...\rcdcon \bd_n \sattsk \sktp_1\times...\times \sktp_n}$
\end{tabular}
\smallskip
\hrule
\normalsize
\end{figure}

We will use $I.\TBSK\!\downarrow\!^T_t$ to represent the projection of the schema $I.\TBSK$ 
according to the template $T$ (matching the format requirements imposed by $I.\TBSK$) and the tuple $t$. 
The projection result is a new schema that describes only the columns referred to by the components of $t$, 
obtainable according to Definition~\ref{def:sk_proj}. 

\begin{defi}[Projection of Schema]\label{def:sk_proj}
Let $W_1,...,W_k$ be a template and $e_1,...,e_n$ be made of a constants or variables bound in some $\bd_j$, and such that $n \leq k$. The projection
$\skproj{(\sktp_1\times...\times\sktp_k)}{\bd_1,...,\bd_k}{~\!e_1,...,e_n}$ is defined as 
\small
$$
  \skproj{(\sktp_1\times...\times\sktp_k)}{\bd_1,...,\bd_k}{~\!e_1,...,e_n}=
    \sktp'_1\times...\times\sktp'_n\\
      \quad \mrm{where}~
        \sktp'_i=
          \begin{cases}
          \sktp & \mrm{where}~e_i\sattsk \sktp \\
          \sktp_j & \mrm{where}~\bd_j=~\! !e_i
          \end{cases}
$$
\normalsize
In all other cases the operation is undefined. 
\end{defi}

The following example demonstrates this projection operation on schemas. 

\begin{exa}\label{ex:projection} 
The table {\sf KLD} of Figure~\ref{fig:KLD_table} has the schema
 $\tau_{\mrm{KLD}}=\strtp \times \strtp \times \strtp \times \strtp \times \strtp \times \inttp \times \inttp$.
Suppose in addition that 
  \small
  $$T=(!\mi{id}, !\mi{tp}, !\mi{yr}, !\mi{cr}, !\mi{sz}, !\mi{is}, !\mi{ss})$$
  \normalsize
  and $t=(\mi{cr},\mi{sz},\mi{ss})$. 
Then $\skproj{\tau_{\mrm{KLD}}}{T}{t}$ gives the types for tuples containing data from three of the columns, 
  for the color, size and sales figures of KLD shoes. 
According to Definition~\ref{def:sk_proj}, we have $\skproj{\tau_{\mrm{KLD}}}{T}{t}=\strtp\times \strtp\times \inttp$. 
\qed
\end{exa}

To join tables, 
we define the operations \small$\prodSK(\LST{\CALTB},\LSTStruct{\llocConst_i::(I_i,R_i)}{i})$ \normalsize and 
\small$\prodR(\LST{\CALTB},\LSTStruct{\llocConst_i::(I_i,R_i)}{i})$ \normalsize that give the product 
of the schemas and data sets, respectively, of the list $\LST{\CALTB}$ of tables given in different forms. 
Elements of $\LST{\CALTB}$ that are table identifiers are regarded as references to tables 
that need to be provided in the list $\LSTStruct{\llocConst_i::(I_i,R_i)}{i}$. 

In Definition~\ref{def:prod_sk}, the operation $\flattensk(\cdot)$ 
eliminates one-level nesting in product types:
\small
$$
\flattensk((\msettp^{11}\times...\times\msettp^{1n_1})\times ...
\times (\msettp^{k1}\times...\times\msettp^{kn_k}))=
\msettp^{11}\times...\times\msettp^{1n_1}\times ...\times \msettp^{k1}\times...\times\msettp^{kn_k}
$$
\normalsize
In Definition~\ref{def:prod_R}, the operation $\flattendt(\cdot)$ takes a multiset of tuples 
as argument and eliminates one-level nesting in all the tuples in this argument set. 
\small
$$
\flattendt(R) = \{ (v_{11},...,v_{1n_1},...,v_{k1},...,v_{kn_k}) \mid 
 ((v_{11},...,v_{1n_1}),...,(v_{k1},...,v_{kn_k}))\in R \}
$$
\normalsize

\begin{defi}[$\prodSK(\LST{\CALTB},\LSTStruct{\llocConst_i::(I_i,R_i)}{i})$]\label{def:prod_sk} 
  Let $\LST{\CALTB}$ be a list of tables and $\LSTStruct{\llocConst_i::(I_i,R_i)}{i}$ be a list of localized tables. 
  \small
  $$
  \begin{array}{l}
    \prodSK(\LST{\CALTB},\LSTStruct{\llocConst_i::(I_i,R_i)}{i}) = \\
    \begin{cases}
      \mrm{undef} & \!\!\!\!
      \begin{array}{r@{~\!}l}
        \mrm{if} & (\exists \TBID,l,j:\CALTB_j=\TBID@l\land 
                    \forall i: l_i\neq l\lor I_i.\TBID\neq \TBID) ~\!\lor \\
                    & \LST{\CALTB}~\mrm{contains~some~}\TBV\mrm{~or~}\TBID@u 
        \end{array} 
      \\\\[-2ex]
      \flattensk(\sktp_1\times...\times \sktp_n) ~\mrm{where}~&  \mrm{otherwise} \\
      ~~n=|\LST{\CALTB}|~\mrm{and} & \\
      ~~\sktp_j=
      \begin{cases}
      I_k.\TBSK & \mrm{if}~\exists k:\CALTB_j=I_k.\TBID @l_k \\
      I_0.\TBSK & \mrm{if}~\exists R_0: \CALTB_j=(I_0,R_0) 
      \end{cases}
      &       
    \end{cases} 
  \end{array}
  $$
  \normalsize
\end{defi}

\begin{defi}[$\prodR(\LST{\CALTB},\LSTStruct{\llocConst_i::(I_i,R_i)}{i})$]\label{def:prod_R}
  Let $\LST{\CALTB}$ be a list of tables and $\LSTStruct{\llocConst_i::(I_i,R_i)}{i}$ be a list of localized tables. 
  \small
  $$
  \begin{array}{l}
    \prodR(\LST{\CALTB},\LSTStruct{\llocConst_i::(I_i,R_i)}{i}) = \\
    \begin{cases}
      \mrm{undef} & \!\!\!\!
      \begin{array}{r@{~\!}l}
         \mrm{if} & (\exists \TBID,l,j:\CALTB_j=\TBID@l\land 
                    \forall i: l_i\neq l\lor I_i.\TBID\neq \TBID) ~\!\lor \\
          & \LST{\CALTB}~\mrm{contains~some~}\TBV\mrm{~or~}\TBID@u 
        \end{array} 
      \\\\[-2ex]
      \flattendt(R_1\times...\times R_n) ~\mrm{where}~ & \mrm{otherwise} \\
      ~~n=|\LST{\CALTB}|~\mrm{and} & \\
      ~~R_j=
      \begin{cases}
      R_k & \mrm{if}~\exists k: \CALTB_j=I_k.id@l_k \\
      R_0 & \mrm{if}~\exists I_0: \CALTB_j=(I_0,R_0) 
      \end{cases}
      &     
    \end{cases} 
  \end{array}
  $$
  \normalsize
\end{defi}

\begin{exa}
Consider joining the table {\sf KLD} at the locality $\ttloc_1$ and 
the table {\sf KLD} at $\ttloc_2$. 
Suppose the first table is given by ${\sf KLD}@\ttloc_1$ and the second by $(I_2,R_2)$. 
We have \small
$\prodSK(\LSTParenL {\sf KLD}@\ttloc_1, (I_2,R_2) \LSTParenR,\LSTParenL \ttloc_1\!::\!(I_1,R_1) \LSTParenR)=
\flattensk(I_1.\TBSK\times I_2.\TBSK)=\strtp\times \strtp\times \strtp\times \strtp\times \strtp\times \inttp
\times \inttp\times \strtp\times \strtp\times \strtp\times \strtp\times \strtp\times \inttp\times \inttp$.  
\normalsize 
Here $I_1.\TBSK$ is obtained for the first table ${\sf KLD}@\ttloc_1$ 
from the singleton list $\LSTParenL \ttloc_1::(I_1,R_1) \LSTParenR$. 
We also have \small
$\prodR(\LSTParenL {\sf KLD}@\ttloc_1, (I_2,R_2) \LSTParenR,\LSTParenL \ttloc_1\!::\!(I_1,R_1) \LSTParenR)=
\flattendt(R_1\times R_2)$. 
We dispense with expanding the latter for its verbosity. 
\normalsize
\qed
\end{exa}

The semantics is designed to monitor evaluation errors and format mismatches 
such as between a tuple and a schema, that could arise at runtime. 
The side conditions of the transition rules are frequently structured as a case analysis --- 
if an evaluation error or a format mismatch occurs, then the net $\ErrNet$ is resulted; 
otherwise a normal transition step will take place. 
In Section~\ref{sec:typing}, we will develop a type system for which well-typedness implies the absence of situations where $\ErrNet$ is produced. 

The semantic judgment is of the form $\envnet\vdash N\rightarrow N'$, 
  representing that the net $N$ can make a transition into the net $N'$, 
  with the net $\envnet$ as the environment. 
We allow the abbreviation of $\NILN\vdash N\rightarrow N'$ as $\vdash N\rightarrow N'$. 
The environment $\envnet$ is consulted when creating a table at some locality $\llocConst$, 
  to check whether a table with the same identifier already exists at $\llocConst$. 

We proceed with a detailed explanation of the semantic rules for Klaim-DB actions 
in Figure~\ref{fig:sem_ins_del}, Figure~\ref{fig:sem_sel_upd_aggr} and Figure~\ref{fig:sem_create_drop}, 
and the semantic rules for processes and nets in Figure~\ref{fig:semantics_cont}. 
In the explanation, we will avoid reiterating that each table resides in a database located at some $l$, 
  but directly state ``table ... located at $l$''.
Note that an action that refers to a locality variable $u$, e.g., $\INS{\TBID}{t}{u}$, 
cannot be executed until $u$ is bound to some locality $l$.
Hence in our semantic rules the actions (which are ready for execution) always come with 
concrete localities $l$, rather than $u_0$, or $\ell$. 

Examples in the setting of the chain of department stores will be provided to illustrate these rules at work. 
The empty environment will be used for all the illustrations except that of {\sf create}, 
  since it is only with the {\sf create} action that the environment is consulted for information. 

\subsection*{Insertion and Deletion}

Recall that the action $\INS{\TBID}{t}{\lloc}$ 
inserts a new row $t$ into a table with identifier $\TBID$ inside the database at $\lloc$. 
In closed nets, $t$ cannot contain variables when the insertion action is about to be executed, 
just as $\lloc$ cannot contain locality variables.
The rule $\REDINS$ of Figure~\ref{fig:sem_ins_del} describes the performance of this action
from locality $\llocConst_1$, with the parameter $\lloc$ already replaced with some locality $\llocConst_2$. 
It is checked that the table identifier $\TBID$ needs to agree with 
that of a destination table $(I,R)$ already existing at $\llocConst_2$.
If such agreement exists, 
but the evaluation result of tuple $t$ does not satisfy the format requirements 
imposed by the schema $I.\TBSK$, 
then the net $\ErrNet$ is produced; 
otherwise the evaluated tuple is added into the data set $R$.

\begin{figure}[ht!]
  \centering
  \caption{The Semantics for Insertion and Deletion}\label{fig:sem_ins_del}
  \hrule  
  \hspace{-.0cm}\vbox{\small
  \begin{tabular}{l}{~}\\[-1ex]
  $
  \begin{array}{rl}
  \REDINS &
  \Inference{I.\TBID=\TBID}
     {\envnet \vdash \PARNET{\LOCATED{\llocConst_1}{}{\INS{\TBID}{t}{\llocConst_2}.P}~}{~\LOCATED{\llocConst_2}{}{(I,R)}}\rightarrow N'} \\\\[-1ex]
  & \mrm{where}~N'=
     \begin{cases}
      \ErrNet & \mrm{if}~\EVALT{t}\nsattsk I.\TBSK \\
      \PARNET{\LOCATED{\llocConst_1}{}{P}~}{~\LOCATED{\llocConst_2}{}{(I,R\MUNION\{\EVALT{t}\})}} & \mrm{otherwise}
      \end{cases}
  \end{array}
  $
  \\\\
  $\begin{array}{rl}
     \REDDELBND & 
      \Inference
      {I.\TBID=\TBID}                  
      {\envnet\vdash\PARNET{\LOCATED{\llocConst_1}{}{\DEL{\TBID}{T}{\psi}{\llocConst_2}.P}~}{~\LOCATED{\llocConst_2}{}{(I,R)}}\rightarrow N'} \\\\[-1ex]
     & \mrm{where}~N'=
     \begin{cases}
      \ErrNet & \mrm{if}~ T\nsattsk I.\TBSK \lor \exists t\in R:(\subst{t}{T}=\Err \lor
                          \EVALPred{\psi(\subst{t}{T})}=\Err) \\\\[-2.5ex]
      \LOCATED{\llocConst_1}{}{P} ~||~\LOCATED{\llocConst_2}{}{(I,R')} 
                  & \mrm{otherwise} \\
                  &\mrm{where}~R'=\{t\in R~|~\EVALPred{\psi(\subst{t}{T})}\neq \TT\} 
     \end{cases}
     \end{array}$  
    \\\\[-1ex]
  \end{tabular}
}
\hrule
\end{figure}

\begin{exa}[Adding New Shoes]\label{ex:adding_shoes}
  By rule $\REDINS$, the insertion action of Example~\ref{ex:act_ins} can be executed locally at $\ttloc_1$
  as follows.
  \small$$
  \begin{array}{l}
    \vdash \ttloc_1::\INS{{\sf KLD}}{(``001", ``\mi{HB}", ``2015", ``white", ``37", 6, 0)}{\ttloc_1}.
    \NILP~\!||~\!\ttloc_1::(I_1,R_1) \\ 
    \quad \rightarrow~\ttloc_1::\NILP~\!||~\! \ttloc_1::(I_1,R'_1), 
  \end{array}$$
  \normalsize
  where $$R'_1=R_1\uplus \{(``001", ``\mi{HB}", ``2015", ``white", ``37", 6, 0)\}.$$ 
  This is because $(``001", ``\mi{HB}", ``2015", ``white", ``37", 6, 0)\sattsk I_1.\TBSK$ holds. \qed 
\end{exa}

Recall that the action $\DEL{\TBID}{T}{\psi}{\lloc}$ deletes 
from table $\TBID$ at $\lloc$ all rows properly binding the pattern $T$ such that 
the predicate $\psi$ is satisfied.
The rule $\REDDEL$ in Figure~\ref{fig:sem_ins_del} describes the execution of this action from a locality $l_1$ with 
$\lloc$ already replaced with some concrete locality $l_2$. 
It is checked in the premise that 
the specified table identifier $\TBID$ needs to agree with that of the table $(I,R)$ targeted. 
In addition, it is examined whether the template $T$ satisfies the schema $I.\TBSK$, 
and for each row $t$, whether the evaluation of the predicate $\psi$ is erroneous under the substitution 
produced from the match. 
If the satisfaction of $I.\TBSK$ fails, or the evaluation of $\psi$ is sometimes erroneous, then 
the net $\ErrNet$ is produced; 
otherwise the deletion operation is carried out normally, and
the rows that do not match the pattern $T$ or do not satisfy the condition $\psi$ 
will constitute the resulting data set of the target. 
Note that it is most natural for 
the ``where'' clause in the rule $\REDDEL$ to follow the resulting net 
$\LOCATED{\llocConst_1}{}{P} ~||~ \LOCATED{\llocConst_2}{}{(I,R')}$. 
However, we have placed it below ``otherwise'' for space reasons, and will follow the same style for 
some other semantic rules to come. 

\begin{exa}[Deleting Existing Shoes]
  By rule $\REDDEL$, the deletion action of Example~\ref{ex:act_del} 
  can be executed locally at $\ttloc_1$ as follows, 
  targeting the table resulting from the insertion action of Example~\ref{ex:adding_shoes}.
  \small
  $$
  \begin{array}{r@{~\!}l}
    \vdash \ttloc_1:: & {\sf delete}({\sf KLD}@\ttloc_1,(!\mi{id},!\mi{tp},!\mi{yr},!\mi{cr},!\mi{sz},!\mi{is},!\mi{ss}),\\
    &\qquad\qquad\quad~\! \mi{tp}=``\mi{HB}"\land\mi{cr}=``\mi{white}"\land\mi{sz}=``37").\NILP~\!||~\!
    \ttloc_1::(I_1,R'_1) \\
    &\hspace{-.6cm} \rightarrow ~ \ttloc_1::\NILP~\! ||~\! \ttloc_1::(I_1,R_1).
  \end{array}
  $$
  \normalsize
  This reflects that the original table is recovered after the deletion. \qed  
\end{exa}

\subsection*{Selection, Update and Aggregation} 

\begin{figure}[ht!]
  \centering
  \vspace*{-.8cm}\vbox{  \caption{The Semantics for Selection, Update and Aggregation}\label{fig:sem_sel_upd_aggr}
    \hspace{-.8cm}\vbox{\footnotesize
      \begin{tabular}{l}
        \hline\\
   $\begin{array}{rl}
      \REDSEL &
      \Inference{
                \forall \TBID, l:\TBID@l\in \{\LST{\CALTB}\} \Rightarrow
                (\exists j\in\{1,...,k\}: I_{j}.\TBID=\TBID \land \llocConst_j=\llocConst)
      }
      {
        \envnet \vdash 
          \LOCATED{\llocConst_0}{}{\SEL{\LST{\CALTB}}{T}{\psi}{t}{!\TBV}}.P~||~
          \LOCATED{\llocConst_{1}}{}{\!(I_{1},R_{1})}~||~...~||~\LOCATED{\llocConst_{k}}{}{\!(I_{k},R_{k})} 
          \rightarrow N'
      } \\\\[-1ex]
      & \mrm{where}~N'=
      \begin{cases}
      \ErrNet  
        & \mrm{if}~ T\nsattsk \prodSK(\LST{\CALTB},\LSTStruct{\llocConst_i\!::\!(I_i,R_i)}{i})~ \lor \\
        & \quad \exists t'\in \prodR(\LST{\CALTB},\LSTStruct{\llocConst_i\!::\!(I_i,R_i)}{i}): \\
        & \qquad \forall \sigma=\subst{t'}{T}:
          (\sigma=\Err\lor \EVALPred{\psi\sigma}=\Err \lor \EVALT{t\sigma}=\Err) \\\\[-2ex]
      \LOCATED{\llocConst_0}{}{P\sigma'}~\!||~\!N''
        &  \mrm{otherwise} \\
        
        & \mrm{where} \\
        & \quad I'= (\bot,\prodSK(\LST{\CALTB},\LSTStruct{\llocConst_i\!::\!(I_i,R_i)}{i})\!\downarrow^{T}_{t}) \land \\ 
        & \quad R'=\{\EVALT{t\sigma}~|~
                \exists t'\in \prodR(\LST{\CALTB},\LSTStruct{\llocConst_i\!::\!(I_i,R_i)}{i}): \\
        & \qquad\qquad\qquad\qquad\qquad\subst{t'}{T}=\sigma \land ~\! \EVALPred{\psi\sigma}=\TT \}~\!\land \\
        & \quad     \sigma'=[(I',R')/\TBV] \land \\
        & \quad N''=\LOCATED{\llocConst_{1}}{}{(I_{1},R_{1})}~||~...~||~\LOCATED{\llocConst_{k}}{}{(I_{k},R_{k})} 
      \end{cases}
    \end{array}
   $
 \\\\
  $
  \begin{array}{rl}
  \REDUPD & 
  \Inference{I.\TBID=\TBID}
   {\envnet \vdash \PARNET{\LOCATED{\llocConst_1}{}{\UPDATE{\TBID}{T}{\psi}{t}{\llocConst_2}.P}~}{~\LOCATED{\llocConst_2}{}{(I,R)}}\rightarrow N'} \\\\[-1ex]
  & \mrm{where}~N'=
  \begin{cases}
  \ErrNet & \mrm{if}~T\nsattsk I.\TBSK ~\lor \\
  & \quad \exists t'\in R: \\
  & \quad \forall \sigma=\subst{t'}{T}: (\sigma=\Err\lor \EVALPred{\psi\sigma}=\Err \lor \EVALT{t\sigma}=\Err \lor \\
  & \qquad\qquad\qquad\qquad\qquad\quad~~ \EVALPred{\psi\sigma}=\TT \land \EVALT{t\sigma}\nsattsk I.\TBSK) 
  \\\\[-2ex]
  \LOCATED{\llocConst_1}{}{P} ~\!||~\!\LOCATED{\llocConst_2}{}{(I,R'_1\MUNION R'_2)} 
  & \mrm{otherwise} \\
    & \mrm{where}~ 
        R'_1=\{t'\in R~|~\EVALPred{\psi(\subst{t'}{T})}\neq \TT\} \land \\
    & \qquad\quad\! R'_2=\{\EVALT{t\sigma}~|~\exists t'\!:t'\in R \land \subst{t'}{T}=\sigma \land
      \EVALPred{\psi\sigma}=\TT\}
  \end{cases}
  \end{array}
  $
   \\\\
  $
   \begin{array}{rl}
    \REDAGGR & 
    \Inference{I.\TBID=\TBID}
    {\envnet \vdash \LOCATED{\llocConst_1}{}{\AGGR{\TBID}{T}{\psi}{f}{T'}{\llocConst_2}.P}~||~\LOCATED{\llocConst_2}{}{(I,R)}\rightarrow N'} \\\\[-1ex]
    & \mrm{where}~ N'= 
    \begin{cases}
    \ErrNet & \mrm{if}~T \nsattsk I.\TBSK ~\lor \\
    & \quad \exists t\in R: (\subst{t}{T}=\Err\lor\EVALPred{\psi(\subst{t}{T})}=\Err \lor \\
    & \qquad\qquad~~~ 
      \exists \tau_1,\tau_2: (f:\msett{\sktp_1}\rightarrow\tau_2~\land~ (t\nsattsk \sktp_1 \lor T'\nsattsk\sktp_2))) \\\\[-2ex]
    \LOCATED{\llocConst_1}{}{P\sigma'}~\!||~\!\LOCATED{\llocConst_2}{}{(I,R)} 
      & \mrm{otherwise} \\
      & \mrm{where}~t'=f(\{t\in R~|~\EVALPred{\psi(\subst{t}{T})}=\TT\})~\! \land \\
      & \qquad\quad\!\! \sigma'=\subst{t'}{T'} 
    \end{cases}
   \end{array}
  $
   \\\\[-1.5ex]
\hline
  \end{tabular}
  \normalsize
}}
\end{figure}

Recall that the action $\SEL{\LST{\CALTB}}{T}{\psi}{t}{!\TBV}$ picks from the Cartesian product of the data sets 
of tables identified by the list elements of $\LST{\CALTB}$ all rows that properly bind the template $T$  
such that the predicate $\psi$ is satisfied, 
and produces a new table (bound to $\TBV$) that contains the instantiations of $t$ with these bindings. 
The rule $\REDSEL$ in Figure~\ref{fig:sem_sel_upd_aggr} describes the execution of this action at locality $\llocConst_0$.  
It is checked that for each table reference of the form $\TBID@l$ in $\LST{\CALTB}$
with some table identifier $\TBID$ and locality $l$, 
the referenced table is one of those in $\LSTStruct{\llocConst_i::(I_i,R_i)}{i}$. 
In addition, it is examined whether the template $T$ satisfies the product of all the schemas of 
the tables in $\LST{\CALTB}$, and
whether for each tuple $t'$ in the Cartesian product of the data sets of tables in $\LST{\CALTB}$, 
the evaluation of both the predicate $\psi$ and the tuple $t$ under the substitution $t'/T$ can be performed 
successfully. 
In the case of schema mismatches or evaluation errors, the net $\ErrNet$ is produced, signaling the erroneous situation; 
otherwise the selection is carried out properly, with the resulting table $(I',R')$ 
substituted into each occurrence of the table variable $\TBV$ used in the continuation $P$. 
Here we have $\bot$ as the value for $I'.\TBID$, representing that a temporary table is generated, and
the projection of the product of the schemas of the source tables according to $T$ and $t$ as the value for $I'.\TBSK$.
On the other hand, $R'$ consists of all the instantiations of $t$ with the bindings produced by matching the tuples in the 
product data set against $T$, and satisfying the predicate $\psi$. 

\begin{exa}[Selection of Shoes in a Certain Color]\label{ex:select}
  According to the rule $\REDSEL$, the selection action of Example~\ref{ex:act_sel} can be 
  executed from the head office (at $\ttloc_0$) as follows:
  \small
  $$
  \begin{array}{l}
    \vdash \ttloc_0\!::{\sf select}({\sf KLD}@~\!\ttloc_1,(!\mi{id},!\mi{tp},!\mi{yr},!\mi{cr},!\mi{sz},!\mi{is},!\mi{ss}), \\
    ~~~\qquad\qquad\qquad\qquad~
    \mi{id}=``001"\land\mi{tp}=``HB"\land\mi{cr}\neq ``red", (\mi{cr},\mi{sz},\mi{ss}), !\TBV).\NILP~||~ 
   \ttloc_1::(I_1,R_1) \\
    \quad \rightarrow ~\ttloc_0::\NILP ~\!||~\!\ttloc_1::(I_1,R_1).
  \end{array}
  $$
  \normalsize
  The condition $I_1.\TBID={\sf KLD} \land \ttloc_1=\ttloc_1$ 
  as required by the premise of the rule is satisfied. 
  The $I'$ for the resulting table is such that $I'.\TBID=\bot$
  and $I'.\TBSK=\strtp\times \inttp\times \inttp$.
  The table variable $\TBV$ is replaced by $(I',R')$, 
  for some $R'=\{(``black",$ $``38", 2), (``black", ``37", 2) \}$. 
  Note that the selection is completed in one \emph{atomic} transition step, which implies, e.g., 
  that no concurrent update to the table ${\sf KLD}$ at $\ttloc_1$ can be performed in the process.
  \qed
\end{exa}

In the example above, the substitution of the result table $(I',R')$ into the table variable $\TBV$
does not affect the trivial continuation $\NILP$, in which $\TBV$ is not used.
The use of result tables for selection actions will be demonstrated in our case study to be presented in
Section~\ref{sec:case_study}. 
The following example illustrates selection from the join of multiple tables. 

\begin{exa}[Comparing Inventories]\label{ex:compare_ivt}
  The query below generates all pairs of shop names such that
  in each pair, both shops are in the same city and
  the brands sold in the first constitute a proper subset of the brands sold in the second. 
  \small
  $$
  \begin{array}{l}
    {\sf select}(\LSTParenL{\sf Stores}@\ttloc_0,{\sf Stores}@\ttloc_0\LSTParenR, 
    (!x,!y,!z,!w,!p,!x',!y',!z',!w',!p'),
    x=x'\land w\subset w',(z,z'),!\TBV) 
  \end{array}
  $$
  \normalsize
  This action needs to be executed with the parallel component $\ttloc_0::(I_0,R_0)$.
  The result of 
  $\prodSK(\LSTParenL {\sf Stores}@\ttloc_0,{\sf Stores}@\ttloc_0\LSTParenR,
  \LSTParenL \llocConst_0\!::\!(I_0,R_0) \LSTParenR)$ is $\flattensk(I_0.\TBSK\times I_0.\TBSK)$
  and the result of
  $\prodR(\LSTParenL {\sf Stores}@\ttloc_0,{\sf Stores}@\ttloc_0\LSTParenR,
  \LSTParenL \llocConst_0\!::\!(I_0,R_0) \LSTParenR)$ is
  $\flattendt(R_0\times R_0)$. 
  \qed
\end{exa}

The task accomplished in Example~\ref{ex:compare_ivt} can be achieved in SQL using selection 
with an \emph{inner join} operation. 
Although this example only illustrates joining tables at the same locality, 
in Klaim-DB, the join of multiple tables at different localities is generally allowed. 

Recall that the action $\UPDATE{\TBID}{T}{\PRED}{t}{\lloc}$ replaces each row matching $T$ and satisfying $\psi$ 
in table $\TBID$ at $\lloc$ with a new row $t$, while leaving the rest of the rows unchanged. 
The rule $\REDUPD$ in Figure~\ref{fig:sem_sel_upd_aggr} describes the execution of this action at $l_1$ with $\lloc$ already replaced 
with some constant locality $l_2$. 
The premise checks that the specified table identifier matches the actual one. 
In addition, it is examined 
if the template $T$ matches the schema $I.\TBSK$, and
for all tuples $t'$ in $R$, producing substitution $\sigma$ when matched against $T$, whether 
$\psi\sigma$ and $t\sigma$ evaluates properly to values, and whether 
the evaluation result of $t\sigma$ satisfies the schema $I.\TBSK$, whenever $\psi\sigma$ evaluates to $\TT$.
In case of any evaluation error or schema mismatch, 
the net $\ErrNet$ is produced, signaling the abnormal situation; 
otherwise the update is carried out properly.
In more detail, a row $t'$ of $R$ is updated only if it matches the template $T$, resulting in the substitution $\sigma$ that 
makes the predicate $\psi$ satisfied.
The row $t'$ is then updated by applying the substitution $\sigma$. 

\begin{exa}[Update of Shoes Information]
According to the rule $\REDUPD$, 
the update action of Example~\ref{ex:act_upd} can be executed from $\ttloc_1$ as follows: 
\small
$$
\begin{array}{r@{~\!}l}
  \vdash & \ttloc_1\!::\! {\sf update}({\sf KLD}@{\ttloc_1},
  (!\mi{id},!\mi{tp},!\mi{yr},!\mi{cr},!\mi{sz},!\mi{is},!\mi{ss}), \\
  & \qquad\quad\!\! \mi{tp}=``HB"\!\land\!\mi{cr}=``red"\!\land\!\mi{sz}=``37", 
    (\mi{id},\mi{tp},\mi{yr},\mi{cr},\mi{sz}, \mi{is} - 2, \mi{ss} + 2)).\NILP~||~
  \ttloc_1\!::\!(I,R)  \\
  & \rightarrow \ttloc_1::\NILP~\! ||~\!\ttloc_1::(I,R'_{1}\uplus R'_{2}).
\end{array}
$$
\normalsize
The multiset $R'_{1}$ consists of all the entries that are intact --- shoes that are not red high boots sized $37$, 
  while $R'_{2}$ contains all the updated items. \qed
\end{exa}
% set-notation comes close to the tuple relational calculus (possible reference), 
% atomic batch operations are facilitated, sacrificing flexibility in the manipulation of individual rows of data

Recall that the action $\AGGR{\TBID}{T}{\psi}{f}{T'}{\lloc}$ applies the aggregator function $f$ 
on the multiset of all rows matching $T$ and satisfying $\psi$ in table $\TBID$ (at $\lloc$) and 
binds the aggregation result to the pattern $T'$.
The rule $\REDAGGR$ in Figure~\ref{fig:sem_sel_upd_aggr} describes the performance of this action 
from the locality $l_1$, 
with $\lloc$ replaced with some constant locality $l_2$. 
The aggregation is over the table $(I,R)$. 
The matching of localities and table identifiers is still required. 
In addition, it is checked whether the template $T$ matches the schema $I.\TBSK$, and 
for all rows $t$ in $R$, 
whether $t$ matches $T$, 
whether the evaluation of $\psi$ can be performed properly under the substitution produced by $t/T$, and 
whether $t$ and the template $T'$ are respectively well-sorted under 
the domain and range types of the aggregator function $f$. 
If any of these additional checks fails, then the net $\ErrNet$ is produced, 
signaling the error that occurred;  
otherwise the aggregation is carried out properly by
applying $f$ to the multiset of all tuples matching $T$, 
resulting in a binding that satisfies the predicate $\psi$. 
The result $t'$ is bound to the specified template $T'$, 
producing a substitution that is further applied to the continuation $P$. 

\begin{exa}[Aggregation of Sales Figures]
According to the rule $\REDAGGR$, 
the aggregation action of Example~\ref{ex:act_aggr} can be executed 
from the head office as follows: 
\small
$$
\begin{array}{l}
  \vdash \ttloc_1::{\sf aggr}({\sf KLD}@{\ttloc_1},
    (!\mi{id},!\mi{tp},!\mi{yr},!\mi{cr},!\mi{sz},!\mi{is},!\mi{ss}), \mi{id}=``001", \SUM_7, !\mi{res}).\NILP~||~ 
    \ttloc_1::(I_1,R_1) \\
  \quad \rightarrow ~\ttloc_1::\NILP~\! ||~\! \ttloc_1::(I_1,R_1).
\end{array}
$$
\normalsize
The variable $!\mi{res}$ is then bound to the integer value $12$ ($2+5+1+2+2$). \qed
\end{exa}

\begin{exa}[Selection using Aggregation Results]
  Consider the query from $\ttloc_0$ that selects the colors, sizes and sales of all types of high boots
  whose sales figures are above average.
  This query can be modeled as a sequence of actions at $\ttloc_0$, as follows.
\small
$$
\begin{array}{rl}
  &\!\!\!\!\!\! \ttloc_0 :: \AGGR{{\sf KLD}}{T_0}{\mi{tp}=``HB"}{\mi{avg}_7}{!\mi{res}}{\ttloc_1}.
    \SEL{{\sf KLD}@\ttloc_1}{T'_0}{\mi{ss}'\ge\mi{res}}{(\mi{cr}',\mi{sz}',\mi{ss}')}{!\TBV}.\NILP
\end{array}
$$
\normalsize
where \small$T_0=(!\mi{id},!\mi{tp},!\mi{yr},!\mi{cr},!\mi{sz},!\mi{is},!\mi{ss})$, \normalsize
      \small$T'_0=(!\mi{id}',!\mi{tp}',!\mi{yr}',!\mi{cr}',!\mi{sz}',!\mi{is}',!\mi{ss}')$, \normalsize and
\small
$$
\AVG_7=\lambda R.\frac{\SUM_7(R)}{|R|}=\lambda R.(\frac{\SUM(\{v_7|(v_1,...,v_7)\in R\})}{|R|}).
$$
\normalsize
In other words, $\AVG_7$ is a function from multisets $R$ to unary tuples containing the average value of the 7-th components of the tuples in $R$. \qed
\end{exa}

\subsection*{Creating and Dropping Tables}

Recall that the action $\CREATENEW{\TBID@\lloc}{\TBSK}$ creates a table with interface $(\TBID,\TBSK)$ 
at locality $\lloc$. 
The rule $\REDCREATE$ in Figure~\ref{fig:sem_create_drop} describes the execution of this action from the locality $l_1$, 
with $\lloc$ already instantiated with some locality constant $\llocConst_2$. 
It is ascertained that no table having the same identifier $\TBID$ exists in the environment at the target locality. 
If this is the case, then a table with the interface $I=(\TBID,\TBSK)$ and an empty data set is created
at the specified locality; otherwise the creation is skipped. 
The aforementioned check for clashes of table identifiers at the same localities 
is realized with the help of the function $\LIDN{...}$ that 
gives the multiset of pairs of localities and table identifiers in nets and components. 
This function is overloaded on nets and components, and is defined inductively as follows.
\small
$$
\begin{array}{ll}
\LIDN{\NILN}=\emptyset ~& \LIDN{\LOCATED{\llocConst}{}{C}}=\LIDC{\llocConst}{C} \\
\LIDN{\ErrNet}=\emptyset ~& \LIDC{\llocConst}{P}=\emptyset \\
\LIDN{\PARNET{N_1}{N_2}}=\LIDN{N_1}\MUNION\LIDN{N_2} ~& \LIDC{\llocConst}{(I,R)}=\{(\llocConst,I.id)\} \\
\LIDN{\RESNET{\llocConst}{N}}=\LIDN{N} ~& \LIDC{\llocConst}{\PAR{C_1}{C_2}}=\LIDC{\llocConst}{C_1}\MUNION \LIDC{\llocConst}{C_2} \\
\end{array}
$$
\normalsize

Recall that the action $\DROP{\TBID}{\lloc}$ drops the table with identifier $\TBID$ at locality $\lloc$.
The rule $\REDDROP$ in Figure~\ref{fig:sem_create_drop} describes the performance of this action from the locality $l_1$, 
with $\lloc$ already replaced with some constant locality $l_2$. 
It checks that a table with the specified identifier $\TBID$ 
does exist at the specified locality $\llocConst_2$. 
Then the table is dropped by replacing it with $\NILP$.

\begin{figure}[ht!]
  \caption{The Semantics for Creating and Dropping Tables}\label{fig:sem_create_drop}
  \centering
  \hrule
  \hspace{-.0cm}\vbox{\small
    \begin{tabular}{l}{~}\\[-1ex]
      $
      \begin{array}{rl}
        \REDCREATE & 
                     \envnet \vdash \LOCATED{\llocConst_1}{}{\CREATENEW{\TBID@\llocConst_2}{\TBSK}.P}~||~
                     \LOCATED{\llocConst_2}{}{\NILP}\rightarrow N' \\\\[-2ex]
                   & \mrm{where}~
                     N'=
                     \begin{cases}
                       \LOCATED{\llocConst_1}{}{P}~||~\LOCATED{\llocConst_2}{}{((\TBID,\TBSK),\emptyset)} 
                       & \mrm{if}~(\llocConst_2,\TBID)\not\in\LIDN{\envnet} \\
                       \LOCATED{\llocConst_1}{}{P}~||~\LOCATED{\llocConst_2}{}{\NILP} & \mrm{otherwise}
                     \end{cases}
      \end{array}
  $
   \\\\[-1ex]
   $\REDDROP$ 
   $\Inference{I.\TBID=\TBID}
   {\envnet\vdash
    \LOCATED{\llocConst_1}{}{\DROP{\TBID}{\llocConst_2}.P}~||~\LOCATED{\llocConst_2}{}{(I,R)}\rightarrow 
    \LOCATED{\llocConst_1}{}{P}~||~\LOCATED{\llocConst_2}{}{\NILP}}$ \\\\[-1ex]
    \end{tabular}
  }
  \hrule
\end{figure}
\enlargethispage{\baselineskip}

\begin{exa}[Creating Table for Turnovers]\label{ex:create}
Consider the net $N_{\mrm{r}}$, which is 
\small
$$\ttloc_0::((I_0,R_0)|C''_0) || \ttloc_1::((I_1,R_1)|C'_1) || ... || \ttloc_n::((I_n,R_n)|C'_n),$$
\normalsize
where $C''_0$ is the residual component at $\ttloc_0$ --- the locality of the head office. 
Suppose $(\ttloc_0,{\sf Turnover})\not\in \LIDN{N_\mrm{r}}$, 
which indicates that there is no table with identifier {\sf Turnover} currently existing at $\ttloc_0$. 
According to the rule $\REDCREATE$, 
the table creation action of Example~\ref{ex:act_crt} can be executed 
from the head office, with $N_{\mrm{r}}$ as the environment, as follows: 
\small
$$
  \begin{array}{r@{~\!}l}
  N_{\mrm{r}}\vdash & \ttloc_0::\CREATENEW{{\sf Turnover}@\ttloc_0}{\strtp\times \inttp}.\NILP~\! ||~\! 
  \ttloc_0::\NILP\\
  & \rightarrow \ttloc_0::\NILP ~\!||~\! \ttloc_0::(({\sf Turnover},\strtp\times\inttp),\emptyset).
  \end{array}\vspace{-16 pt}
$$
\normalsize
\qed
\end{exa}

\subsection*{Code Mobility}
Recall that the action $\EVAL{P}{\lloc}$ spawns the process $P$ at the locality $\lloc$.
The rule $\REDEVALP$ describes the execution of this action 
at the locality $l_1$ with $\lloc$ already replaced by the concrete locality constant $\llocConst_2$. 
It is ensured that $P$ is a closed process that does not contain \emph{free variables}
before the action can be performed.  
\small
$$
   \REDEVALP
   \quad
   \envnet\vdash
    \LOCATED{\llocConst_1}{}{\EVALP{P}{\llocConst_2}.P'}~||~\LOCATED{\llocConst_2}{}{\NILP}\rightarrow 
    \LOCATED{\llocConst_1}{}{P'}~||~\LOCATED{\llocConst_2}{}{P} 
   \quad \mrm{if}~\FV{P}=\emptyset
$$   
\normalsize

In Klaim-DB, it is possible to spawn multiple processes that migrate among different sets of localities
to accomplish diverse database query and management tasks.
In this way, the parallelization and distribution of operations on databases is facilitated. 

\subsection*{Semantics for Processes and Nets}

Directing our attention to Figure~\ref{fig:semantics_cont}, the rules $\REDFORTT$ and $\REDFORFF$ specify 
  the continuation and completion of {\sf foreach} loops, respectively. 
The rule $\REDFORTT$ says that if there are still more tuples in the multiset $R$ 
  matching the pattern specified by the template $T$, 
  such that $\psi$ holds under the resulting substitution, 
  then we first execute one round of the loop, 
  instantiating variables in $T$ with a (non-deterministically chosen) minimal matching tuple $t_0$, 
  and then continue with the remaining rounds of the loop by removing $t_0$ from $R$. 
In the premise, $\minimal(R,\ordr)$ gives the set of all tuples that are minimal in $R$ 
  with respect to the partial order $\ordr$:
\small
$$
\minimal(R,\ordr) = \{t\in R\mid \forall t'\in R: t'\ordr t \Rightarrow t'=t\}.
$$
\normalsize
The rule $\REDFORFF$ says that if there are no more tuples in $R$ matching $T$, then the loop is completed. 

\begin{figure}[ht!]
  \centering
  \caption{The Semantics for Processes and Nets}\label{fig:semantics_cont}
 \vbox{\small
  \begin{tabular}{c}
\hline \\[-1ex]
$\REDFORTT$ $\Inference{
  t_0\in\minimal(R,\ordr)\quad \subst{t_0}{T}\neq\Err \quad \EVALPred{\psi(\subst{t_0}{T})}=\TT
}
{
  \envnet\vdash
  \LOCATED{\llocConst}{}{\FOREACH{(I,R)\!}{T}{\psi}{\ordr}{\!P}}\rightarrow \LOCATED{\llocConst}{}{P(t_0/T)}; 
  (\FOREACH{(I,R\setminus \{t_0\})}{T}{\psi}{\ordr}{P})
}$ 
\\ \\
$
\begin{array}{rl}
 \REDFORFF &
 \Inference{\lnot(\exists t_0\in R: \EVALPred{\psi(\subst{t_0}{T})}=\TT)}
           {\envnet\vdash\LOCATED{\llocConst}{}{\FOREACH{(I,R)\!}{T}{\psi}{\ordr}{\!P}\rightarrow N'}} \\
 & \mrm{where}~N'= 
  \begin{cases}
    \ErrNet & \mrm{if}~\exists t_0\in R: 
            (\subst{t_0}{T}=\Err \lor \EVALPred{\psi(\subst{t_0}{T})}=\Err) \\
    l::\NILP & \mrm{otherwise}
   \end{cases}
\end{array}
$ 
\\ \\
$\REDSEQTT$ 
$\Inference{
  \envnet \vdash \LOCATED{\llocConst}{}{P_{1}}~||~N \rightarrow\LOCATED{\llocConst}{}{P'_{1}}~||~N'}
{\envnet\vdash\LOCATED{\llocConst}{}{\SEQ{P_{1}}{P_{2}}}~||~N \rightarrow \LOCATED{\llocConst}{}{\SEQ{P'_{1}}{P_{2}}}~||~N'}$
\quad if~$P'_1\neq\NILP$
\\ \\ 
$\REDSEQFF$ $\Inference{
  \envnet \vdash \LOCATED{\llocConst}{}{P_{1}}~||~N \rightarrow \LOCATED{\llocConst}{}{\NILP}~||~N'}
  {\envnet\vdash \LOCATED{\llocConst}{}{\SEQ{P_{1}}{P_{2}}}~||~N \rightarrow \LOCATED{\llocConst}{}{P_{2}}~||~N'}$
\\ \\

$\begin{array}{rl}
  (\mrm{CALL}) & \envnet\vdash \llocConst:: A(\tilde{e})\rightarrow \llocConst:: P[v_1/\avar_1]...[v_n/\avar_n] \\
   & \quad \mrm{if}~A(\LSTStruct{\avar_i:\tau_i}{i\le n})\triangleq P\land \EVALT{\tilde{e}}=\tilde{v}
\end{array}
$
\\ \\
$\REDPAR$ 
  $\Inference{N_2\vdash N_1 \rightarrow N'_1}{\vdash N_1||N_2 \rightarrow N'_1||N_2}$ \\ \\
$\REDRES$ $\Inference{\envnet\vdash N\rightarrow N'}{\envnet\vdash \RESNET{\llocConst}{N}\rightarrow \RESNET{\llocConst}{N'}}$ \qquad
$\REDEQUIV$ 
  $\Inference{N_1\equiv N_2\quad \envnet\vdash N_2\rightarrow N_3\quad N_3\equiv N_4}
    {\envnet\vdash N_1\rightarrow N_4}$ \\ \\[-1ex]
\hline\normalsize
  \end{tabular}
}
\end{figure}

The rules $\REDSEQTT$ and $\REDSEQFF$ describe the transitions that can be performed by 
sequential compositions $\SEQ{P_1}{P_2}$. 
In particular, the rule $\REDSEQTT$ accounts for the case where $P_1$ cannot finish after one more step; 
whereas the rule $\REDSEQFF$ accounts for the opposite case. 
Note that by our language design, the variables declared in $P_1$ have local scopes; 
hence there is no need to apply any substitution on $P_2$ when the execution embarks on it. 

The rule $\REDPAR$ says that if a net $N_1$ can make a transition assuming a net $N_2$ as the environment, 
then the parallel composition of $N_1$ with $N_2$ can also make a transition, 
with an empty environment ($\NILN$). 
Having an empty environment in the conclusion indicates that the rule cannot be used 
multiple times with non-trivial effects. 
The idea is: 
to build a transition from a net with multiple parallel components, 
we use the structural congruence to extrude all the $\RESNET{l}{(\cdot)}$ to the outer level, 
and collect all the unaffected parallel nets into $N_2$. 
We can then build the transition using $\REDPAR$ only once, 
combined with further uses of $\REDRES$, and of $\REDEQUIV$. 

The rules $(\mrm{CALL})$, $\REDRES$ and $\REDEQUIV$ are self-explanatory.

It is a property of our semantics that no local repetition of table identifiers can be caused by a transition. 
Define $\NOREP{A}=(A=\mi{set}(A))$, which expresses that there is no repetition in the multiset $A$, 
thus $A$ coincides with its underlying set. 
This property is formalized in Lemma~\ref{lem:no_rep} whose proof can be found in Appendix~\ref{app:proofs}.
\begin{lem}\label{lem:no_rep} 
If $\NOREP{\LIDN{N}\uplus \LIDN{\envnet}}$ and $\envnet\vdash\TRANS{N}{N'}$, 
then it holds that $\NOREP{\LIDN{N'}\uplus \LIDN{\envnet}}$.
\end{lem}
Thus, imposing ``non-existence of local repetition of table identifiers'' as an integrity condition 
for the initial net will guarantee the satisfaction of this condition in all possible derivatives of the net.

The following example illustrates a transition of a net obtained from one of an action. 
\begin{exa}[Transition of Nets]\label{ex:net_trans}
Continuing with Example~\ref{ex:create}, with the help of the rules $\REDPAR$, and $\REDEQUIV$, 
we can derive 
\small
$$
\begin{array}{l}
  \vdash 
  \ttloc_0\!::\! (I_0,R_0)|\CREATENEW{{\sf Turnover}@\ttloc_0}{\strtp\times\inttp}.\NILP|C''_0~ ||~ 
  \ttloc_1\!::\! (I_1,R_1)|C'_1 ~||~ ... ~||~\ttloc_n\!::\! (I_n,R_n)|C'_n \\
  \quad \rightarrow \LOCATED{\ttloc_0}{}{(I_0,R_0)|(({\sf Turnover},\strtp\times\inttp),\emptyset)|C''_0} ~||~ 
  \LOCATED{\ttloc_1}{}{(I_1,R_1)|C'_1}~||~...~||~\LOCATED{\ttloc_n}{}{(I_n,R_n)|C'_n}\mathstrut. 
\end{array}\vspace{-16 pt}
$$
\normalsize
\qed
\end{exa}

%%% Local Variables:
%%% mode: latex
%%% TeX-master: "paper"
%%% End:

%!TEX root = ./paper.tex

\section{Type System}\label{sec:typing}

It is desirable to eliminate certain inherent discrepancies in a piece of specification
before ever executing it. 
In this section we develop a type system 
that scans specifications for discrepancies that lead to the run-time errors
described by the semantics in Section~\ref{sec:semantics}. 
We now detail the types of such errors that we alluded to in the introduction:
\begin{itemize}
\item inconsistency between the format of a template $T$ and a schema $\TBSK$,
  such as when $T$ is used in a selection from a table having the schema $\TBSK$,
  $T$ has three components while $\TBSK$ involves only two basic datatypes;
\item inconsistency between the format of a tuple $t$ and a schema $\TBSK$,
  such as when $t$ is to be inserted into a table having the schema $\TBSK$,
  the second component of $t$ is an integer, but is supposed to be a string according to $\TBSK$; 
\item mismatch of the signature of an aggregator function against
  the format of the data source; 
\item ill-formed expressions and predicates,
  such as a multiplication of two strings, or a predicate that cannot be evaluated to a boolean value; 
\item use of unbound (undefined) variables.  
\end{itemize}

A typing environment $\Gamma$ is a finite sequence of bindings $\avar:\tau$ of variables to types. 
For the variables in the domain of $\Gamma$ (written $\dom{\Gamma}$), 
$\Gamma(x)$ is the data type bound to $x$, 
$\Gamma(\TBV)$ is the schema bound to $\TBV$, and
$\Gamma(u)$ is $\loctp$ --- the only type that can be bound to locality variables $u$. 
We assume that at most one binding can be present for each variable in $\Gamma$, 
which creates no essential restriction since bound variables can be renamed apart. 
We also write $\emptyTPEnv$ for the empty typing environment, 
$(\ext{\Gamma}{\avar:\tau})$ for the \emph{extension} of $\Gamma$ with the binding $\avar:\tau$, and
similarly $\ext{\Gamma_1}{\Gamma_2}$ for the \emph{extension} of $\Gamma_1$ with (all bindings of) $\Gamma_2$. 

We assume that only one schema can be used for each table identifier at all localities, all the time. 
This assumption may cause slight peculiarities concerning the use of table identifiers, 
but avoids the potential explosion caused by over-approximating the sets of localities evaluated to 
by locality variables. 
Correspondingly, we introduce a partial function $\itfc$ from table identifiers to 
their corresponding schemas. 

The typing judgments for expressions, predicates and tuples are as follows. 
\small
$$
\begin{array}{r@{~~~\!}l}
\Gamma &\vdash e \smalltriangleright \type \\[1ex]
\Gamma &\vdash \psi \smalltriangleright \booltp \\[1ex]
\Gamma &\vdash t \smalltriangleright \tau 
\end{array}
$$
\normalsize
The environment $\Gamma$ provides the types for the constituent variables of these syntactical entities, 
and the types $\type$ derived from them are placed after the separator $\smalltriangleright$. 
A well-typed predicate always has the type $\booltp$. 

The judgment for templates is:
\small
$$
\type  \vdash T \smalltriangleright \Gamma 
$$
\normalsize
Due to the involvement of variable binding,
a typing environment $\Gamma$ is derived instead of a type.
In more detail, $\type$ is a product type and $\Gamma$ records the association of
each variable in $T$ to the corresponding component type in $\type$.
When this judgment is used, $\type$ is often the schema of a table or the join of some tables,
and deriving $\Gamma$ from $\type$ captures that the bound variables of $T$
obtain their types from the actual types of the fields of the table.

The judgment for tables is: 
\small
$$
\Gamma,\itfc \vdash \CALTB \smalltriangleright \tau 
$$
\normalsize
Both the variable environment $\Gamma$ and the partial function $\itfc$ are needed --- 
depending on the particular form of $\CALTB$, 
its schema $\tau$ is derived from either $\Gamma$ or $\itfc$.
A frequent usage pattern is that of deriving from a table its schema $\type$,
and then using the derived $\type$ to type a template $T$ used in an operation to match against the rows of the table. 

The judgment for actions is:
\small
$$
\Gamma,\itfc \vdash a \smalltriangleright \Gamma'
$$
\normalsize
Due to the potential involvement of tables in actions, the partial function $\itfc$ is again needed. 
Actions such as aggregation can create bindings to be used in the continuation. 
Correspondingly, the derived environment $\Gamma$ associates the variables in the domain of such bindings
with their types. 

Finally, the judgments for processes, components and nets are listed below. 
\small
%$$
%\begin{array}{r@{~~~\!}l}
%\Gamma,\itfc &\vdash P  \\[1ex]
%\Gamma,\itfc &\vdash C \\[1ex]
%\Gamma,\itfc &\vdash N 
%\end{array} 
%$$
$$
\Gamma,\itfc\vdash P \qquad 
\Gamma,\itfc\vdash C \qquad
\Gamma,\itfc\vdash N 
$$
\normalsize
Similar to the previous cases,
$\Gamma$ and $\itfc$ are needed for obtaining the types of variables
and tables referenced with identifiers, respectively.
These judgments only assert that a process, component, or net is well-typed,
without deriving a type. 

\subsection{Typing Rules}

The typing rules for expressions, predicates, tuples, templates and tables are shown in Figure~\ref{fig:tp_ett}. 
The typing rules for actions are then given in Figure~\ref{fig:tp_act}.
Finally, the typing rules for processes, components, and nets are displayed in Figure~\ref{fig:tp_pcn}. 

\begin{figure}[ht!]
\caption{Typing Rules for Expressions, Predicates, Tuples, Templates and Tables}\label{fig:tp_ett}
\hrule
\centering
\small
\begin{tabular}{c}
{~}\\[-1ex]\hspace{-12cm}\framebox[1.2\width]{Expressions}
\\\\[-2ex]
$\jdgtt{\Gamma}{x}{\tau}\quad\mrm{if}~\Gamma(x)=\tau\land \tau\neq\loctp$\qquad 
$\jdgtt{\Gamma}{u}{\loctp}\quad\mrm{if}~\Gamma(u)=\loctp$
\\\\[-1ex]
$\jdgtt{\Gamma}{num}{\inttp}$\qquad
$\jdgtt{\Gamma}{str}{\strtp}$\qquad
$\jdgtt{\Gamma}{\TBID}{\idtp}$\qquad 
%$\jdgtt{\Gamma}{s}{\stp}$\qquad
$\jdgtt{\Gamma}{l}{\loctp}$
\\\\[-1ex]
$
\Inference{\jdgtE{\Gamma}{e_1}{\strtp}\quad\jdgtE{\Gamma}{e_2}{\strtp}}{\jdgtE{\Gamma}{e_1\strcon e_2}{\strtp}}
$\qquad
$
\Inference{\jdgtE{\Gamma}{e_1}{\inttp}\quad \jdgtE{\Gamma}{e_2}{\inttp}}{\jdgtE{\Gamma}{e_1~\!\aop~\!e_2}{\inttp}}
$
\\\\[-1ex]
$
\Inference{\jdgtE{\Gamma}{e_1}{\dttp}~~...~~\jdgtE{\Gamma}{e_n}{\dttp}}
{\jdgtE{\Gamma}{\{e_1,...,e_n\}}{\msett{\dttp}}}
$
\\\\[-2ex]
\hspace{-12cm}\framebox[1.2\width]{Predicates}
\\
$\jdgtpred{\Gamma}{\TRUE}$
\\[1.5ex]
$
\Inference{\jdgtE{\Gamma}{e_1}{\dttp}\quad\jdgtE{\Gamma}{e_2}{\dttp}}{\jdgtpred{\Gamma}{e_1~\!\cop~\!e_2}}
$
\qquad
$
\Inference{\jdgtE{\Gamma}{e_1}{\dttp}\quad\jdgtE{\Gamma}{e_2}{\msett{\dttp}}}{\jdgtpred{\Gamma}{e_1\in e_2}}
$
\\\\[-1ex]
$
\Inference{\jdgtpred{\Gamma}{\psi}}{\jdgtpred{\Gamma}{\lnot\psi}}
$
\qquad\quad
$
\Inference{\jdgtpred{\Gamma}{\psi_1}\quad \jdgtpred{\Gamma}{\psi_2}}{\jdgtpred{\Gamma}{\psi_1\land\psi_2}}
$
\\\\
\hspace{-12cm}\framebox[1.2\width]{Tuples}
\\[-2ex]
$
\Inference{
	\jdgtt{\Gamma}{e_1}{\type_1}\quad...\quad \jdgtt{\Gamma}{e_n}{\type_n} 
}
{
	\jdgtt{\Gamma}{e_1\rcdcon...\rcdcon e_n}{\tau_1\times...\times\tau_n}
}
$
\\\\[-2ex]
\hspace{-12cm}\framebox[1.2\width]{Templates}
\\\\[-2ex]
$
\jdgtT{\msettp}{!x}{[x:\msettp]}\quad \mrm{if}~\msettp\neq\loctp 
$
\\\\[-1ex]
$
\jdgtT{\loctp}{!u}{[u:\loctp]}
$
\qquad
$
\Inference{
	\jdgtT{\sktp_1}{\bd_1}{\Gamma_1} \quad ...\quad \jdgtT{\sktp_n}{\bd_n}{\Gamma_n}
}
{
	\jdgtT{\sktp_1\times...\times\sktp_n}{\bd_1\rcdcon...\rcdcon \bd_n}{\Gamma_1,...,\Gamma_n} 
}
$
\\\\[-1.5ex]
\hspace{-12cm}\framebox[1.2\width]{Tables}
\\\\[-1ex]
$
\Inference
{\jdgtE{\Gamma}{\lloc}{\loctp}}
{\jdgtE{\Gamma,\itfc}{\TBID@\lloc}{\itfc(\TBID)}} \quad 
\mrm{if}~\TBID\in\dom{\itfc}$\qquad
$\jdgtE{\Gamma,\itfc}{\TBV}{\Gamma(\TBV)} \quad \mrm{if}~\TBV\in\dom{\Gamma}$\\\\[-1ex]
$\Inference{\forall t\in R: \jdgtt{\Gamma}{t}{I.\TBSK}}{\jdgtE{\Gamma,\itfc}{(I,R)}{I.\TBSK}}\quad 
	\mrm{if}~I.\TBID\neq\bot\Rightarrow\itfc(I.\TBID)=I.\TBSK$
\\\\[-1ex]
\end{tabular}
\normalsize
\hrule
\end{figure}

\subsubsection*{Typing Expressions, Predicates, Tuples, Templates and Tables}

Directing our attention to Figure~\ref{fig:tp_ett}, 
for expressions we obtain the types of data variables and locality variables
directly from the typing environment $\Gamma$ (first and second rule).
For numerals, string literals, table identifiers and locality constants,
we assume that it is clear which types they are from their literal appearance (next four rules).
The concatenation of two expressions of type $\strtp$ has type $\strtp$, 
and an arithmetic expression involving sub-expressions of type $\inttp$ has type $\inttp$ (next two rules).
The type $\msett{\type}$ can be derived from a multiset only if
all the elements have the same type $\type$ (last rule for expressions). 

For predicates, 
a comparison $e_1~\!\cop~\!e_2$ has type $\booltp$ if its sub-expressions $e_1$ and $e_2$ 
have the same type. 
A membership test $e_1\in e_2$ has type $\booltp$ if 
the type of $e_2$ reflects that it is a multiset of elements of the type that $e_1$ has. 
The composite predicates $\lnot\PRED$ and $\PRED_1\land\PRED_2$ are well-typed 
if their sub-predicates are well-typed. 

The type of a tuple $e_1,...,e_n$ is the product of the types of its component expressions. 

For templates, a singleton template $!x$ for a data variable $x$ can be typed 
given a multiset type $\msettp$ that is not $\loctp$, 
with the environment $[x:\msettp]$ derived. 
This reflects the fact that when the type of a bound variable is not declared with it in the syntax, 
we obtain the type of the variable from the type $\msettp$ of its data source. 
For a locality variable, the only possibility is to derive the environment $[u:\loctp]$, 
provided that the type $\loctp$ is given to the left of the turnstile. 
Finally, for a composite template $\bd_1,...,\bd_n$, 
a product type $\type_1\times ... \times\type_n$ should be given,  
each component $\bd_j$ needs to be typed under $\type_j$, 
and the bindings from the derived typing environments $\Gamma_1$, $\Gamma_2$ ..., and $\Gamma_n$ 
constitute the typing environment derived from the template $\bd_1,...,\bd_n$. 

\begin{exa}
The template $(!id,!tp,!yr,!cr,!sz,!is,!ss)$ that 
can be used to match against the records of the table {\sf KLD} 
(such as in the selection operation of Example~\ref{ex:select})
can be typed under the schema of the table (on the left of the turnstile) as follows. 
\small
$$
\begin{array}{l}
	\strtp\times\strtp\times\strtp\times\strtp\times\strtp\times\inttp\times\inttp
	\vdash (!id,!tp,!yr,!cr,!sz,!is,!ss) \smalltriangleright \\
      \qquad\quad [id:\strtp,tp:\strtp,yr:\strtp,cr:\strtp,sz:\strtp,is:\inttp,ss:\inttp]
\end{array}\vspace{-16 pt}
$$
\normalsize \qed
\end{exa}

\noindent The rules for tables distinguish between three cases. 
If the table to be typed is a reference of the form $\TBID@\lloc$, 
then its type is obtained from $\itfc$ as $\itfc(\TBID)$, 
as long as $\TBID$ is in the domain $\dom{\itfc}$ of $\itfc$, 
and $\lloc$ is well-typed under the same typing environment. 
The well-typedness of $\lloc$ ensures that when $\lloc$ is a locality variable $u$,
$u$ exists in the domain of the typing environment, 
which in turn indicates that $u$ must have been bound previously. 
If the table is a variable $\TBV$ in the domain of the typing environment $\Gamma$, 
then a binding of $\TBV$ to a schema is supposed to exist in $\Gamma$, and
the schema is obtained from $\Gamma$ as the type of $\TBV$. 
Finally, if the table is a concrete one, $(I,R)$, then its schema is directly obtained from $I$. 
However, it needs to be checked that each tuple in $R$ indeed has the schema $I.\TBSK$ as its type, and 
in case the table identifier $I.\TBID$ is not $\bot$, 
$\itfc$ correctly records the schema $I.sk$ corresponding to this identifier.

\begin{figure}[ht!]
\caption{The Typing Rules for Actions}\label{fig:tp_act}
\centering
\hrule
\small
\begin{tabular}{c}{~}\\[-1ex]
$
\Inference{
	\jdgtt{\Gamma}{t}{\itfc(\TBID)}\quad \jdgtE{\Gamma}{\lloc}{\loctp} 
}
{
	\jdgta{\Gamma,\itfc}{\INS{\TBID}{t}{\ell}}{\emptyTPEnv}
}
$
\\\\[-1ex]
$
\Inference{
	\jdgtT{\itfc(\TBID)}{T}{\Gamma'}\quad \jdgtpred{\ext{\Gamma}{\Gamma'}}{\psi}\quad 
	\jdgtE{\Gamma}{\lloc}{\loctp}
}
{
	\jdgta{\Gamma,\itfc}{\DEL{\TBID}{T}{\psi}{\ell}}{\emptyTPEnv}
} 
$
\\\\[-1ex]
$
\Inference{
	\begin{array}{c}
	\jdgtE{\Gamma,\itfc}{\CALTB_1}{\sktp_1} \quad...\quad \jdgtE{\Gamma,\itfc}{\CALTB_n}{\sktp_n}\quad 
	\jdgtT{\flattensk(\sktp_1\times...\times\sktp_n)}{T}{\Gamma'} \\
	\jdgtpred{\ext{\Gamma}{\Gamma'}}{\psi} \quad \jdgtt{\ext{\Gamma}{\Gamma'}}{t}{\tau'} 	
	\end{array}
}
{
	\jdgta{\Gamma,\itfc}{\SEL{\LST{\CALTB}}{T}{\psi}{t}{!\TBV}}{[\TBV:\tau']}
}
$
\\\\[-1ex]
$
\Inference{
	\jdgtT{\itfc(\TBID)}{T}{\Gamma'}\quad \jdgtpred{\ext{\Gamma}{\Gamma'}}{\psi}\quad  
	\jdgtt{\ext{\Gamma}{\Gamma'}}{t}{\itfc(\TBID)}\quad 
	\jdgtE{\Gamma}{\lloc}{\loctp}
}
{
	\jdgta{\Gamma,\itfc}{\UPDATE{\TBID}{T}{\psi}{t}{\ell}}{\emptyTPEnv}
}
$
\\\\[-1ex]
$
\Inference{
	\jdgtT{\itfc(\TBID)}{T}{\Gamma'}\quad \jdgtpred{\ext{\Gamma}{\Gamma'}}{\psi}\quad 
	f:\msett{\itfc(\TBID)}\rightarrow\sktp' \quad \jdgtT{\sktp'}{T'}{\Gamma''}\quad 
	\jdgtE{\Gamma}{\lloc}{\loctp}
}
{
	\jdgta{\Gamma,\itfc}{\AGGR{\TBID}{T}{\psi}{f}{T'}{\ell}}{\Gamma''}
}
$
\\\\
$
\Inference{\jdgtE{\Gamma}{\lloc}{\loctp}}
{\jdgta{\Gamma,\itfc}{\CREATENEW{\TBID@\lloc}{\itfc(\TBID)}}{\emptyTPEnv}}
$
\qquad
$
\Inference{\jdgtE{\Gamma}{\lloc}{\loctp}} 
{\jdgta{\Gamma,\itfc}{\DROP{\TBID}{\ell}}{\emptyTPEnv}}
$
\\\\
$
\Inference{\jdgtP{\Gamma,\itfc}{P}\quad \jdgtE{\Gamma}{\lloc}{\loctp}}
{\jdgta{\Gamma,\itfc}{\EVALP{P}{\ell}}{\emptyTPEnv}}
$
\\ \\[-1ex]
\end{tabular}
\normalsize
\hrule
\end{figure}

\subsubsection*{Typing Actions}
Turning to Figure~\ref{fig:tp_act},
the rules always type check the localities occurring in the actions to be typed --- 
we do not repeat this aspect in each individual rule that we explain below.

For an \emph{insertion} action to be typed, 
the tuple $t$ inserted into the table $\TBID@\lloc$ needs to have the type $\itfc(\TBID)$ --- 
the schema of all potential tables referenced by the identifier $\TBID$. 

For a \emph{deletion} action to be typed, 
the specified template $T$ needs to match the structure of the target table, 
thus it needs to be well-typed under the schema $\itfc(\TBID)$ of that table.
A typing environment $\Gamma'$ is then derived, associating the bound variables in $T$ to their types. 
In addition, under the extension $\Gamma,\Gamma'$ of the original typing environment 
$\Gamma$ with $\Gamma'$, the predicate $\psi$ needs to be well-typed. 
This extension captures that the free variables of $\psi$ 
can be partly bound before the deletion action is ever performed and
associated with types by $\Gamma$, 
and partly bound in $T$ and associated with types by $\Gamma'$. 

To type a \emph{selection} action, 
the specified template $T$ needs to be well-typed under the 
(flattened) product of the types of all the source tables, 
to ensure that $T$ matches the structure of the Cartesian product of these source tables.
A typing environment $\Gamma'$ is then derived from typing $T$, binding variables of $T$ to their types. 
Under the extension $\ext{\Gamma}{\Gamma'}$ of $\Gamma$ with $\Gamma'$, 
it is then checked that both the predicate $\psi$ and the tuple $t$ should be typable. 
In particular, the type $\tau'$ of $t$ is supposed to be identical to 
the schema of the result table to be bound to $\TBV$.
Correspondingly the environment $\TBV:\tau'$ is derived. 

The typing rule for the \emph{update} action can be understood similarly. 
Since the instantiations of the tuple $t$ are put back in the target
table with identifier $\TBID$ that already exists, 
$t$ is supposed to have the schema $\itfc(\TBID)$ as its type, 
and the empty environment $\emptyTPEnv$ is derived 
since no variables bound in this action are further used in the continuation. 

In the typing rule for \emph{aggregation} actions, 
a similar pattern is involved in the typing of the specified template $T$ and predicate $\psi$. 
Since the aggregator function $f$ takes a multiset of tuples in the data set of the table $\TBID$ 
as its argument, 
the type of $f$ is supposed to be $\msett{\itfc(\TBID)}\rightarrow \tau'$, 
for some type $\tau'$ of the aggregation result. 
Since this result is further bound to $T'$, 
the template $T'$ needs to be well-typed under $\tau'$, 
producing a new typing environment $\Gamma'$ that is also the environment derived 
in the typing of the whole action. 

The action for \emph{table creation} is typed by ensuring that the schema of the table 
to be created at $\lloc$ with table identifier $\TBID$ is in accordance with what $\itfc$ prescribes.

The typing rule for the action to \emph{drop} a table does not require any constraint to be satisfied:
a finer analysis may be able to give an approximation of whether the table currently exists; 
for simplicity we stick to the scope of errors listed in the beginning of this section. 

Finally, the typing rule for the action $\EVAL{P}{\lloc}$ type checks $P$ 
in the same environment $\Gamma$ as for $\EVAL{P}{\lloc}$, 
which corresponds to the static scoping of the free variables of $P$. 

For examples on typing actions we refer the reader to Section~\ref{sec:case_study}. 

\begin{figure}[ht!]
\caption{The Typing Rules for Processes, Components, and Nets}\label{fig:tp_pcn}
\hrule
\centering
\small
\begin{tabular}{c}{~}\\[-1ex]
\hspace{-12cm}\framebox[1.2\width]{Processes}
\\\\[-2ex]
$
\jdgtP{\Gamma,\itfc}{\NILP}
$
\qquad
$
\Inference{
	\jdgta{\Gamma,\itfc}{a}{\Gamma'}\quad \jdgtP{(\ext{\Gamma}{\Gamma'}),\itfc}{P}
}
{
	\jdgtP{\Gamma,\itfc}{a.P}
}
$
\qquad
$
\Inference{
	\jdgtP{\Gamma,\itfc}{P_1}\quad \jdgtP{\Gamma,\itfc}{P_2}
}
{
	\jdgtP{\Gamma,\itfc}{P_1;P_2} 
}
$
\\\\[-1ex]
$
\Inference{
	\begin{array}{c}
		~~\jdgtt{\Gamma}{e_1}{\tau_1}~~\ldots~~\jdgtt{\Gamma}{e_n}{\tau_n}~~ \\
		\jdgtP{(\avar_1:\tau_1,...,\avar_n:\tau_n),\itfc}{P}
	\end{array}
}
{
	\jdgtP{\Gamma,\itfc}{A(\LST{e})}
}~~\mrm{where}~ A(\LSTStruct{\avar_i:\tau_i}{i\le n})\triangleq P
$
\\\\[-1ex]
$
\Inference{
	\jdgtE{\Gamma,\itfc}{\CALTB}{\sktp}\quad \jdgtT{\sktp}{T}{\Gamma'}\quad \jdgtpred{\ext{\Gamma}{\Gamma'}}{\psi} \quad 
	\jdgtP{(\ext{\Gamma}{\Gamma'}),\itfc}{P}
}
{\jdgtP{\Gamma,\itfc}{\FOREACH{\CALTB}{T}{\psi}{\ordr}{P}}}
$
\\\\[-1ex]
\hspace{-12cm}\framebox[1.2\width]{Components}
\\\\[-1ex]
$
\Inference{\forall t\in R: \jdgtt{\Gamma}{t}{I.\TBSK}}{\jdgtC{\Gamma,\itfc}{(I,R)}}
\quad\mrm{if}~\itfc(I.\TBID)=I.\TBSK
$
\qquad
$
\Inference{
	\jdgtC{\Gamma,\itfc}{C_1}\quad \jdgtC{\Gamma,\itfc}{C_2}
}
{
	\jdgtC{\Gamma,\itfc}{C_1|C_2}
}
$
\\ \\[-1ex]
\hspace{-12cm}\framebox[1.2\width]{Nets}
\\[-1ex]
$\jdgtC{\Gamma,\itfc}{\NILN}$
\qquad
$
\Inference{
	\jdgtt{\Gamma}{\llocConst}{\loctp}\quad \jdgtC{\Gamma,\itfc}{C}
}
{
	\jdgtN{\Gamma,\itfc}{\LOCATED{\llocConst}{}{C}}
}
$
\\\\[-1ex]
$
\Inference{
  \jdgtN{\Gamma,\itfc}{N}
}
{
  \jdgtN{\Gamma,\itfc}{(\nu \llocConst)N}
}
$
\qquad
$
\Inference{
	\jdgtN{\Gamma,\itfc}{N_1}\quad \jdgtN{\Gamma,\itfc}{N_2} 
}
{
	\jdgtN{\Gamma,\itfc}{N_1||N_2} 
}
$
\\\\[-1ex]
\end{tabular}
\normalsize
\hrule
\end{figure}

\subsubsection*{Typing Processes, Components and Nets}
Directing our attention to the typing rules for processes in Figure~\ref{fig:tp_pcn}, 
the rule for the \emph{inert process} $\NILP$ is trivial. 
The rule for \emph{prefixing} says that in order for $a.P$ to be typed, 
we need to type the action $a$, deriving some environment $\Gamma'$, 
and the process $P$ needs to be well-typed under the original typing environment $\Gamma$ 
extended with $\Gamma'$. 
On the other hand, to type a \emph{sequential composition} $P_1;P_2$, 
each process only needs to be typed in the original environment $\Gamma$, 
which is because of the local scoping of the variables bound in $P_1$. 
To type a \emph{procedure call} $A(\LST{e})$ where $A$ has the procedure body $P$, 
the process $P$ needs to be typed in the typing environment that 
binds all the formal parameters of $A$ to their types obtained from the signature of $A$. 
Naturally, these types also need to agree with those of the actual arguments in $\LST{e}$. 
The rule for the {\sf foreach} \emph{loop} features a similar pattern in the treatment of 
the source table $\CALTB$, the template $T$ and the predicate $\psi$, 
compared to that of certain typing rules for actions such as ${\sf select}$.
In addition, the body $P$ of the loop needs to be typed with the original typing environment $\Gamma$ extended with $\Gamma'$ that reflects the bindings produced by 
matching tuples from the source table against $T$. 

The next group of typing rules in Figure~\ref{fig:tp_pcn} covers the typing of components. 
It is worth noting that the rule for typing a \emph{concrete table} $(I,R)$ as a component resembles 
the one for typing it as a table,  
except that the case where the table identifier is $\bot$ is impossible and 
need not be dealt with for a table that stands as a component. 
Note also that the next rule for typing a \emph{parallel composition} $C_1|C_2$ of components 
uses the same environments $\Gamma$ and $\itfc$ for the components $C_1$ and $C_2$ as for $C_1|C_2$: 
since there are no \emph{linearity} issues of concern for the variables, 
no split operation is needed on $\Gamma$, and $\itfc$ is simply a global constant. 

The last group of typing rules in Figure~\ref{fig:tp_pcn} covers the typing of nets. 
The rule for the \emph{empty net} $\NILN$ is trivial. 
The rule for a \emph{located component} $l::C$ checks that $l$ is indeed a locality 
($\jdgtE{\Gamma}{l}{\loctp}$). 
The rule for \emph{restriction} $\RESNET{l}{N}$ only ensures that the net $N$ is itself typable. 
This check is performed with the same typing environment $\Gamma$ since $l$ is not a variable 
and we have assumed that the types of constants are clear for their literal appearance. 
Finally, the rule for parallel composition $\PARNET{N_1}{N_2}$ is conceived with a similar rationale 
to that of the rule for $C_1|C_2$.

\subsection{Theoretical Results} % (fold)
\label{sec:theoretical_results}

It is usually required of type systems that they should be safe, in the sense of
subject reduction~\cite{Pierce02_Types} --- 
well-typedness is preserved under transitions.  
In most cases it is important for type systems to be sound, i.e., 
the guarantees they aim to deliver should indeed be delivered at run time, 
for well-typed specifications and programs. 
It is also important for type checking to be efficient, 
to smooth the development process. 
In this section, we provide theoretical results demonstrating that 
our type system has all of the aforementioned properties. 
In the statements of these results, we will leave universal quantifiers implicit, 
and assume that it is clear from the use of symbols what kinds of entities they refer to. 

The subject reduction result is given in Theorem~\ref{thm:subj_red}. 
Its proof is structured as an induction on the semantic derivation, 
with details provided in Appendix~\ref{app:proofs}. 

\begin{thm}[Subject Reduction]\label{thm:subj_red}
  For a closed net $N$, if $\BV{N}\cap\dom{\Gamma}=\emptyset$, 
  $\jdgtN{\Gamma,\itfc}{N}$, and $\vdash N\rightarrow N'$, then $\jdgtN{\Gamma,\itfc}{N'}$. 
\end{thm}

We characterize a net $N$ being non-erroneous as $\oknet(N)$, which says that
	one is unable to identify a parallel component that is $\ErrNet$ in $N$. 
The formal definition of $\oknet(N)$ is given below. 

\begin{defi}
$\oknet(N)\triangleq \lnot(\exists N': N\equiv N'||\ErrNet)$
\end{defi}

With a straightforward induction on the typing derivation, 
	we can show that if a net $N$ is typable, then it cannot be erroneous. 
\begin{lem}\label{thm:typable_ok}
If $\jdgtN{\Gamma,\itfc}{N}$, then $\oknet(N)$ holds. 
\end{lem}

Combining Theorem~\ref{thm:subj_red} and Lemma~\ref{thm:typable_ok}, 
we can easily obtain the following result that
the well-typedness of closed nets implies its non-erroneous execution. 
\begin{thm}[Soundness]
If $\jdgtN{\emptyTPEnv,\itfc}{N}$, and $\vdash N\rightarrow^* N'$, then $\oknet(N')$ holds. 
\end{thm}

To be able to discuss the time complexity of type checking, 
a characterization of the sizes of nets is needed. 
We characterize the size of a net $N$ as
the number of ASCII characters used for the full textual representation of $N$ itself, 
as well as the definitions of procedures invoked in $N$: 
\small
$$\size{N}=\sz{N}+\sum_{A(\LSTStruct{\avar_i:\tau_i}{i})\triangleq P} \sz{P}$$
\normalsize
By ``full textual representation'', it is meant that 
the contents of interfaces $I$ and data sets $R$ for tables need to be completely spelled out. 

The following theorem states that the time complexity of type checking a net $N$
is linear in the size of $N$. 
Its proof is sketched in Appendix~\ref{app:proofs}. 
\begin{thm}[Efficiency of Type Checking]\label{thm:efficiency}
With a given $\itfc$, 
the time complexity of type checking net $N$ is linear in $\mi{size}(N)$, provided that 
it takes $O(1)$ time to 
\begin{itemize}
  \item determine the types of all constant expressions, 
  \item decide the equality/inequality of table identifiers and types, 
  \item construct a singleton typing environment, 
  \item construct the extension $(\Gamma_1,\Gamma_2)$ of a typing environment $\Gamma_1$ with $\Gamma_2$, and
  \item look up environments $\Gamma$ (resp. $\itfc$) for variables (resp. table identifiers), 
\end{itemize}
and that it takes $O(n)$ time to 
\begin{itemize}
  \item form an $n$-ary product type, and 
  \item perform the operation $\flattensk(\type_1\times...\times\type_n)$.
\end{itemize}
\end{thm}

% subsubsection theoretical_results (end)

%%% Local Variables:
%%% mode: latex
%%% TeX-master: "paper"
%%% End:

%!TEX root = ./paper.tex

\section{Case Study}\label{sec:case_study}

Revisiting our scenario with the chain of department stores, 
we illustrate the modeling of data aggregation over multiple databases local to its different branches in Klaim-DB.
In more detail, a manager of the head office wants statistics on the total sales of KLD high boots from the year 2015, in each branch operating in Copenhagen.

We will think of a procedure $\mi{stat}$ at the locality $\ttloc_0$ of the head office, carrying out the aggregation needed. 
Thus the net $\dcnet$ for the database systems of the department store chain, 
as considered in Section~\ref{sec:syntax}, 
specializes to the following, 
where $C''_0$ is the remaining tables and processes at $\ttloc_0$ apart from 
the table {\sf Stores} and the procedure $\mi{stat}$. 
\small
$$
\ttloc_0:: ((I_0,R_0) | \mi{stat} | C''_0) ~||~ \ttloc_1:: ((I_1,R_1) | C'_1) 
  ~||~ ... ~||~ \ttloc_n::((I_n,R_n) | C'_n)
$$
\normalsize

A detailed specification of $\mi{stat}$ is then given in Figure~\ref{fig:case_study}. 
\begin{figure}[ht!]
  \caption{The Procedure for Distributed Data Aggregation}
  \label{fig:case_study}
  \centering
  \small
  $$
  \begin{array}{r@{~\!}l}
    \mi{stat}\triangleq & \CREATENEW{\restbid@\ttloc_0}{\strtp\times \strtp \times \inttp}. \\
    & \SEL{{\sf Stores}@{\ttloc_0}}{(!x,!y,!z,!w,!p)}{{\sf KLD}\in w\land x=``\mi{CPH}''}{(z,p)}{!\TBV}. \\
    & {\sf foreach}(\TBV,(!q,!u),\TRUE,\{\}): \\
    & \quad \AGGR{{\sf KLD}}{(!\mi{id},!\mi{tp},!\mi{yr},!\mi{cr},!\mi{sz},!\mi{is},!\mi{ss})}
            {\mi{tp}=``\mi{HB}''}{\SUM_7}{!res}{u}. \\
    & \quad \INS{\restbid}{(q, ``\mi{HB}'', res)}{\ttloc_0}. \\
    & \quad \NILP; \\
    & ... \\
    & \DROP{\restbid}{\ttloc_0}.\NILP 
  \end{array}
  $$
  \normalsize
\end{figure}

First of all, a result table with identifier ${\sf SSResult}$ 
and schema $\strtp\times \strtp\times \inttp$ is created.
Then all the localities of the databases used by the branches in Copenhagen 
that actually sell KLD shoes are selected, together with the shop names of such branches. 
This result set is then processed by a {\sf foreach} loop. 
The number of KLD high boots from 2015 that are sold is counted at each of these localities (branches), 
and is inserted into the resulting table together with 
the corresponding shop name and the string ``HB'' describing the shoe type concerned. 
The resulting table, displayed in Figure~\ref{tab:Result_table}, can still be queried/manipulated before being dropped. 

\begin{figure}[ht!]
  \caption{The Table {\small \restbid \normalsize}}\label{tab:Result_table}
  \centering
\begin{tabular}{|c|c|c|}
  \hline
  ~$~\mi{Shop\_name}~$~ & ~$~\mi{Shoe\_type}~$~ & ~$~\mi{Sales}~$~\\
  \hline\hline
  ~Shop1~ & ~HB~  & ~12~ \\
  \hline
  ~Shop2~ & ~HB~  & ~53~ \\
  \hline
  ~Shop3~ & ~HB~  & ~3~ \\
  \hline
  ... & ... & ... \\
\end{tabular}
\end{figure}

The typability of $\dcnet$ is formally stated in Fact~\ref{fct:tpdcnet}. 
\begin{fact}\label{fct:tpdcnet}
  Suppose 
  \small
  $$
  \begin{array}{r@{~\!}l}
  \dcnab= & [{\sf Stores}\mapsto \strtp\times\strtp\times\strtp\times(\msett{\idtp})\times\loctp] \\
          & [{\sf KLD}\mapsto \strtp\times\strtp\times\strtp\times\strtp\times\strtp\times\inttp\times\inttp] \\
          & [\restbid\mapsto \strtp\times\strtp\times\inttp]. 
  \end{array}
  $$
  \normalsize
  We have $\jdgtN{\emptyTPEnv,\dcnab}{\dcnet}$, 
    given $\forall i\in\{1,...,n\}: \jdgtC{\emptyTPEnv,\dcnab}{C'_i}$ and $\jdgtC{\emptyTPEnv,\dcnab}{C''_0}$. 
\end{fact}

\begin{exa}
We briefly go through the typing of $\dcnet$. 

It can be established for all $j\in\{0,1,...,n\}$ that $\jdgtC{\emptyTPEnv,\dcnab}{(I_j,R_j)}$.
First of all, with the schemas of these tables given at the end of Section~\ref{sec:syntax}, 
we have $\dcnab(I_j.\TBID)=I_j.\TBSK$. 
Second of all, by the assumption made in the same section that 
the data stored in these tables have the types specified in the schemas, 
we have $\forall t\in R_j:\jdgtt{\emptyTPEnv}{t}{I_j.\TBSK}$ for all $j$ in range. 

It is not difficult to see that under 
the assumed well-typedness of the components $C'_1$, ..., $C'_n$ and $C''_0$, 
it boils down to establishing $\jdgtP{\emptyTPEnv,\dcnab}{\mi{stat}}$, 
in order to obtain $\jdgtN{\emptyTPEnv,\dcnab}{\dcnet}$. 

For the action creating the result table we have 
\small
$$
\Inference
{\jdgtE{\emptyTPEnv}{\ttloc_0}{\loctp}}
{\jdgta{\emptyTPEnv,\dcnab}{\CREATENEW{{\restbid}@\ttloc_0}{\strtp\times\strtp\times\inttp}}{\emptyTPEnv}} 
$$
\normalsize
since it holds that \small $\strtp\times\strtp\times\inttp=\dcnab(\restbid)$. \normalsize

For the action selecting the pairs of shop names and localities, we have 
\small
$$
\Inference
{
  \begin{array}{l}
  \jdgtt{\emptyTPEnv,\dcnab}{{\sf Stores}@\ttloc_0}
  {\strtp\times\strtp\times\strtp\times(\msett{\idtp})\times\loctp} \\
  \strtp\times\strtp\times\strtp\times(\msett{\idtp})\times\loctp\vdash (!x,!y,!z,!w,!p) \smalltriangleright \\
  \qquad\qquad\qquad\qquad\qquad\qquad\qquad [x:\strtp,y:\strtp,z:\strtp,w:\msett{\idtp},p:\loctp] \\
  \jdgtpred{[x:\strtp,y:\strtp,z:\strtp,w:\msett{\idtp},p:\loctp]}{{\sf KLD}\in w \land x=``CPH"} \\
  \jdgtt{[x:\strtp,y:\strtp,z:\strtp,w:\msett{\idtp},p:\loctp]}{(z,p)}{\strtp\times\loctp} 
  \end{array}
}
{
  \begin{array}{l}
  {\emptyTPEnv,\dcnab}\vdash
  {\SEL{{\sf Stores}@\ttloc_0}{(!x,!y,!z,!w,!p)}{{\sf KLD}\in w\land x=``CPH''}{(z,p)}{!\TBV}}
  \smalltriangleright\\
  \qquad\qquad\quad\! [\TBV:\strtp\times\loctp]
  \end{array}
}
$$
\normalsize\smallskip

For the aggregation action producing the sales figures, we have 
\small
$$
\Inference
{
  \begin{array}{l}
    \dcnab({\sf KLD})\vdash (!id,!tp,!yr,!cr,!sz,!is,!ss) \smalltriangleright \\
      \qquad\quad [id:\strtp,tp:\strtp,yr:\strtp,cr:\strtp,sz:\strtp,is:\inttp,ss:\inttp] \\
    {[}\TBV:\strtp\times\loctp,q:\strtp,u:\loctp,id:\strtp,tp:\strtp, \\
      \qquad\quad yr:\strtp,cr:\strtp,sz:\strtp,is:\inttp,ss:\inttp] \vdash 
                          tp=``HB"\smalltriangleright\booltp \\
    \mi{sum}_7:\msett{(\strtp\times\strtp\times\strtp\times\strtp\times\strtp\times\inttp\times\inttp)}\rightarrow\inttp \\
    \jdgtT{\inttp}{!\mi{res}}{[\mi{res}:\inttp]} \qquad
    \jdgtE{[\TBV:\strtp\times\loctp,q:\strtp,u:\loctp]}{u}{\loctp}
  \end{array}
}
{
  \begin{array}{l}
  [\TBV:\strtp\times\loctp,q:\strtp,u:\loctp], \dcnab \vdash \\
  \qquad\qquad\AGGR{{\sf KLD}}{(!id,!tp,!yr,!cr,!sz,!is,!ss)}{tp=``HB"}{\mi{sum}_7}{!res}{u}\smalltriangleright[\mi{res}:\inttp]
  \end{array}
}
$$
\normalsize\smallskip

For the action inserting the aggregation result, we have 
\small
$$
\Inference
{
  \jdgtt{[\TBV:\strtp\times\loctp,q:\strtp,u:\loctp,\mi{res}:\inttp]}
        {(q,``HB",\mi{res})}{\strtp\times\strtp\times\inttp} \\
  \jdgtE{[\TBV:\strtp\times\loctp,q:\strtp,u:\loctp,\mi{res}:\inttp]}{\ttloc_0}{\loctp}
}
{
  \begin{array}{l}
  {[\TBV:\strtp\times\loctp,q:\strtp,u:\loctp,\mi{res}:\inttp],\dcnab}\vdash \\
  \qquad\qquad\qquad\qquad\qquad\qquad\qquad\qquad
  {\INS{\restbid}{(q,``HB",\mi{res})}{\ttloc_0}}\smalltriangleright {\emptyTPEnv}
  \end{array}
}
$$
\normalsize
since it holds that \small $\strtp\times\strtp\times\inttp=\dcnab(\restbid)$. \normalsize

We can thus establish 
\small
$$
\jdgtP{[\TBV:\strtp\times\loctp,q:\strtp,u:\loctp],\dcnab}{\UL{\sf aggr}.\UL{\sf ins}.\NILP}
$$
\normalsize
where the underlined names abbreviate the actions under the same names in the procedure $\mi{stat}$. 

For the {\sf foreach} construct we have 
\small
$$
\Inference
{
  \begin{array}{c}
    \jdgtt{[\TBV:\strtp\times\loctp],\dcnab}{\TBV}{\strtp\times\loctp} \\
    \jdgtT{\strtp\times\loctp}{(!q,!u)}{[q:\strtp,u:\loctp]} \\
    \jdgtpred{[\TBV:\strtp\times\loctp,q:\strtp,u:\loctp]}{\TRUE} \\
    \jdgtP{[\TBV:\strtp\times\loctp, q:\strtp, u:\loctp], \dcnab}{\UL{\sf aggr}.\UL{\sf ins}.\NILP}
  \end{array}
}
{
  \jdgtP{[\TBV:\strtp\times\loctp], \dcnab}{\FOREACH{\TBV}{(!q,!u)}{\TRUE}{\emptyset}
  {\UL{\sf aggr}.\UL{\sf ins}.\NILP}}
}
$$
\normalsize\smallskip

The typing of the action dropping the result table can be easily established, i.e., 
\small
$$
\Inference
{\jdgtE{\emptyTPEnv}{\ttloc_0}{\loctp}}
{\jdgta{\emptyTPEnv,\dcnab}{\DROP{\restbid}{\ttloc_0}}{\emptyTPEnv}}
$$
\normalsize

Using the typing rule for prefixing repeatedly, and the rule for sequential composition, we can obtain:
\small
$$
\jdgtP{\emptyTPEnv,\dcnab}{\UL{\sf create}.\UL{\sf select}.\UL{\sf foreach};\UL{\sf drop}}
$$
\normalsize\smallskip

Finally, for the procedure call to $\mi{stat}$ passing no arguments, we have 

\small
$$
\Inference
{
  \jdgtP{\emptyTPEnv,\dcnab}{\UL{\sf create}.\UL{\sf select}.\UL{\sf foreach};\UL{\sf drop}}
}
{
  \jdgtP{\emptyTPEnv,\dcnab}{\mi{stat}}
}\eqno{\qEd}
$$
\normalsize
%\qed
\end{exa}

\begin{rem}
The statistics-gathering task performed by the procedure $\mi{stat}$ of Figure~\ref{fig:case_study}
can also be accomplished using \emph{code mobility} combined with local aggregations, 
instead of a (sequential) series of remote aggregations.
In more detail, the head office can maintain a topology of all the branches, 
and spawn multiple agents that are in charge of different subsets of the localities 
$\ttloc_1$, ..., $\ttloc_n$. 
These agents then travel from branch to branch according to the topology  
(potentially stored in tables) of which they are informed, 
aggregating the sales figures from each branch they arrive at, 
that does indeed sell shoes of the brand KLD. 
The agents can either insert the result obtained from each locality remotely into 
a designated result table at $\ttloc_0$, 
or return to $\ttloc_0$ in the end, when their results are gathered.
\end{rem}

%%% Local Variables:
%%% mode: latex
%%% TeX-master: "paper"
%%% End:

% \input{extension.tex}
%!TEX root = ./paper.tex

\section{Discussion and Related Work}\label{sec:related}

This section provides a discussion of our work around several topics, where we provide additional motivations for our design choices, related to similar works and argue about possible future developments.

\paragraph*{\bf Tuple-based coordination languages.}
%There are relatively few pieces of work closely related to ours. 
%
Our work studies what tuple-based coordination languages have to offer for distributed database systems. Our main sources of inspiration have been the family of such languages which includes Linda~\cite{Gelernter85_Generative}, Klaim~\cite{Nicola98_Klaim}, and SCEL, to mention a few.
The suitability of Linda~\cite{Gelernter85_Generative} to construct 
a software framework for distributed database systems was investigated in~\cite{Thirukonda2002_Linda}.   
Their approach, however, was not concerned with formal language design. 
Instead, the authors focus on a low-level representation of individual records in tables as tuples in the Linda tuple spaces, 
and propose the strategy of translating SQL queries into basic Linda operations such as output and input. 
Correspondingly, they considered some fine-grained issues such as data replication, concurrency control, 
and fault tolerance. 

A detailed discussion with respect to Klaim has been already provided in the Introduction.
We would like to mention here an implementation-oriented extension of Klaim, namely X-Klaim~\cite{BDL06_XKlaim}. 
X-Klaim is a full-fledged programming language targeting mobile applications.
It provides language primitives for stronger forms of mobility --- 
migrating not only code but also the current execution state. 
Later in this section, we will envision an implementation of Klaim-DB that builds on the X-Klaim implementation. 

The Service Component Ensemble Language (SCEL)~\cite{NicolaLPT14_SCEL,DBLP:series/lncs/NicolaLLLMMMPTV15} can also be seen as an evolution of Klaim, aimed at 
facilitating the programming of software component ensembles
whose coordination patterns can adapt to environmental conditions.
A salient feature of SCEL is that components publish their own attributes through interfaces
and groups of components can be addressed in basic actions using predicates over these attributes in what the authors call \emph{attribute-based communication} (see also~\cite{DBLP:conf/forte/AlrahmanNL16}). 
The publication of table identifiers and schemas by interfaces of tables in Klaim-DB 
is analogous to the publication of attributes by interfaces of components in SCEL, although in our work we did not exploit such interface to provide richer attribute-based primitives. Further developments of Klaim-DB could indeed incorporate programming abstractions that may allow to operate on sets of databases without referring them by name but by their properties.

\paragraph*{\bf Interoperability.} 
A closely related topic that we considered during the design of Klaim-DB is that of interoperability among multiple databases and/or user applications (possibly Klaim-based). In this regard, there are two aspects that we would like to mention. First, concerning homogeneous interoperability among multiple databases, we think that the support of join operations in selection actions paves the way for the general ability to operate on multiple databases by a single action, 
which is in line with the design philosophy of multi-database systems~(e.g., \cite{Kuhn2000_MultiDB}). Second, heterogeneous interoperability between database systems and user applications that are not data-centric 
can also be realized by bringing back the original Klaim primitives 
(such as ${\sf out}(t)@\ell$ and ${\sf in}(T)@\ell$) and 
allowing the co-existence of tables and plain tuples. 

\paragraph*{\bf Pattern matching and predicates.}
Our preliminary work~\cite{WuLLNN15_KlaimDB} considered templates featuring both formal and actual fields and was in line with Klaim templates and the database query language QBE (Query by Example~\cite{Zloof75_QBE}). 
In the present work templates are restricted to formal fields only, since the combined power of pattern matching with such templates and predicates provides a more powerful mechanism with respect to the bare pattern matching used, for instance, in Linda and Klaim variants (e.g. it is not possible to match tuples with distinct values in the fields). Moreover, \emph{conditional} pattern matching is present in database query languages like SQL (through the \texttt{WHERE} clause). 

\paragraph*{\bf Multiset semantics and temporary tables}
Our use of multiset operations for the semantics of select actions 
has the flavor of the Domain Relational Calculus underlying QBE~\cite{Zloof75_QBE}.  
Concerning the treatment of the result of such operations, 
an alternative that we initially considered but did not adopt is the direct placement of the result in a separate table. 
This table would either be created automatically by the selection operation itself, 
with an identifier specified in the selection action, 
or a designated ``result table'' --- one at each locality. 
However, a problem with this option is that the removal of the automatically created tables would need to be taken care of 
by the system designer making the specification, using $\DROP{\TBID}{\ell}$ actions. 
Similar problems would arise with the maintenance of the designated ``result tables'' 
(e.g., the alteration of its schema, the cleaning of old results, etc.). 
To abstain from these low-level considerations, table variables were finally introduced and binding is used for the selection results.

\paragraph*{\bf Formal approaches to query languages.}
%The central aim of Klaim-DB is to support the structure and operations of distributed database systems.
In Klaim-DB, the DB-oriented actions correspond tightly to SQL-like queries. 
In the literature, formalizations of the relational database model and SQL exist. 
One piece of work along this direction is \cite{Malecha2010_Toward}. 
In \cite{Malecha2010_Toward}, the syntax and semantics of two different query formalisms, 
\emph{relational algebra} and \emph{conjunctive queries}, are specified, 
some existing methods of logical query optimization are dealt with, 
and the soundness and/or completeness of two existing procedures for inferring integrity constraints
are proved. 
Another recent development, \cite{Benzaken2014_CoqRelational}, focuses on a certified realization of 
core SQL operations. 
The relational algebra is used to define the semantics of certain key SQL operations. 
Some source-to-source optimizations of SQL queries are proven to be semantically preserving 
and crude accounts of correspondences in running times are given. 
Lastly, an efficient implementation of the SQL model considered is performed with the Ynot extension to Coq,  
and is proved to satisfy its abstract specification. 
The main differences with respect to our work is that~\cite{Malecha2010_Toward} and~\cite{Benzaken2014_CoqRelational} do not 
consider \emph{distributed} databases, and are not concerned with the kind of properties our type system deals with. 

\paragraph{\textbf{Security.}} Security is an important concern, especially in a distributed setting. 
A \emph{basic step} towards security is that of preventing the exploitation of \emph{unsafe} language operations by a malicious user.  
An example of this kind with programming languages is unsafe C library functions giving rise to buffer overflow exploits.
Providing safe operations needs to be considered along with the core language design.
In our case this amounts to catching the insertion of mis-formatted records or the use of undefined variables, to mention a few, 
in the formal semantics of Klaim-DB and its type system. 
A \emph{further step} towards security is the provision of abilities for access control and information flow control. 
For instance, when inserting data into a remote table, it is important to know 
whether the current locality trusts the remote locality with respect to confidentiality, 
and whether the remote locality trusts the current locality with respect to integrity. 
In addition, information should not be directly or indirectly leaked to or influenced from, unauthorized localities. 
We expect that access control mechanisms can be devised for Klaim-DB using security type systems or flow logics 
following the paths taken by \cite{Nicola98_Klaim} and \cite{NicolaGHNNPP10_FLTP}. 
On the other hand, information flow mechanisms can be devised based on the insights in 
\cite{TolstrupNH06_Locality-Based}, which enforces locality-based information flow policies in Klaim, and 
\cite{Lourenco13_Compartments}, which is concerned with information flow security for SQL-like language primitives.

\paragraph{\textbf{Implementation.}}
While our focus was on developing a modelling language,
it is worth discussing how the Klaim-DB could be turned into a domain-specific programming language
following the inspiration of Klaim and its counterpart X-Klaim. 
As a matter of fact, we envision an implementation of the Klaim-DB language, 
building on the implementation~\cite{BDL06_XKlaim} of X-Klaim.  
The latter consists of a compiler that translates X-Klaim programs into Java programs
that make use of the Klava package~\cite{BDP02_Klava} for tuple-based operations. 

An important part of the implementation would be the Klaim-DB compiler, 
which would extend the X-Klaim compiler and turn a Klaim-DB specification into a Java program. 
The Java program would ``execute'' the specification; 
hence it would perform operations on \emph{tables} according to the Klaim-DB semantics. 
The Klava package would need to be extended to support these operations. 
For instance, classes \texttt{Table}, \texttt{Schema}, etc., 
would need to be created to represent tables, schemas, etc., and 
the methods of these classes would correspond to the table operations. 
The side conditions of the semantic rules can be easily checked by generating appropriate Java code. 
Because of the correspondence between the DB-oriented actions in Klaim-DB and SQL queries
described in Section~\ref{sec:syntax}, 
we would produce the query results by
constructing SQL objects and queries via a Java API for database connectivity such as JDBC.
Hence classes supporting the conversion between objects of the classes \texttt{Table}, \texttt{Schema}, etc., 
and SQL objects would be included in the extended Klava package.  
Since a SQL query is not always executed as an atomic operation, 
the locking mechanism of Java would be needed. 
To sum up, the Java code to be generated
for executing a DB-oriented action would accomplish the following: 
\begin{itemize}
	\item Lock the table objects involved in the action; 
	\item Check the side conditions specified in the semantics;
	\item Convert table objects involved in the action to SQL objects;
	\item Perform the query by interfacing with a SQL engine via JDBC;
	\item Convert the results back to internal objects;
	\item Unlock relevant objects.
\end{itemize}

The type checking of Klaim-DB specifications, on the other hand, 
can be performed directly on the AST (Abstract Syntax Tree)
obtained by parsing the Klaim-DB specification in the Klaim-DB compiler.

\paragraph{\textbf{Transactions.}}
In supporting database operations, 
the Klaim-DB language is positioned between succinct core calculi such as the relational algebra, 
and full-fledged query languages such as SQL, 
to facilitate both the formalization\footnote{Note that formalizations of SQL exist but often do not cover transactions~\cite{Gogolla94_FormalSQL,Negri91_FormalSQL}.} of its semantics, typing, etc., and 
its utilization in modeling tasks. 
Corresponding to its level of abstraction, 
\emph{high-level atomic operations} are usually used in Klaim-DB to achieve tasks
that need to be accomplished with transactions in other languages such as SQL. 
Nevertheless, Klaim-DB can be extended in the future with explicit transactions, 
to support use cases such as where a user creates a table, 
but does not want other users to have access to it before she has inserted all her desired records into it. 
In general, an extension to provide full support of transactions needs not only to provide
appropriate locking mechanisms (as language primitives), 
but also to support \emph{rollbacks} to recover from inconsistent states.

%%% Local Variables:
%%% mode: latex
%%% TeX-master: "paper"
%%% End:

%!TEX root = ./paper.tex

\section{Conclusion}\label{sec:conclusion}

In answer to the challenge of easily and correctly structuring, managing, and utilizing distributed databases, 
we have presented the design of a coordination language, Klaim-DB. 
The central notion of data in Klaim-DB is that of localized databases. 
Primitives for data definition and manipulation are provided at a high abstraction level, 
guarding against misuses of programming constructs and subtle concurrency errors in specifications. 
The correctness of data formats is checked in the semantics, 
and to a great extent statically guaranteed by a type system, 
which allows efficient type checking. 
We have modeled in the language a data aggregation task performed across geographically scattered databases, 
guided by a coordinator.

%%% Local Variables:
%%% mode: latex
%%% TeX-master: "paper"
%%% End:

\section*{Acknowledgement}
We would like to thank the anonymous reviewers for their useful comments.

\bibliographystyle{plain}
\bibliography{references}

\appendix
%!TEX root = ./paper.tex

\section{Multiset Notation}\label{app:multiset}

We use $\uplus$, $\cap$ and $\setminus$ to represent the union, intersection and substraction, respectively, of multisets. 
For a multiset $S$, and an element $s$, the application $M(S,s)$ of the multiplicity function $M$ gives the number of repetitions of $s$ in $S$. 
Note that for $s\not\in S$, $M(S, s)=0$. 
Then our notions of union, intersection and subtraction are such that
$$
\begin{array}{l}
  M(S_1\uplus S_2,s)= M(S_1,s)+M(S_2,s) \\
  M(S_1\cap S_2, s)= \mathit{min}(M(S_1,s), M(S_2,s)) \\
  M(S_1\setminus S_2, s)= \mi{max}(M(S_1,s)-M(S_2,s),0)
\end{array}
$$

\section{Proofs}\label{app:proofs}

\newcounter{mycounter}  
\newenvironment{mylist1}
{\begin{list}{\arabic{mycounter}.~~}{\usecounter{mycounter} \labelsep=0em \labelwidth=0em \leftmargin=1em \parsep=.5em \itemindent=0em \labelsep=5pt}}
  {\end{list}}

\newenvironment{proofsketch}{%
  \renewcommand{\proofname}{Sketch of Proof}\proof}{\endproof}

\newenvironment{proofidea}{%
  \renewcommand{\proofname}{Proof Idea}\proof}{\endproof}

We first restate Lemma~\ref{lem:no_rep} and give its proof. 

\begingroup
\def\thethm{\ref{lem:no_rep}}
\begin{lem}
If $\NOREP{\LIDN{N}\uplus \LIDN{\envnet}}$ and $\envnet\vdash\TRANS{N}{N'}$, 
then it holds that $\NOREP{\LIDN{N'}\uplus \LIDN{\envnet}}$.
\end{lem}
\addtocounter{thm}{-1}
\endgroup

\begin{proof}

Assume that $\NOREP{\LIDN{N}\uplus\LIDN{\envnet}}$ holds. 
We proceed with an induction on the derivation of $\envnet\vdash\TRANS{N}{N'}$ 
  to establish $\NOREP{\LIDN{N'}\uplus\LIDN{\envnet}}$. 

\begin{mylist1}
    
\item[\CASE $\REDINS$, \CASE $\REDDEL$, \CASE $\REDUPD$, and \CASE $\REDAGGR$:]{~}\\
  we have $\LIDN{N}=\LIDN{N'}=\{(l_2,I.\TBID)\}$ if $N'\neq\ErrNet$, 
  and $\LIDN{N'}=\emptyset$ if $N'=\ErrNet$. 
  Hence $\NOREP{\LIDN{N'}\uplus\LIDN{\envnet}}$ holds. 

\item[\CASE $\REDSEL$:]
  We have $\LIDN{N}=\LIDN{N'}=\{(l_1,I_1.\TBID),...,(l_k,I_k.\TBID)\}$ if $N'\neq\ErrNet$, 
  and $\LIDN{N'}=\emptyset$ if $N'=\ErrNet$. 
  Hence $\NOREP{\LIDN{N'}\uplus\LIDN{\envnet}}$ holds. 

\item[\CASE $\REDDROP$, \CASE $\REDEVALP$, \CASE $\REDFORTT$, \CASE $\REDFORFF$, and \CASE $(\mrm{CALL})$:]
  {~}\\
  We have $\LIDN{N'}=\emptyset$.
  By $\NOREP{\LIDN{N}\uplus\LIDN{\envnet}}$, we have $\NOREP{\LIDN{\envnet}}$. 
  Hence $\NOREP{\LIDN{N'}\uplus\LIDN{\envnet}}$ holds. 

\item[\CASE $\REDCREATE$:] 
  If $(l_2,\TBID)\not\in\LIDN{\envnet}$, 
  then since $\LIDN{N'}$ is the multi-set whose only element is $(l_2,\TBID)$, 
  with multiplicity $1$, we have $\NOREP{\LIDN{N'}\uplus\LIDN{\envnet}}$. 
  If $(l_2,\TBID)\in\LIDN{\envnet}$, then $\LIDN{N'}=\emptyset$. 
  Hence $\NOREP{\LIDN{N'}\uplus\LIDN{\envnet}}$ also holds. 

\item[\CASE $\REDSEQTT$:]
  We have $\LIDN{\llocConst:: \SEQ{P'_1}{P_2}||N}=\LIDN{\llocConst:: \SEQ{P_1}{P_2}||N}$.

\item[\CASE $\REDSEQFF$:] Analogous to Case $\REDSEQTT$.

\item[\CASE $\REDPAR$:]
  The net $N$ is $\PARNET{N_1}{N_2}$ and $N'$ is $\PARNET{N'_1}{N_2}$. 
  We have $\NOREP{\LIDN{N_1||N_2}\uplus\LIDN{\NILN}}$, which implies $\NOREP{\LIDN{N_1}\uplus\LIDN{N_2}}$. 
  By the induction hypothesis for $N_2\vdash N_1\rightarrow N'_1$, we have $\NOREP{\LIDN{N'_1}\uplus\LIDN{N_2}}$. 
  Therefore we also have\\ $\NOREP{\LIDN{N'_1||N_2}\uplus\LIDN{\NILN}}$. 

\item [\CASE $\REDRES$:] Trivial by using the induction hypothesis.

\item[\CASE $\REDEQUIV$:] 
  If neither of the equivalences, $N_1\equiv N_2$ and $N_3\equiv N_4$, 
  is concerned with $\alpha$-renaming, 
  then it is obvious that $\LIDN{N_1}=\LIDN{N_2}$ and $\LIDN{N_3}=\LIDN{N_4}$. 
  Using the \IH it is straightforward to derive $\NOREP{\LIDN{N_4}\uplus\LIDN{\envnet}}$. 

  Suppose $\alpha$-renaming is involved in $N_1\equiv N_2$. 
  By the global assumption that bound localities should be distinctly named,
  no locality name can be bound in both $N_2$ and $\envnet$. 
  By $\NOREP{\LIDN{N_1}\uplus \LIDN{\envnet}}$, 
  we also have $\NOREP{\LIDN{N_2}\uplus\LIDN{\envnet}}$. 
  By the \IH we have $\NOREP{\LIDN{N_3}\uplus\LIDN{\envnet}}$. 
  By similar reasoning with $N_3\equiv N_4$ we can deduce $\NOREP{\LIDN{N_4}\uplus\LIDN{\envnet}}$. 

\end{mylist1}
This completes the proof. 
\end{proof}

We proceed with proving Theorem~\ref{thm:subj_red}. 

\begin{lem} 
  If $N_1\equiv_\alpha N_2$, then $\jdgtN{\Gamma,\itfc}{N_1}\Leftrightarrow \jdgtN{\Gamma,\itfc}{N_2}$. 
\end{lem}

\begin{lem}\label{lem:tp_equiv}
  If $N_1\equiv N_2$, then $\jdgtN{\Gamma,\itfc}{N_1}\Leftrightarrow \jdgtN{\Gamma,\itfc}{N_2}$. 
\end{lem}

\begin{lem}\label{lem:reordering}
Suppose at most one binding is present in $\Gamma$ for each variable in $\dom{\Gamma}$, and  
$\Gamma_1$ is a re-ordering of the bindings of $\Gamma$. 
The the following statements hold. 
\begin{enumerate}
  \item If $\jdgtE{\Gamma}{e}{\type}$, then we have $\jdgtE{\Gamma_1}{e}{\type}$. 
  \item If $\jdgtpred{\Gamma}{\PRED}$, then we have $\jdgtpred{\Gamma_1}{\PRED}$. 
  \item If $\jdgtT{\Gamma}{t}{\type}$, then we have $\jdgtT{\Gamma_1}{t}{\type}$. 
  \item If $\jdgtT{\Gamma,\itfc}{\CALTB}{\type}$, then we have $\jdgtT{\Gamma_1,\itfc}{\CALTB}{\type}$.
  \item
    If $\jdgta{\Gamma,\itfc}{a}{\Gamma'}$, and $\BV{a}\cap\dom{\Gamma}=\emptyset$,
    then $\jdgta{\Gamma_1,\itfc}{a}{\Gamma'}$, and at most one binding is present for each variable in $\Gamma'$,
    and{~}\\
    if $\jdgtP{\Gamma,\itfc}{P}$, and $\BV{P}\cap\dom{\Gamma}=\emptyset$, 
    then $\jdgtP{\Gamma_1,\itfc}{P}$. 
\end{enumerate}
\end{lem}

\begin{proofsketch}
The proofs of (1) -- (4) are straightforward. 
For the proof of (5), we proceed with an induction on the structure of actions and processes.
We only present a few representative cases. 
\begin{mylist1}
\item[\CASE $\SEL{\LST{\CALTB}}{T}{\psi}{t}{!\TBV}$:]
  Since $\jdgta{\Gamma,\itfc}{\SEL{\LST{\CALTB}}{T}{\psi}{t}{!\TBV}}{\Gamma'}$ we have
  $\jdgtT{\Gamma,\itfc}{\CALTB_1}{\type_1}$, ..., $\jdgtT{\Gamma,\itfc}{\CALTB_n}{\type_n}$
  for some $\type_1$, ..., $\type_n$ such that $\jdgtT{\type_1\times...\times\type_n}{T}{\Gamma''}$,
  for some $\Gamma''$ such that
  $\jdgtpred{\Gamma,\Gamma'}{\PRED}$ and $\jdgtt{\Gamma,\Gamma'}{t}{\type'}$ for some $\type'$ such that 
  $\Gamma'=[\TBV:\type']$.
  Hence at most one binding exists in $\Gamma'$ for the only variable $\TBV$ in $\dom{\Gamma'}$.
  
  Using (4) we have
  $\jdgtT{\Gamma_1,\itfc}{\CALTB_1}{\type_1}$, ..., $\jdgtT{\Gamma_1,\itfc}{\CALTB_n}{\type_n}$.
  By our assumption that templates are linear,
  at most one binding exists in $\Gamma''$ for each variable in $\dom{\Gamma''}$.
  By the conditions of (5), we have $\BV{T}\cap\dom{\Gamma}=\emptyset$.
  Hence it holds that $\dom{\Gamma''}\cap\dom{\Gamma}=\emptyset$.
  Thus at most one binding exists in $\Gamma,\Gamma''$ for each variable in its domain,
  and the same holds for $\Gamma_1,\Gamma''$,
  which is also a reordering of the bindings of $\Gamma,\Gamma''$. 
  By (2) and (3) we have $\jdgtpred{\Gamma_1,\Gamma''}{\PRED}$ and $\jdgtt{\Gamma_1,\Gamma''}{t}{\type'}$.
  We can now establish $\jdgta{\Gamma_1,\itfc}{\SEL{\LST{\CALTB}}{T}{\psi}{t}{!\TBV}}{\Gamma'}$. 
\item[\CASE $\EVAL{P'}{\lloc}$:]
  By $\jdgta{\Gamma,\itfc}{\EVAL{P'}{\lloc}}{\emptyTPEnv}$ we have $\jdgtP{\Gamma,\itfc}{P'}$ 
  and $\jdgtE{\Gamma}{\lloc}{\loctp}$ 
  We have $\BV{\EVAL{P'}{\lloc}}=\BV{P'}$. 
  Hence $\BV{P'}\cap\dom{\Gamma}=\emptyset$. 
  By the \IH we have $\jdgtP{\Gamma_1,\itfc}{P'}$. 
  By (1) we have $\jdgtE{\Gamma_1}{\lloc}{\loctp}$. 
  We can then establish $\jdgtP{\Gamma_1,\itfc}{\EVAL{P'}{\lloc}}$.
\item[\CASE $a'.P'$:]
  By $\jdgtP{\Gamma,\itfc}{a'.P'}$ we have
  $\jdgta{\Gamma,\itfc}{a'}{\Gamma''}$ for some $\Gamma''$ such that 
  $\jdgtP{(\Gamma,\Gamma''),\itfc}{P'}$ holds.
  By $\BV{a'.P'}\cap\dom{\Gamma}=\emptyset$ we have
  \begin{align}
    \label{eq:reord_pref}
    \BV{a'}\cap\dom{\Gamma}= & \emptyset \\
    \BV{P'}\cap\dom{\Gamma}= & \emptyset 
  \end{align}

  Hence by the \IH we have $\jdgta{\Gamma_1,\itfc}{a'}{\Gamma''}$, and
  at most one binding exists in $\Gamma''$ for each variable in its domain. 
  Since $\BV{a'}\cap\dom{\Gamma}=\emptyset$,
  it is obvious that $\dom{\Gamma''}\cap\dom{\Gamma}=\emptyset$, and
  $\dom{\Gamma''}\cap\dom{\Gamma_1}=\emptyset$. 
  It is not difficult to see that at most one binding exists in $\Gamma_1,\Gamma''$
  for each variable in $\dom{\Gamma_1,\Gamma''}$.
  It is also obvious that $\Gamma_1,\Gamma''$ is a reordering of the bindings of $\Gamma,\Gamma''$. 
  Since variables are renamed apart in processes, we have $\BV{P'}\cap\BV{a'}=\emptyset$.
  Hence $\BV{P'}\cap\dom{\Gamma''}=\emptyset$.
  By \eqref{eq:reord_pref} we have $\BV{P'}\cap\dom{\Gamma,\Gamma''}=\emptyset$;
  hence $\BV{P'}\cap\dom{\Gamma_1,\Gamma''}=\emptyset$.
  By the \IH we have $\jdgtP{(\Gamma_1,\Gamma''),\itfc}{P'}$.
  We can now establish $\jdgtP{\Gamma_1,\itfc}{a'.P'}$. \qedhere
\end{mylist1}
\end{proofsketch}

\begin{lem}\label{lem:fv_in_dom} The following statements hold.
  \begin{enumerate}
  \item If $\jdgtE{\Gamma}{e}{\type}$, then $\FV{e}\subseteq \dom{\Gamma}$. 
  \item If $\jdgtE{\Gamma,\itfc}{\CALTB}{\type}$, then $\FV{\CALTB}\subseteq\dom{\Gamma}$.
  \item If $\jdgtpred{\Gamma}{\PRED}$, then $\FV{\PRED}\subseteq\dom{\Gamma}$. 
  \item If $\jdgtT{\Gamma}{t}{\type}$, then $\FV{t}\subseteq\dom{\Gamma}$. 
  \item If $\jdgta{\Gamma}{a}{\Gamma'}$ then $\FV{a}\subseteq\dom{\Gamma}$ 
    and $\BV{a}\supseteq \dom{\Gamma'}$, 
    and if $\jdgtP{\Gamma,\itfc}{P}$, then $\FV{P}\subseteq\dom{\Gamma}$. 
  \end{enumerate}
\end{lem}

The proof of Lemma~\ref{lem:fv_in_dom} is straightforward. 

\begin{lem}\label{lem:smaller_dom}
  The type environment can be shrunk as long as its domain still constains all the free variables
  of the objects to be typed. Formally, supposing $\Gamma'\subseteq\Gamma$, 
  and at most one binding is present in $\Gamma$ for each variable in $\dom{\Gamma}$, 
  we have the following results:
\begin{enumerate}
  \item If $\jdgtE{\Gamma}{e}{\type}$ and $\FV{e}\subseteq\dom{\Gamma'}$, then $\jdgtE{\Gamma'}{e}{\type}$. 
  \item If $\jdgtT{\Gamma,\itfc}{\CALTB}{\tau}$, and $\FV{\CALTB}\subseteq\dom{\Gamma'}$, 
        then $\jdgtT{\Gamma',\itfc}{\CALTB}{\tau}$.
  \item If $\jdgtpred{\Gamma}{\psi}$ and $\FV{\psi}\subseteq\dom{\Gamma'}$, then $\jdgtpred{\Gamma'}{\psi}$. 
  \item If $\jdgtt{\Gamma}{t}{\tau}$ and $\FV{t}\subseteq\dom{\Gamma'}$, then $\jdgtt{\Gamma'}{t}{\tau}$. 
\end{enumerate}
\end{lem}
\begin{proof}\hfill
  \begin{mylist1}
  \item[(1)] The proof is by induction on the structure of $e$,
    we only present the representative cases.
    \begin{mylist1}
    \item[\CASE $x$:] By $\jdgtE{\Gamma}{x}{\type}$ we have $\Gamma(x)=\tau$ and $\tau\neq\loctp$.
      Since $\Gamma'\subseteq\Gamma$, only one binding exists for $x$ in $\Gamma$, 
      and $\FV{x}\subseteq \dom{\Gamma'}$, it holds that $\Gamma'(x)=\tau$.
      Hence $\jdgtE{\Gamma'}{x}{\type}$ can be established. 
    \item[\CASE $\{e_1,...,e_n\}$:] By $\jdgtE{\Gamma}{\{e_1,...,e_n\}}{\type}$,
      there exists some $\type_0$ such that 
      $\type=\msett{\type_0}$, and
      $\jdgtE{\Gamma}{e_1}{\type_0}$, ..., $\jdgtE{\Gamma}{e_1}{\type_0}$.
      By $\FV{\{e_1,...,e_n\}}\subseteq \dom{\Gamma'}$ we have
      $\FV{e_1}\subseteq \dom{\Gamma'}$, ..., $\FV{e_n}\subseteq \dom{\Gamma'}$. 
      By the \IH we have $\jdgtE{\Gamma'}{e_1}{\type_0}$, ..., $\jdgtE{\Gamma'}{e_n}{\type_0}$.
      We thus have $\jdgtE{\Gamma'}{\{e_1,...,e_n\}}{\type}$. 
    \end{mylist1}

  \item[(2)] The proof is by a case analysis on $\CALTB$.
    \begin{mylist1}
    \item[\CASE $\TBID@\lloc$:] By $\jdgtE{\Gamma,\itfc}{\TBID@\lloc}{\type}$
      we have $\jdgtE{\Gamma}{\lloc}{\loctp}$ and $\type=\itfc(\TBID)$. 
      By $\FV{\CALTB}\subseteq\dom{\Gamma'}$ we have $\FV{\lloc}\subseteq\dom{\Gamma'}$. 
      Using (1) we have $\jdgtE{\Gamma'}{\lloc}{\loctp}$. 
      Hence $\jdgtE{\Gamma',\itfc}{\TBID@\lloc}{\type}$ can be established. 
    \item[\CASE $\TBV$:] By $\jdgtT{\Gamma,\itfc}{\TBV}{\tau}$ we have
      $\TBV\in\dom{\Gamma}$ and $\type=\Gamma(\TBV)$.
      Since $\Gamma'\subseteq \Gamma$, 
      only one binding exists for $\TBV$ in $\Gamma$, and
      $\FV{\TBV}\subseteq \dom{\Gamma'}$,
      we have $\TBV\in\dom{\Gamma'}$ and $\type=\Gamma'(\TBV)$.
      Hence $\jdgtT{\Gamma',\itfc}{\TBV}{\tau}$ can be established. 
    \item[\CASE $(I,R)$:] By $\jdgtT{\Gamma,\itfc}{\TBV}{\tau}$ we have
      $\forall t\in R: \jdgtt{\Gamma}{t}{I.\TBSK}$ and $I.\TBID\neq\bot\Rightarrow\itfc(I.\TBID)=I.\TBSK$.
      For each tuple $t$ in the data set $R$, we have $\FV{t}=\emptyset$ and
      it is obvious that the typing of $t$ is irrelevant to the particular typing environment in use.
      Hence we have $\forall t\in R: \jdgtt{\Gamma'}{t}{I.\TBSK}$.
      It can then be established that $\jdgtT{\Gamma',\itfc}{\TBV}{\tau}$.
    \end{mylist1}
  
  \item[(3)] The proof is straightforward by induction on the structure of $\psi$, using (1). 
  \item[(4)] The proof is trivial using (1). \qedhere
  \end{mylist1}
\end{proof}

\begin{lem}\label{lem:tp_eval} The following statements hold. 
\begin{enumerate}
  \item If $\jdgtt{\emptyTPEnv}{e}{\tau}$, then $\EVALT{e}\sattsk \tau$, and vice versa. 
  \item If $\jdgtt{\emptyTPEnv}{t}{\tau}$, then $\EVALT{t}\sattsk \tau$, and vice versa.  
  \item If $\jdgtpred{\emptyTPEnv}{\PRED}$, then $\EVALT{\PRED}\neq \Err$, and vice versa. 
\end{enumerate}
\end{lem}
\begin{proof}
  The proofs are by structural induction. 
  We present only representative cases. 
  \begin{mylist1}
  \item[(1a)] We prove if $\jdgtt{\emptyTPEnv}{e}{\tau}$, then $\EVALT{e}\sattsk \tau$, 
    with an induction on the structure of $e$.
    \begin{mylist1}
    \item[\CASE $\num$:] We have $\EVALT{\num}=\num$.
      By $\jdgtE{\emptyTPEnv}{\num}{\type}$ we have $\type=\inttp$.
      It indeed holds that $\num\sattsk \inttp$. 
    \item[\CASE $x$:] The typing cannot be performed with the empty environment $\emptyset$;
      hence the statement vacuously holds. 
    \item[\CASE $e_1\strcon e_2$:]  By $\jdgtE{\emptyTPEnv}{e_1\strcon e_2}{\type}$ we have
      $\jdgtE{\emptyTPEnv}{e_1}{\type}$, $\jdgtE{\emptyTPEnv}{e_2}{\type}$ and $\type=\strtp$.
      By the \IH we have $\EVALT{e_1}\sattsk \strtp$ and $\EVALT{e_2}\sattsk \strtp$.
      By the definition of $\sattsk$, $e_1=\str_1$ and $e_2=\str_2$ for strings $\str_1$ and $\str_2$.
      Hence $\EVALT{e_1\strcon e_2}=(\str_1\strcon \str_2)$ is a string.
      We have $\EVALT{e_1\strcon e_2}\sattsk \strtp$. 
    \end{mylist1}
  \item[(1b)] We prove if $\EVALT{e}\sattsk \tau$, then $\jdgtt{\emptyTPEnv}{e}{\tau}$, 
    with an induction on the structure of $e$.
    \begin{mylist1}
    \item[\CASE $\num$:] We have $\EVALT{\num}\sattsk \inttp$; hence $\type=\inttp$.
      It indeed holds that $\jdgtt{\emptyTPEnv}{\num}{\inttp}$.
    \item[\CASE $x$:] We do \emph{not} have $\EVALT{x}\sattsk \type$ for any $\type$;
      hence the statement vacuously holds. 
    \item[\CASE $e_1\strcon e_2$:] Since $\EVALT{e_1\strcon e_2}\sattsk \type$, we have
      $\EVALT{e_1\strcon e_2}$ yields $\EVALT{e_1}\strcon\EVALT{e_2}$ rather than $\Err$, and
      $\EVALT{e_1}=\str_1$ and $\EVALT{e_2}=\str_2$ for some $\str_1$ and $\str_2$.
      Hence $\EVALT{e_1}\sattsk \strtp$ and $\EVALT{e_2}\sattsk \strtp$.
      By the \IH we have $\jdgtE{\emptyTPEnv}{e_1}{\strtp}$ and $\jdgtE{\emptyTPEnv}{e_2}{\strtp}$.
      Hence it can be established that $\jdgtE{\emptyTPEnv}{e_1\strcon e_2}{\strtp}$. 
    \end{mylist1}
  \item[(2a)] We prove if $\jdgtt{\emptyTPEnv}{t}{\tau}$, then $\EVALT{t}\sattsk \tau$. 

    The typing judgment for $t$ must be $\jdgtT{\emptyTPEnv}{e_1,...,e_n}{\type}$. 
    We thus have
    $\type=\type_1\times...\times\type_n$ for some $\type_1$, $\type_2$, ..., and $\type_n$ such that
    $\jdgtt{\emptyTPEnv}{e_1}{\type_1}$, ..., $\jdgtt{\emptyTPEnv}{e_n}{\type_n}$.
    By (1a) we have $\forall j\in\{1,...,n\}:\EVALT{e_j}\sattsk \type_j$.
    Hence $\forall j\in\{1,...,n\}:\EVALT{e_j}\neq\Err$ and
    by the definitions of $\EVALT{\cdot}$ and $\cdot\sattsk\cdot$, it holds that
    $\EVALT{e_1,...,e_n}=(\EVALT{e_1},...,\EVALT{e_n})\sattsk 
    \type_1\times ... \times \type_n=\type$. 
    
  \item[(2b)] We prove if $\EVALT{t}\sattsk \tau$, then $\jdgtt{\emptyTPEnv}{t}{\tau}$. 

    Suppose $t=e_1,...,e_n$. 
    Since we do \emph{not} have $\Err\sattsk \type$, it holds that $\EVALT{e_1,...,e_n}\neq\Err$,
    which means $\EVALT{e_1,...,e_n}=(\EVALT{e_1},...,\EVALT{e_n})$ and 
    $\forall j\in\{1,...,n\}:\EVALT{e_j}\neq\Err$. 
    By $(\EVALT{e_1},...,\EVALT{e_n})\sattsk\type$ we have $\type=\type_1\times...\times\type_n$
    for some $\type_1$, $\type_2$, ..., and $\type_n$ such that  
    $\forall j\in\{1,...,n\}: \EVALT{e_j}\sattsk\type_j$. 
    By (1b) we have $\forall j\in\{1,...,n\}:\jdgtE{\emptyTPEnv}{e_j}{\type_j}$. 
    Hence it holds that $\jdgtE{\emptyTPEnv}{e_1,...,e_n}{\type}$. 

\item[(3a),(3b)] The proof ideas are analogous. \qedhere
  \end{mylist1}
\end{proof}

\noindent Pedantically, the condition that 
``at most one binding exists in a typing environment for each variable in its domain''
would be needed in the statements of a few further lemmas whose proofs make use of 
Lemma~\ref{lem:reordering} and/or Lemma~\ref{lem:smaller_dom}. 
However, to avoid the tediousness 
we dispense hereafter with explicating such assumptions on (all of) our typing environments. 
Since bound variables are renamed apart, 
the extensions of typing environments made in the type system all preserve this condition. 

\begin{lem}[Substitution for Tables, Expressions, Predicates and Tuples]\label{lem:subst_tept} 
The following statements hold.
\begin{enumerate}
  \item 
    If $\jdgtT{(\ext{\Gamma}{\TBV:\sktp_0}),\itfc}{\CALTB}{\sktp'}$, 
    and $\emptyTPEnv,\itfc\vdash(I,R):\sktp_0$, 
    then $\jdgtT{\Gamma,\itfc}{\CALTB[(I,R)/\TBV]}{\sktp'}$ holds. 
  \item 
    If $\jdgtT{(\ext{\Gamma}{u:\loctp}),\itfc}{\CALTB}{\sktp'}$, and $l$ is a locality, 
    then $\jdgtT{\Gamma,\itfc}{\CALTB[l/u]}{\sktp'}$ holds. 
  \item 
    If $\jdgtE{\ext{\Gamma}{x:\tau_0}}{e}{\tau}$ and $\val$ is such that $\emptyTPEnv\vdash\val:\tau_0$, 
    then $\jdgtE{\Gamma}{e[\val/x]}{\tau}$ holds.
  \item 
    If $\jdgtE{\ext{\Gamma}{u:\loctp}}{e}{\tau}$, and $l$ is a locality, 
    then $\jdgtE{\Gamma}{e[l/u]}{\tau}$ holds. 
  \item 
    If $\jdgtpred{\ext{\Gamma}{x:\tau_0}}{\psi}$, and $\emptyTPEnv\vdash\val:\tau_0$, 
    then $\jdgtpred{\Gamma}{\psi[\val/x]}$ holds. 
  \item 
    If $\jdgtpred{\ext{\Gamma}{u:\loctp}}{\psi}$, and $l$ is a locality, then $\jdgtpred{\Gamma}{\psi[l/u]}$ holds. 
  \item 
    If $\jdgtt{\ext{\Gamma}{x:\tau_0}}{t}{\tau}$, and $\emptyTPEnv\vdash\val:\tau_0$, 
    then $\jdgtt{\Gamma}{t[\val/x]}{\tau}$ holds. 
  \item
    If $\jdgtt{\ext{\Gamma}{u:\loctp}}{t}{\type}$, and $l$ is a locality, 
    then $\jdgtt{\Gamma}{t[l/u]}{\type}$ holds. 
\end{enumerate}
\end{lem}

\begin{proof}\hfill
  \begin{mylist1}
  \item[(1)] We proceed by a case analysis on $\TBDst$.
    \begin{mylist1}
    \item[\CASE $\TBID@\lloc$:] 
      By $\jdgtT{(\Gamma,\TBV:\type_0),\itfc}{\TBID@\lloc}{\type'}$ we have 
      $\jdgtE{\Gamma,\TBV:\type_0}{\lloc}{\loctp}$, 
      $\TBID\in\dom{\itfc}$ and $\type'=\itfc(\TBID)$. 
      By Lemma~\ref{lem:fv_in_dom}, we have $\FV{\lloc}\subseteq \dom{\Gamma,\TBV:\type_0}$, 
      which implies $\FV{\lloc}\subseteq \dom{\Gamma}$. 
      Hence by Lemma~\ref{lem:smaller_dom}, we have $\jdgtE{\Gamma}{\lloc}{\loctp}$. 
      We can now establish $\jdgtE{\Gamma,\itfc}{(\TBID@\lloc)[(I,R)/\TBV]}{\type'}$ since
      $(\TBID@\lloc)[(I,R)/\TBV]=\TBID@\lloc$.
    \item[\CASE $\TBV'$:] Suppose $\TBV'\neq\TBV$.
      We have $\TBV'[(I,R)/\TBV]=\TBV'$.
      By $\jdgtT{(\Gamma,\TBV:\type_0),\itfc}{\TBV'}{\type'}$ we have $\TBV'\in\dom{\Gamma}$
      and $\type'=\Gamma(\TBV')$. 
      By Lemma~\ref{lem:smaller_dom} we have $\jdgtT{\Gamma,\itfc}{\TBV'}{\type'}$.

      Suppose $\TBV'=\TBV$.
      By $\jdgtt{\Gamma,\TBV:\type_0}{\TBV}{\type'}$ we have $\type'=\type_0$. 
      It is obvious that well-typedness is preserved under weakening of the typing environment.
      Hence by $\jdgtT{\emptyTPEnv,\itfc}{(I,R)}{\type_0}$, 
      we have $\jdgtT{\Gamma,\itfc}{(I,R)}{\type_0}$, or $\jdgtT{\Gamma,\itfc}{(I,R)}{\type'}$. 
    \item[\CASE $(I',R')$:] By $\jdgtT{\Gamma,\TBV:\type_0}{(I',R')}{\type'}$ we have
      $\forall t\in R': \jdgtt{\Gamma,\TBV:\type_0}{t}{I'.\TBSK}$, $I'.\TBID\neq \bot\Rightarrow \itfc(I'.\TBID)=I'.\TBSK$, and
      $I'.\TBSK=\type'$.
      Since for each $t$ in $R'$, it holds that $\FV{t}=\emptyset$, we have
      $\forall t\in R': \jdgtt{\Gamma}{t}{I'.\TBSK}$ by Lemma~\ref{lem:smaller_dom}.
      We can thus establish $\jdgtT{\Gamma,\itfc}{(I',R')[(I,R)/\TBV]}{\type'}$. 
    \end{mylist1}
  \item[(2)] We proceed by a case analysis on $\TBDst$.
    \begin{mylist1}
    \item[\CASE $\TBID@\lloc$:] Suppose $\lloc$ is some constant locality $\llocConst'$.
      We have $(\TBID@\llocConst')[l/u]=\TBID@\llocConst'$.
      By Lemma~\ref{lem:smaller_dom} and $\FV{\TBID@\llocConst'}=\emptyset$,
      we can easily derive $\jdgtt{\Gamma,\itfc}{\TBID@\llocConst'}{\type'}$.

      Suppose $\lloc$ is some locality variable $u'$ such that $u'\neq u$.
      By $\jdgtT{(\Gamma,u:\loctp),\itfc}{\TBID@u'}{\type'}$ we have 
      $\jdgtE{\Gamma,u:\loctp}{u'}{\loctp}$, $\TBID\in\dom{\itfc}$ and $\type'=\itfc(\TBID)$. 
      It is not difficult to see that $u'\in\dom{\Gamma}$ using Lemma~\ref{lem:fv_in_dom}. 
      By Lemma~\ref{lem:smaller_dom} we have $\jdgtE{\Gamma}{u'}{\loctp}$. 
      We can now establish $\jdgtT{\Gamma,\itfc}{(\TBID@u')[l/u]}{\type'}$ since
      $(\TBID@u')[l/u]=\TBID@u'$.

      Suppose $\lloc$ is $u$.
      By $\jdgtT{(\Gamma,u:\loctp),\itfc}{\TBID@u}{\type'}$ 
      we have $\TBID\in\dom{\itfc}$ and $\type'=\itfc(\TBID)$.
      For the locality $l$ we have $\jdgtE{\Gamma}{l}{\loctp}$. 
      Hence we can also establish $\jdgtT{\Gamma,\itfc}{\TBID@l}{\type'}$. 
    \item[\CASE $\TBV$:] Analogous to the corresponding case $\TBV'$ under (1)
      where $\TBV'$ is not changed by the substitution. 
    \item[\CASE $(I,R)$:] Analogous to the corresponding case $(I',R')$ under (1). 
    \end{mylist1}
  \item[(3)] The proof is by induction on the structure of $e$. We omit the trivial cases. 
    \begin{mylist1}
    \item[\CASE $x'$:] Suppose $x'\neq x$.
      By $\jdgtE{\Gamma,x:u_0}{x'}{\type}$ we have $x'\in\dom{\Gamma}$.
      By Lemma~\ref{lem:smaller_dom} we have $\jdgtE{\Gamma}{x'}{\type}$.

      Suppose $x'=x$.
      By $\jdgtE{\Gamma,x:\type_0}{x}{\type}$ we have $\type_0=\type$ and $\type\neq\loctp$.
      By $\jdgtE{\emptyTPEnv}{\val}{\type_0}$ we have $\jdgtE{\Gamma}{\val}{\type_0}$
      ($\val$ is some $\num$, some $\str$ or some $\TBID$).
      Hence it holds that $\jdgtE{\Gamma}{x~\![\val/x]}{\type}$. 
    \item[\CASE $e_1\strcon e_2$:]
      By $\jdgtE{\ext{\Gamma}{x:\tau_0}}{e_1\strcon e_2}{\tau}$ we have
      $\jdgtE{\ext{\Gamma}{x:\tau_0}}{e_1}{\type}$, $\jdgtE{\ext{\Gamma}{x:\tau_0}}{e_2}{\type}$, and
      $\type=\strtp$.
      By the \IH we have $\jdgtE{\Gamma}{e_1[\val/x]}{\strtp}$ and $\jdgtE{\Gamma}{e_2[\val/x]}{\strtp}$.
      We have $e_1[\val/x]\strcon e_2[\val/x]=(e_1\strcon e_2)[\val/x]$; 
      hence it can be established that $\jdgtE{\Gamma}{(e_1\strcon e_2)[\val/x]}{\type}$. 
    \item[\CASE $e_1~\!\aop~\! e_2$:] Analogous to \CASE $e_1\strcon e_2$. 
    \item[\CASE $\{e_1,...,e_n\}$:] Analogous to \CASE $e_1\strcon e_2$.
    \end{mylist1}
  \item[(4)] The proof is by induction on the structure of $e$.
    The main differences compared with the proof of (3) lie with the treatment of
    data variables and locality variables.
    Of the non-inductive cases, we only exhibit the one for locality variables.
    \begin{mylist1}
    \item[\CASE $u'$:] Suppose $u'\neq u$.
      The reasoning is analogous to the sub-case $x'\neq x$ under \CASE $x'$ of (3).

      Suppose $u'=u$.
      By $\jdgtE{\Gamma,u:\loctp}{u}{\type}$ we have $\type=\loctp$.
      Since $l$ is a locality, it can be established that $\jdgtE{\Gamma}{u~\![l/u]}{\type}$. 
    \end{mylist1}
    The inductive cases are analogous to those of (3).
  \item[(5)] The proof is by induction on the structure of $\psi$.
    We only present the representative cases below. 
    \begin{mylist1}
    \item[\CASE $e_1~\!\cop~\!e_2$:]
      Because of $\jdgtpred{\ext{\Gamma}{x:\tau_0}}{e_1~\!\cop~\!e_2}$ we have
      $\jdgtE{\ext{\Gamma}{x:\type_0}}{e_1}{\dttp}$ and $\jdgtE{\ext{\Gamma}{x:\type_0}}{e_2}{\dttp}$,
      for some $\dttp\in\{\inttp,\strtp,\idtp,\loctp\}$.
      By (3), we get $\jdgtE{\Gamma}{e_1[\val/x]}{\dttp}$ and $\jdgtE{\Gamma}{e_2[\val/x]}{\dttp}$.
      Since $(e_1~\!\cop~\!e_2)[\val/x]=(e_1[\val/x]~\!\cop~\!e_2[\val/x])$ holds,
      we have $\jdgtpred{\Gamma}{(e_1~\!\cop~\! e_2)[\val/x]}$. 
    \item[\CASE $\psi_1\land\psi_2$:] Because of $\jdgtpred{\ext{\Gamma}{x:\tau_0}}{\psi_1\land\psi_2}$ we have
      $\jdgtpred{\ext{\Gamma}{x:\tau_0}}{\psi_1}$ and $\jdgtpred{\ext{\Gamma}{x:\tau_0}}{\psi_2}$.
      By the \IH we have $\jdgtpred{\Gamma}{\psi_1[\val/x]}$ and $\jdgtpred{\Gamma}{\psi_2[\val/x]}$.
      Since $(\psi_1\land \psi_2)[\val/x]=(\psi_1[\val/x]\land \psi_2[\val/x])$ holds, 
      it can be established that $\jdgtpred{\Gamma}{(\psi_1\land\psi_2)[\val/x]}$. 
    \end{mylist1}
  \item[(6)] Analogous to (5).
  \item[(7)] The proof is straightforward using (3). 
  \item[(8)] The proof is straightforward using (4). \qedhere
  \end{mylist1}
\end{proof}

\begin{lem}[Substitution for Actions and Processes]\label{lem:subst_aP}
The following statements hold. 
\begin{enumerate}
  \item 
    If $\jdgta{(\ext{\Gamma}{x:\tau_0}),\itfc}{a}{\Gamma'}$, $\BV{a}\cap\dom{\Gamma,x:\type_0}=\emptyset$, 
    and $\val$ is such that $\emptyTPEnv\vdash\val:\tau_0$, then 
    $\jdgta{\Gamma,\itfc}{a[\val/x]}{\Gamma'}$ holds; \\
    if $\jdgtP{(\ext{\Gamma}{x:\tau_0}),\itfc}{P}$, $\BV{P}\cap\dom{\Gamma,x:\type_0}=\emptyset$, 
    and $\val$ is such that $\emptyTPEnv\vdash\val:\tau_0$, then
    $\jdgtP{\Gamma,\itfc}{P[\val/x]}$ holds. 
  \item 
    If $\jdgta{(\ext{\Gamma}{\TBV:\tau_0}),\itfc}{a}{\Gamma'}$, $\BV{a}\cap\dom{\Gamma,\TBV:\type_0}=\emptyset$, 
    and $\emptyTPEnv,\itfc\vdash(I,R):\tau_0$, 
    then $\jdgta{\Gamma,\itfc}{a[(I,R)/\TBV]}{\Gamma'}$ holds; \\
    if $\jdgtP{(\ext{\Gamma}{\TBV:\tau_0}),\itfc}{P}$, $\BV{P}\cap\dom{\Gamma,\TBV:\type_0}=\emptyset$, 
    and $\emptyTPEnv,\itfc\vdash(I,R):\tau_0$, then 
    $\jdgtP{\Gamma,\itfc}{P[(I,R)/\TBV]}$ holds. 
  \item 
    If $\jdgta{(\ext{\Gamma}{u:\loctp}),\itfc}{a}{\Gamma'}$, $\BV{a}\cap\dom{\Gamma,u:\loctp}=\emptyset$,
    and $l$ is a locality, then $\jdgta{\Gamma,\itfc}{a[l/u]}{\Gamma'}$ holds; \\
    if $\jdgtP{(\ext{\Gamma}{u:\loctp}),\itfc}{P}$,
    $\BV{P}\cap\dom{\Gamma,u:\loctp}=\emptyset$, and $l$ is a locality, then 
    $\jdgtP{\Gamma,\itfc}{P[l/u]}$ holds. 
\end{enumerate}
\end{lem}

\begin{proof}\hfill
\begin{mylist1}
  \item[(1)] We proceed with an induction on the structure of actions and processes. 
  \begin{mylist1}
    \item[\CASE $\INS{\TBID}{t}{\lloc}$:] 
      Due to $\jdgta{(\Gamma,x:\type_0),\itfc}{\INS{\TBID}{t}{\lloc}}{\Gamma'}$ 
      we obtain $\jdgtT{\Gamma,x:\type_0}{t}{\itfc(\TBID)}$, $\jdgtE{\Gamma,x:\type_0}{\lloc}{\loctp}$, 
      and $\Gamma'=\emptyTPEnv$. 
      By Lemma~\ref{lem:subst_tept}, we get $\jdgtT{\Gamma}{t[\val/x]}{\itfc(\TBID)}$. 
      By Lemma~\ref{lem:fv_in_dom} and Lemma~\ref{lem:smaller_dom} we have $\jdgtE{\Gamma}{\lloc}{\loctp}$. 
      Since $\INS{\TBID}{t}{\lloc}[\val/x]$ gives $\INS{\TBID}{t[\val/x]}{\lloc}$, 
      we can indeed establish 
      \small
      $$\jdgta{\Gamma,\itfc}{\INS{\TBID}{t}{\lloc}[\val/x]}{\Gamma'}.$$
      \normalsize
    \item[\CASE $\DEL{\TBID}{T}{\PRED}{\lloc}$:]
      Because of $\jdgta{(\Gamma,x:\type_0),\itfc}{\DEL{\TBID}{T}{\PRED}{\lloc}}{\Gamma'}$ 
      there exists some typing environment $\Gamma''$ such that 
      $\jdgtT{\itfc(\TBID)}{T}{\Gamma''}$, $\jdgtpred{(\Gamma,x:\type_0,\Gamma''),\itfc}{\PRED}$, 
      $\jdgtE{\Gamma,x:\type_0}{\lloc}{\loctp}$, and $\Gamma'=\emptyTPEnv$. 
      Since $\BV{\DEL{\TBID}{T}{\PRED}{\lloc}}\cap\dom{\Gamma,x:\type_0}=\emptyset$, 
      we have $\dom{\Gamma''}\cap\dom{\Gamma,x:\type_0}=\emptyset$. 
      By Lemma~\ref{lem:reordering}, we have $\jdgtpred{\Gamma,\Gamma'',x:\type_0}{\PRED}$. 
      By Lemma~\ref{lem:subst_tept}, we have $\jdgtpred{\Gamma,\Gamma''}{\PRED[\val/x]}$. 
      By Lemma~\ref{lem:fv_in_dom} and Lemma~\ref{lem:smaller_dom} we have $\jdgtE{\Gamma}{\lloc}{\loctp}$. 
      Combining this with the typing for $T$ we can establish 
      \small
      $$\jdgta{\Gamma,\itfc}{\DEL{\TBID}{T}{\PRED}{\lloc}[\val/x]}{\Gamma'},$$
      \normalsize
      using the fact that the substitution gives $\DEL{\TBID}{T}{\PRED[\val/x]}{\lloc}$.       
    \item[\CASE $\SEL{\LST{\CALTB}}{T}{\psi}{t}{!\TBV}$:] 
      By $\jdgta{(\Gamma,x:\type_0),\itfc}{\SEL{\LST{\CALTB}}{T}{\psi}{t}{!\TBV}}{\Gamma'}$ 
      we have $\jdgtE{(\Gamma,x:\type_0),\itfc}{\CALTB_j}{\type_j}$ for all $j\in\{1,...,n\}$, 
      $\jdgtT{\flattensk(\type_1\times...\times\type_n)}{T}{\Gamma''}$ for some typing environment $\Gamma''$ 
      such that $\jdgtpred{(\Gamma,x:\type_0),\Gamma''}{\PRED}$, 
      and $\jdgtT{(\Gamma,x:\type_0),\Gamma''}{t}{\type'}$, 
      for some $\type'$ such that $\Gamma'=[\TBV:\type']$. 
      For each $j\in\{1,...,n\}$ we have $\CALTB_j[\val/x]=\CALTB_j$. 
      By Lemma~\ref{lem:fv_in_dom}, we have $\FV{\CALTB_j}\subseteq \dom{\Gamma}\cup\{x\}$ for each $j$. 
      Since it is obvious that the only kind of variables each $\CALTB_j$ can contain are table variables, 
      we have $\FV{\CALTB_j}\subseteq \dom{\Gamma}$. 
      By Lemma~\ref{lem:smaller_dom}, we have 
      \small
      \begin{equation}\label{eq:sub_a_tbtp'} 
      \forall j\in\{1,...,n\}: \jdgtE{\Gamma,\itfc}{\CALTB_j}{\type_j}  
      \end{equation}
      \normalsize
      By Lemma~\ref{lem:reordering} and Lemma~\ref{lem:subst_tept} (analogously to the previous case), 
      we have 
      \small
      \begin{equation}\label{eq:sub_a_predtp'} 
      \jdgtpred{\Gamma,\Gamma''}{\PRED[\val/x]} 
      \end{equation}
      \normalsize
      Again by Lemma~\ref{lem:reordering} and Lemma~\ref{lem:subst_tept}, we have 
      \small
      \begin{equation}\label{eq:sub_a_ttp'} 
      \jdgtT{\Gamma,\Gamma''}{t[\val/x]}{\type'} 
      \end{equation}
      \normalsize 
      Combining \eqref{eq:sub_a_tbtp'}, \eqref{eq:sub_a_predtp'}, \eqref{eq:sub_a_ttp'}, 
      and the original typing for $T$, we can establish 
      $\jdgta{\Gamma,\itfc}{\SEL{\LST{\CALTB}}{T}{\psi[\val/x]}{t[\val/x]}{!\TBV}}{\Gamma'}$,  
      which gives the desired result since we have  
      $\SEL{\LST{\CALTB}}{T}{\psi}{t}{!\TBV}[\val/x]=
      \SEL{\LST{\CALTB}}{T}{\psi[\val/x]}{t[\val/x]}{!\TBV}$. 
    \item[\CASE $\UPDATE{\TBID}{T}{\PRED}{t}{\lloc}$:] Analogous to the previous cases. 
    \item[\CASE $\AGGR{\TBID}{T}{\PRED}{f}{T'}{\lloc}$:] Analogous to the previous cases. 
    \item[\CASE $\CREATENEW{\TBID@\lloc}{\TBSK}$:] Trivial.
    \item[\CASE $\DROP{\TBID}{\lloc}$:] Trivial. 
    \item[\CASE $\EVAL{P'}{\lloc}$:] 
    Because of $\jdgta{(\Gamma,x:\type_0),\itfc}{\EVAL{P'}{\lloc}}{\Gamma'}$ 
    we have $\jdgtP{(\Gamma,x:\type_0),\itfc}{P'}$, $\jdgtE{\Gamma,x:\type_0}{\lloc}{\loctp}$, 
    and $\Gamma'=\emptyTPEnv$.\!
    Since $\BV{\EVAL{P'}{\lloc}}\cap\dom{\Gamma,x:\type_0}=\emptyset$,
    we have $\BV{P'}\cap\dom{\Gamma,x:\type_0}=\emptyset$. 
    By the \IH we have $\jdgtP{\Gamma,\itfc}{P'[\val/x]}$, and
    by Lemmas \ref{lem:fv_in_dom} and~\ref{lem:smaller_dom} 
    we get $\jdgtE{\Gamma}{\lloc}{\loctp}$. 
    Hence we can establish $\jdgta{\Gamma,\itfc}{\EVAL{P'[\val/x]}{\lloc}}{\Gamma'}$, 
    which gives the desired result since
    $(\EVAL{P'}{\lloc})[\val/x]=\EVAL{P'[\val/x]}{\lloc}.$

    \item[\CASE $\NILP$:] Trivial. 
    \item[\CASE $a'.P'$:] By $\jdgtP{(\Gamma,x:\type_0),\itfc}{a'.P'}$ 
      we have $\jdgta{(\Gamma,x:\type_0),\itfc}{a'}{\Gamma''}$ for some $\Gamma''$ such that
      $\jdgtP{(\Gamma,x:\type_0,\Gamma''),\itfc}{P'}$. 
      By the \IH we have $\jdgta{\Gamma,\itfc}{a'[\val/x]}{\Gamma''}$.
      By Lemma~\ref{lem:reordering} we have $\jdgtP{\Gamma,\Gamma'',x:\type_0}{P'}$.
      By $\BV{a'.P'}\cap\dom{\Gamma,x:\type_0}=\emptyset$ we have
      $\BV{P'}\cap\dom{\Gamma,x:\type_0}=\emptyset$.
      Since bound variables are renamed apart we have $\BV{P'}\cap\BV{a'}=\emptyset$, and it follows that
      $\BV{P'}\cap\dom{\Gamma,\Gamma'',x:\type_0}=\emptyset$. 
      By the \IH we then have $\jdgtP{(\Gamma,\Gamma''),\itfc}{P'[\val/x]}$.
      Hence we can establish $\jdgtP{\Gamma,\itfc}{(a'.P')[\val/x]}$.
    \item[\CASE $P_1;P_2$:] Analogous.
    \item[\CASE $A(\LST{e})$:] By $\jdgtP{(\Gamma,x:\type_0),\itfc}{A(\LST{e})}$, we have
      $\jdgtE{\Gamma,x:\type_0}{e_1}{\type_1}$, ..., $\jdgtE{\Gamma,x:\type_0}{e_n}{\type_n}$, and 
      $\jdgtP{(\mi{var}_1:\type_1,...,\mi{var}_n:\type_n),\itfc}{P'}$, 
      where $P'$ is the procedure body. 
      By Lemma~\ref{lem:subst_tept}, we have $\jdgtE{\Gamma}{e_1[\val/x]}{\type_1}$, ...,
      $\jdgtE{\Gamma}{e_n[\val/x]}{\type_n}$. 
      Hence we can establish $\jdgtP{\Gamma,\itfc}{A(\LST{e})[\val/x]}$
      (Note that the substitution is not performed to $P'$). 
    \item[\CASE $\FOREACH{\CALTB}{T}{\PRED}{\ordr}{P'}$:]
      By $\jdgtP{(\Gamma,x:\type_0),\itfc}{\FOREACH{\CALTB}{T}{\PRED}{\ordr}{P'}}$ we have\\
      $\jdgtT{(\Gamma,x:\type_0),\itfc}{\CALTB}{\type}$ for some $\type$ such that
      $\jdgtT{\type}{T}{\Gamma''}$ for some $\Gamma''$ such that 
      $\jdgtpred{\Gamma,x:\type_0,\Gamma''}{\psi}$, and
      $\jdgtP{(\Gamma,x:\type_0,\Gamma''),\itfc}{P'}$.
      By Lemma~\ref{lem:fv_in_dom} we have $\FV{\CALTB}\subseteq \dom{\Gamma,x:\type_0}$.
      Since the only kind of variables that $\CALTB$ can contain are table variables, 
      we have $\FV{\CALTB}\subseteq\dom{\Gamma}$. 
      By Lemma~\ref{lem:smaller_dom}, and $\CALTB[\val/x]=\CALTB$, 
      we have
      \small
      \begin{equation}
        \label{eq:sub_P_tbtp'}
        \jdgtE{\Gamma,\itfc}{\CALTB[\val/x]}{\type}
      \end{equation}
      \normalsize
      Since $\BV{\FOREACH{\CALTB}{T}{\PRED}{\ordr}{P'}}\cap\dom{\Gamma,x:\type_0}=\emptyset$, 
      we have $\BV{T}\cap\dom{\Gamma,x:\type_0}=\emptyset$, 
      and $\dom{\Gamma''}\cap\dom{\Gamma,x:\type_0}=\emptyset$. 
      By Lemma~\ref{lem:reordering}, we have
      $\jdgtpred{\Gamma,\Gamma'',x:\type_0}{\PRED}$. 
      By Lemma~\ref{lem:subst_tept}, we have
      \small
      \begin{equation}
        \label{eq:sub_P_predtp'}
        \jdgtpred{\Gamma,\Gamma''}{\PRED[\val/x]}
      \end{equation}
      \normalsize
      Again by Lemma~\ref{lem:reordering}, we have $\jdgtP{(\Gamma,\Gamma'',x:\type_0),\itfc}{P'}$.
      It is not difficult to see that $\BV{P'}\cap\dom{\Gamma,\Gamma'',x:\type_0}=\emptyset$. 
      By the \IH we have
      \small
      \begin{equation}
        \label{eq:sub_P_Ptp'}
        \jdgtP{(\Gamma,\Gamma''),\itfc}{P'[\val/x]}
      \end{equation}
      \normalsize
      Combining \eqref{eq:sub_P_tbtp'}, \eqref{eq:sub_P_predtp'}, \eqref{eq:sub_P_Ptp'}, and
      $\jdgtT{\type}{T}{\Gamma''}$, we can establish
      \small
      $$\jdgtP{\Gamma,\itfc}{(\FOREACH{\CALTB}{T}{\PRED}{\ordr}{P'})[\val/x]}.$$
      \normalsize
    \end{mylist1}
    
  \item[(2)] We proceed with an induction on the structure of actions and processes.
    The reasoning is largely similar to that of (1); hence we only present a few representative cases
    where the differences from (1) can be demonstrated. 
    \begin{mylist1}
    \item[\CASE $\INS{\TBID}{t}{\lloc}$:]
       By $\jdgta{(\Gamma,\TBV:\type_0),\itfc}{\INS{\TBID}{t}{\lloc}}{\Gamma'}$ 
       we have $\jdgtT{\Gamma,\TBV:\type_0}{t}{\itfc(\TBID)}$, 
       $\jdgtE{\Gamma,\TBV:\type_0}{\lloc}{\loctp}$, 
       and $\Gamma'=\emptyTPEnv$.
       Tuples cannot contain table variables; hence $t[(I,R)/\TBV]=t$.
       By Lemma~\ref{lem:fv_in_dom} we have $\FV{t}\subseteq \dom{\Gamma,\TBV:\type_0}$.
       Hence we also have $\FV{t}\subseteq \dom{\Gamma}$.
       By Lemma~\ref{lem:smaller_dom} we have $\jdgtT{\Gamma}{t}{\itfc(\TBID)}$.
       Analogously to the same case under (1), we can deduce $\jdgtE{\Gamma}{\lloc}{\loctp}$. 
       We can now establish
       \small
       $$\jdgtT{\Gamma}{(\INS{\TBID}{t}{\lloc})[(I,R)/\TBV]}{\Gamma'}.$$
       \normalsize
     \item[\CASE $\SEL{\LST{\CALTB}}{T}{\psi}{t}{!\TBV'}$:]
       By $\jdgta{(\Gamma,\TBV:\type_0),\itfc}{\SEL{\LST{\CALTB}}{T}{\psi}{t}{!\TBV'}}{\Gamma'}$
       we have
       $\jdgtT{(\Gamma,\TBV:\type_0),\itfc}{\CALTB_1}{\type_1}$, ...,
       $\jdgtT{(\Gamma,\TBV:\type_0),\itfc}{\CALTB_n}{\type_n}$,
       for some $\type_1$, ..., $\type_n$, such that
       $\jdgtT{\flattensk(\type_1\times...\times\type_n)}{T}{\Gamma''}$, for some $\Gamma''$ such that
       $\jdgtpred{\Gamma,\TBV:\type_0,\Gamma''}{\PRED}$, and
       $\jdgtT{\Gamma,\TBV:\type_0,\Gamma''}{t}{\type'}$ for some $\type'$ such that
       $\Gamma'=[\TBV':\type']$.
       
       By Lemma~\ref{lem:subst_tept} we have
       $\jdgtT{\Gamma,\itfc}{\CALTB_1[(I,R)/\TBV]}{\type_1}$, ...,
       $\jdgtT{\Gamma,\itfc}{\CALTB_n[(I,R)/\TBV]}{\type_n}$.
       We have $\PRED[(I,R)/\TBV]=\PRED$ and $t[(I,R)/\TBV]=t$, and 
       using Lemma~\ref{lem:fv_in_dom} and Lemma~\ref{lem:smaller_dom} we have
       $\jdgtpred{\Gamma,\Gamma''}{\PRED}$ and $\jdgtT{\Gamma,\Gamma''}{t}{\type'}$.
       It is essentially only the tables $\CALTB_1$, ..., $\CALTB_n$ that are potentially affected by
       the substitution, and we can now establish 
       \small
       $$
       \jdgta{\Gamma,\itfc}{(\SEL{\LST{\CALTB}}{T}{\PRED}{t}{!\TBV'})[(I,R)/\TBV]}{\Gamma'}. 
       $$
       \normalsize
     \end{mylist1}
   \item[(3)] The proof is still by induction on the structure of actions and processes.
     We omit the details since they are analogous to those of (1) and (2) above. \qedhere
  \end{mylist1}
\end{proof}

\Comment{
\begin{lem}
  If $\jdgtpred{\Gamma[x:\tau_0]}{\psi}$ and $\emptyTPEnv\vdash\val:\tau_0$ then $\EVALT{\psi[\val/x]}{}\neq \Err$. 
\end{lem}
}

\begin{lem}\label{lem:tp_T_sat_sk}
  If $\jdgtT{\sktp}{T}{\Gamma}$, then $T\sattsk \sktp$. 
\end{lem}
\begin{proof}{~}\\
We first show that if $\jdgtT{\type}{\bd}{\Gamma}$, then $W\sattsk\type$. 
There are two cases:
\begin{mylist1}
\item[$W=!x$:] We have $\jdgtT{\msettp}{!x}{[x:\msettp]}$ where $\type=\msettp$ and $\msettp\neq\loctp$. 
  The type $\msettp$ can thus be one of 
  $\inttp$, $\strtp$, $\idtp$, $\msett{\inttp}$, $\msett{\strtp}$, $\msett{\idtp}$, $\msett{\loctp}$. 
  By the definition of $\cdot\sattsk\cdot$, we can always establish 
  $!x\sattsk \msettp$, which is $!x\sattsk\type$. 
\item[$W=!u$:] We have $\jdgtT{\loctp}{!u}{[u:\loctp]}$ and $\type=\loctp$. 
  By the definition of $\cdot\sattsk\cdot$, we can directly establish $!u\sattsk\type$. 
\end{mylist1}
We now show that if $\jdgtT{\type}{T}{\Gamma}$, then $T\sattsk\type$. 
Suppose $T$ is $\bd_1,...,\bd_n$. 
By $\jdgtT{\type}{T}{\Gamma}$ we have 
$\jdgtT{\type_1}{\bd_1}{\Gamma_1}$, ..., $\jdgtT{\type_n}{\bd_n}{\Gamma_n}$
for some $\type_1$, ..., $\type_n$ and $\Gamma_1$, ..., $\Gamma_n$ such that 
$\type=\type_1\times...\times\type_n$ and $\Gamma=\Gamma_1,...,\Gamma_n$. 
Hence we have $\bd_1\sattsk \type_1$, ..., $\bd_n\sattsk \type_n$. 
By the definition of $\cdot\sattsk\cdot$ we have $\bd_1,...,\bd_n\sattsk \type_1\times...\times\type_n$, 
which is $T\sattsk\type$. 
\end{proof}

\begin{lem}\label{lem:sub_eval_n_err}
	If $t$ does not contain operators ($\strcon$, $\aop$), 
  $\jdgtt{\emptyTPEnv}{t}{\tau}$ and $\jdgtT{\tau}{T}{\Gamma}$, then $\subst{t}{T}\neq\Err$ holds. 
  Additionally, the following statements hold. 
  \begin{enumerate}
    \item If $\jdgtpred{\Gamma}{\psi}$, and $\BV{T}\supseteq \FV{\psi}$, 
      then we have $\EVALPred{\psi(\subst{t}{T})}\neq\Err$.
    \item If $\jdgtt{\Gamma}{t'}{\tau'}$, and $\BV{T}\supseteq \FV{t'}$,
      then we have $\EVALT{t'(\subst{t}{T})}\sattsk\tau'$. 
  \end{enumerate}
\end{lem}

\begin{proof}\hfill
\begin{mylist1}
\item[$\bullet$] 
We first show that if $t$ does not contain operators, 
$\jdgtt{\emptyTPEnv}{t}{\type}$ and $\jdgtT{\type}{\bd}{\Gamma}$, 
then $t/\bd\neq\Err$ holds. 
There are two cases for $W$:
  \begin{mylist1}
  \item[$W=!x$:] It is not difficult to deduce from $\jdgtT{\type}{!x}{\Gamma}$ that 
    $\type$ is a multiest type $\msettp$ that cannot be $\loctp$. 
    By $\jdgtt{\emptyTPEnv}{t}{\msettp}$, the fact that $t$ does not contain operators, 
    and Lemma~\ref{lem:fv_in_dom}, 
    we know that $t$ is a $\num$, $\str$, $\TBID$ or a multiset of 
    intergers/strings/table identifiers/localities.
    By the definintion of pattern matching we have $t/\bd\neq\Err$. 
  \item[$W=!u$:] By $\jdgtT{\type}{!u}{\Gamma}$ we have $\type=\loctp$. 
    By $\jdgtt{\emptyTPEnv}{t}{\loctp}$ we have $t=l$ is a locality constant. 
    We indeed have $l/!u\neq\Err$. 
  \end{mylist1}
\item[$\bullet$] 
We then show that if $t$ does not contain operators, 
$\jdgtt{\emptyTPEnv}{t}{\type}$ and $\jdgtT{\type}{T}{\Gamma}$, 
then $t/\bd\neq\Err$ holds. 
From $\jdgtT{\type}{T}{\Gamma}$ we have $\bd_1$, ..., $\bd_n$, $\type_1$, ..., $\type_n$, and 
$\Gamma_1,...,\Gamma_n$ such that 
$T=\bd_1,...,\bd_n$, 
$\jdgtT{\type_1}{\bd_1}{\Gamma_1}$, ..., $\jdgtT{\type_n}{\bd_n}{\Gamma_n}$, 
$\type=\type_1\times...\times\type_n$ and $\Gamma=\Gamma_1,...,\Gamma_n$. 
By $\jdgtt{\emptyTPEnv}{t}{\type_1\times...\times\type_n}$ we have 
$t=e_1,...,e_n$ for expressions $e_1$, ..., $e_n$ such that 
$\jdgtt{\emptyTPEnv}{e_1}{\type_1}$, ..., $\jdgtt{\emptyTPEnv}{e_n}{\type_n}$. 
Since $t$ does not contain operators, based on the latter typing judgments for $n$ expressions
we can establish the typing of $n$ \emph{singleton tuples} $e_1$, ..., $e_n$ of the same form. 
Hence using previous result we have $e_1/\bd_1\neq\Err$, ..., $e_n/\bd_n\neq\Err$. 
We now have $(e_1,...,e_n)/(\bd_1,...,\bd_n)\neq\Err$, or $t/T\neq\Err$. 

\item[$\bullet$] We next show point (1) in the statement of the lemma. 
  Continuing the argumentation above, 
  since $\jdgtE{\emptyTPEnv}{e_1}{\type_1}$ and $\jdgtT{\type_1}{\bd_1}{\Gamma_1}$, 
  where $e_1$ does not contain variables or operators, 
  we have $\jdgtpred{\Gamma'}{\psi[e_1/\bd_1]}$ 
  using Lemma~\ref{lem:reordering} and Lemma~\ref{lem:subst_tept}, 
  where $\Gamma'$ is obtained by removing from $\Gamma$
  the binding for the only variable in $\BV{\bd_1}$.
  Similarly we have $\jdgtpred{\Gamma''}{\PRED[e_1/\bd_1,e_2/\bd_2]}$, ..., 
  $\jdgtpred{\Gamma^{(n)}}{\PRED[e_1/\bd_1,...,e_n/\bd_n]}$, 
  for some $\Gamma''$, ..., $\Gamma^{(n)}$, that are obtained from $\Gamma'$ 
  by subsequently removing the bindings for the variables of $\BV{\bd_2}$, ..., $\BV{\bd_n}$. 
  The last judgment is actually $\jdgtpred{\Gamma^{(n)}}{\PRED(t/T)}$. 
  Since $\BV{T}\supseteq\FV{\PRED}$, we have $\FV{\PRED(t/T)}=\emptyset$. 
  By Lemma~\ref{lem:smaller_dom} we have $\jdgtpred{\emptyTPEnv}{\PRED(t/T)}$. 
  By Lemma~\ref{lem:tp_eval} we have $\EVALT{\PRED(t/T)}\neq \Err$. 

\item[$\bullet$] 
  The reasoning needed to establish point (2) in the statement of the lemma is analogous. 
\end{mylist1}

\end{proof}

\begin{lem}\label{lem:tp_proj} 
  If each component of $t$ is a constant or a variable bound in $T$, and 
    $\exists \Gamma: \jdgtT{\tau}{T}{\Gamma}\land \jdgtt{\Gamma}{t}{\tau'}$, 
    then we have $\skproj{\tau}{T}{t}=\tau'$. 
\end{lem}
\begin{proof}
  Suppose $T=\bd_1,...,\bd_k$, $t=e_1,...,e_n$.
  Then $\skproj{\tau}{T}{t}$ is an $n$-ary product type.
  Assume $\skproj{\tau}{T}{t}=\type''_1\times...\times \type''_n$ 
  for some $\type''_1$, ..., $\type''_n$. 
  By $\jdgtT{\type}{\bd_1,...,\bd_k}{\Gamma}$ we have 
  $\type=\type_1\times...\times\type_k$ for some $\type_1$, ..., $\type_k$. 
  By $\jdgtt{\Gamma}{e_1,...,e_n}{\type'}$ we have $\type'=\type'_1\times...\times\type'_n$
  for some $\type'_1$, ..., $\type'_n$. 
  We show that for each $i\in\{1,...,n\}$, $\type''_i=\type'_i$, 
  which will give $\skproj{\type}{T}{t}=\type'$. 
  Fixing $i$, we have two cases for $e_i$. 
  \begin{mylist1}
  \item[$e_i$ is a constant value: ] 
    By $\jdgtE{\Gamma}{e_1,...,e_n}{\type'_1\times...\times\type'_n}$ 
    we have $\jdgtE{\Gamma}{e_i}{\type'_i}$. 
    By Lemma~\ref{lem:smaller_dom} we have $\jdgtE{\emptyTPEnv}{e_i}{\type'_i}$. 
    Since $e_i$ is a value ($\EVALT{e_i}=e_i$), by Lemma~\ref{lem:tp_eval} we have $e_i\sattsk \type'_i$. 
    On the other hand, by $\skproj{\type}{T}{t}=\type''_1\times...\times\type''_n$ 
    we have $e_i\sattsk \type''_i$. 
    Hence $\type''_i=\type'_i$ holds. 
  \item[$e_i$ is a variable bound in $\bd_j$: ] 
    By $\jdgtT{\type_1\times...\times\type_k}{T}{\Gamma}$ we have the binding $e_i:\type_j$ in $\Gamma$. 
    By $\jdgtt{\Gamma}{e_1,...,e_n}{\type'_1\times...\times\type'_n}$ we have 
    $\jdgtE{\Gamma}{e_i}{\type'_i}$ and thus $\type'_i=\type_j$. 
    By $\skproj{\type_1\times...\times_k}{T}{t}=\type''_1\times...\times\type''_n$ 
    we also have $\type''_i=\type_j$. 
    Hence $\type''_i=\type'_i$ holds. 
  \end{mylist1}
  This completes the proof. 
\end{proof}

\begin{lem}\label{lem:sk_prod}
  If $\jdgtE{\Gamma,\itfc}{\CALTB_1}{\sktp_1}$, ..., $\jdgtE{\Gamma,\itfc}{\CALTB_n}{\sktp_n}$, and 
    $\jdgtC{\Gamma,\itfc}{(I_k,R_k)}$ holds for each $(I_k,R_k)$ 
    in the list $\LSTStruct{\llocConst_i::(I_i,R_i)}{i}$, 
    then $\prodSK(\LST{\CALTB},\LSTStruct{\llocConst_i::(I_i,R_i)}{i})=
    \flattensk(\sktp_1\times...\times\sktp_n)$ 
    holds as long as $\prodSK(\LST{TB},\LSTStruct{\llocConst_i::(I_i,R_i)}{i})$ is defined. 
\end{lem}
\begin{proof}
  According to the definiton of $\prodSK(\cdot,\cdot)$, if it is defined, 
  there exist $\type'_1$, ..., $\type'_n$ such that 
  \small
  $$\prodSK(\LST{\CALTB}, \LSTStruct{\llocConst_i::(I_i,R_i)}{i})=
  \flattensk(\type'_1\times...\times\type'_n),$$
  \normalsize
  where for each $j$, $\type'_j$ falls into one of the following two cases, 
  in each of which we show that $\type'_j=\type_j$. 
  \begin{mylist1}
    \item[$\bullet$] There exists some $k$ such that $\CALTB_j=I_k.\TBID@\llocConst_k$ 
      and $\type'_j=I_k.\TBSK$. 
      By $\jdgtP{\Gamma,\itfc}{(I_k,R_k)}$ we have $\itfc(I_k.\TBID)=I_k.\TBSK$. 
      From the conditions we also have $\jdgtT{\Gamma,\itfc}{I_k.\TBID@\llocConst_k}{\type_j}$. 
      Hence we have $\type_j=\itfc(I_k.\TBID)=I_k.\TBSK$. 
      Therefore we have $\type'_j=\type_j$. 
    \item[$\bullet$] There exists some data set $R_0$ such that $\CALTB_j=(I_0,R_0)$
      and $\type'_j=I_0.\TBSK$. 
      From the conditions we have $\jdgtT{\Gamma,\itfc}{(I_0,R_0)}{\type_j}$; 
      hence $\type_j=I_0.\TBSK$. 
      Therefore we have $\type'_j=\type_j$. 
  \end{mylist1}
  This completes the proof. 
\end{proof}

\begin{lem}\label{lem:dt_prod}
  If $\jdgtE{\Gamma,\itfc}{\CALTB_1}{\tau_1}$, ..., $\jdgtE{\Gamma,\itfc}{\CALTB_n}{\tau_n}$, and 
    $\jdgtC{\Gamma,\itfc}{(I_k,R_k)}$ holds for each $(I_k,R_k)$ in the list $\LSTStruct{\llocConst_i::(I_i,R_i)}{i}$, 
    then 
    $\forall t'\in\prodR(\LST{TB},\LSTStruct{\llocConst_i::(I_i,R_i)}{i}): 
    \jdgtt{\emptyTPEnv}{t'}{\flattensk(\tau_1\times...\times\tau_n)}$ holds
    as long as $\prodR(\LST{TB},\LSTStruct{\llocConst_i::(I_i,R_i)}{i})$ is defined. 
\end{lem}
\begin{proofidea}
The proof idea is analogous to that of Lemma~\ref{lem:sk_prod}. 
A case analysis based on the definition of $\prodR(\cdot,\cdot)$ is needed. 
\end{proofidea}

\begin{lem}\label{lem:minus_t}
If $\jdgtT{\Gamma,\itfc}{(I,R)}{\tau}$, and $t$ is a tuple, then $\jdgtT{\Gamma,\itfc}{(I,R\setminus\{t\})}{\tau}$. 
\end{lem}
The proof of Lemma~\ref{lem:minus_t} is trivial. 

\begin{lem}\label{lem:subj_red}
If $\envnet$ and $N$ are closed, $\BV{N}\cap\dom{\Gamma}=\emptyset$, 
$\jdgtN{\Gamma,\itfc}{\envnet}$, $\jdgtN{\Gamma,\itfc}{N}$, 
and $\STEP{\envnet}{N}{N'}$, then $\jdgtN{\Gamma,\itfc}{N'}$.
\end{lem}
\begin{proof}
  We proceed with an induction on the derivation of $\STEP{\envnet}{N}{N'}$. 
  \begin{mylist1}
  \item[\CASE $\REDINS$:] 
    The net $N$ is $l_1::\INS{\TBID}{t}{l_2}.P~\!||~\!l_2::(I,R)$. 
    By 
    \small
    $$\jdgtN{\Gamma,\itfc}{l_1::\INS{\TBID}{t}{l_2}.P~\!||~\!l_2::(I,R)}$$
    \normalsize
    we have $\jdgtP{\Gamma,\itfc}{\INS{\TBID}{t}{l_2}.P}$ and $\jdgtC{\Gamma,\itfc}{(I,R)}$;  
    hence $\jdgta{\Gamma,\itfc}{\INS{\TBID}{t}{l_2}}{\emptyTPEnv}$ and (by $\Gamma,\emptyTPEnv=\Gamma$)
    \begin{equation}\label{eq:ins_tp_P}
     \jdgtP{\Gamma,\itfc}{P}. 
    \end{equation}
    By the typing of the insertion action we have $\jdgtt{\Gamma}{t}{\itfc(\TBID)}$, 
    and $\jdgtE{\Gamma}{l_2}{\loctp}$. 
    By $\jdgtC{\Gamma,\itfc}{(I,R)}$ we have $\itfc(I.\TBID)=I.\TBSK$. 
    Since the transition can happen we have $I.\TBID=\TBID$ by the semantic rule $\REDINS$. 
    Hence $\jdgtt{\Gamma}{t}{I.\TBSK}$. 
    Since $N$ is closed, we have $\FV{t}=\emptyset$. 
    By Lemma~\ref{lem:smaller_dom} we have $\jdgtt{\emptyTPEnv}{t}{I.\TBSK}$. 
    By Lemma~\ref{lem:tp_eval} we have $\EVALT{t}\sattsk I.\TBSK$. 
    Hence by $\REDINS$ we have 
    \small
    $$N'=l_1::P~\!||~\!l_2::(I,R\uplus \{\EVALT{t}\}).$$
    \normalsize
    By $\jdgtt{\Gamma}{t}{I.sk}$ and $\jdgtC{\Gamma,\itfc}{(I,R)}$ we have 
    \small
    \begin{equation}\label{eq:ins_tp_TB'} 
    \jdgtC{\Gamma,\itfc}{(I,R\uplus \{\EVALT{t}\})} 
    \end{equation}
    \normalsize
    By \eqref{eq:ins_tp_P} and \eqref{eq:ins_tp_TB'}, it is not difficult to derive 
    $\jdgtN{\Gamma,\itfc}{N'}$. 
%    Taking $\Gamma'=\Gamma$ we then have the desired result for this case. 

  \item[\CASE $\REDDEL$:]
    The net $N$ is $l_1::\DEL{\TBID}{T}{\PRED}{l_2}.P~\!||~\!l_2::(I,R)$. 
    By 
    \small  
    $$\jdgtN{\Gamma,\itfc}{l_1::\DEL{\TBID}{T}{\PRED}{l_2}.P~\!||~\!l_2::(I,R)}$$
    \normalsize
    we have 
    $\jdgtP{\Gamma,\itfc}{\DEL{\TBID}{T}{\PRED}{l_2}.P}$ and $\jdgtC{\Gamma,\itfc}{(I,R)}$.

    Hence $\jdgta{\Gamma,\itfc}{\DEL{\TBID}{T}{\PRED}{l_2}}{\emptyTPEnv}$, and 
    \small
    \begin{equation}\label{eq:del_tp_P}
      \jdgtP{\Gamma,\itfc}{P}.
    \end{equation}
    \normalsize
    By the typing of the deletion action we have 
    $\jdgtT{\itfc(\TBID)}{T}{\Gamma''}$ for some $\Gamma''$ such that 
    $\jdgtpred{\Gamma,\Gamma''}{\PRED}$, and $\jdgtE{\Gamma}{l_2}{\loctp}$. 
    Since the transition can take place we have $I.\TBID=\TBID$. 
    By $\jdgtC{\Gamma,\itfc}{(I,R)}$ we have $\itfc(I.\TBID)=I.\TBSK$. 
    Hence we have $\jdgtT{I.\TBSK}{T}{\Gamma''}$. 
    By Lemma~\ref{lem:tp_T_sat_sk} we have $T\sattsk I.\TBSK$. 
    By $\jdgtC{\Gamma,\itfc}{(I,R)}$ we have 
    $\forall t\in R:\jdgtt{\Gamma}{t}{I.\TBSK}$. 
    Pick an arbitrary $t$ from the data set $R$, since $\FV{t}=\emptyset$, 
    by Lemma~\ref{lem:smaller_dom} we have $\jdgtt{\emptyTPEnv}{t}{I.\TBSK}$. 
    By Lemma~\ref{lem:sub_eval_n_err} we have $t/T\neq\Err$. 
    By the closedness of $N$ we have $\FV{\PRED}\subseteq \BV{T}=\dom{\Gamma''}$. 
    Hence by Lemma~\ref{lem:smaller_dom} we have $\jdgtpred{\Gamma''}{\PRED}$. 
    By Lemma~\ref{lem:sub_eval_n_err} we have $\EVALT{\PRED(t/T)}\neq\Err$. 
    We are now able to assert that 
    \small
    $$
    N'=l_1::P~\!||~\!l_2::(I,\{t\in R\mid \EVALT{\PRED(t/T)}\neq \TT\}). 
    $$
    \normalsize
    It is not difficult to derive 
    \small
    \begin{equation}\label{eq:del_tp_TB'}
      \jdgtC{\Gamma,\itfc}{(I,\{t\in R\mid \EVALT{\PRED(t/T)}\neq \TT\})}.
    \end{equation}
    \normalsize
    By \eqref{eq:del_tp_P} and \eqref{eq:del_tp_TB'} it is not difficult to derive 
    $\jdgtN{\Gamma,\itfc}{N'}$. 
%    Taking $\Gamma'=\Gamma$ we then have the desired result for this case. 

  \item[\CASE $\REDSEL$:] 
    The net $N$ is 
    \small
    $$l_0::\SEL{\LST{\CALTB}}{T}{\PRED}{t}{!\TBV}.P~\!||
        ~\!l_2::(I_1,R_1)~\!||~\!...~\!||~\!l_k::(I_k,R_k).$$ 
    \normalsize
    By the typing of $N$ under $\Gamma$ and $\itfc$ we have 
    $\jdgtP{\Gamma,\itfc}{\SEL{\LST{\CALTB}}{T}{\PRED}{t}{!\TBV}.P}$, 
    $\jdgtC{\Gamma,\itfc}{(I_1,R_1)}$, ..., and $\jdgtC{\Gamma,\itfc}{(I_n,R_n)}$; 
    hence we have 
    $\jdgta{\Gamma,\itfc}{\SEL{\LST{\CALTB}}{T}{\PRED}{t}{!\TBV}}{\Gamma''}$ for some $\Gamma''$
    such that $\jdgtP{(\Gamma,\Gamma''),\itfc}{P}$. 
    By the typing of the selection action we have 
    $\jdgtE{\Gamma,\itfc}{\CALTB_1}{\type_1}$, ..., $\jdgtE{\Gamma,\itfc}{\CALTB_n}{\type_n}$, 
    \small
    \begin{equation}\label{eq:sel_tp_T} 
      \jdgtT{\flattensk(\type_1\times...\times\type_n)}{T}{\Gamma'''}, 
    \end{equation}
    \normalsize
    for some $\Gamma'''$ such that 
    $\jdgtpred{\Gamma,\Gamma'''}{\PRED}$ and $\jdgtt{\Gamma,\Gamma'''}{t}{\type'}$, 
    for some $\type'$ such that $\Gamma''=[\TBV:\type']$. 

    Since the transition can take place, 
    $\prodSK(\LST{\CALTB}, \LSTStruct{l_i::(I_i,R_i)}{i})$ and
    $\prodR(\LST{\CALTB}, \LSTStruct{l_i::(I_i,R_i)}{i})$ are both defined. 
    By Lemma~\ref{lem:sk_prod} we have 
    $\prodSK(\LST{\CALTB}, \LSTStruct{l_i::(I_i,R_i)}{i})=\flattensk(\type_1\times...\times\type_n)$. 
    By Lemma~\ref{lem:tp_T_sat_sk} and \eqref{eq:sel_tp_T}, 
    we have $T\sattsk \flattensk(\type_1\times...\times\type_n)$; hence 
    it holds that $T\sattsk \prodSK(\LST{\CALTB}, \LSTStruct{l_i::(I_i,R_i)}{i})$. 

    Pick an arbitrary $t'\in \prodR(\LST{\CALTB}, \LSTStruct{l_i::(I_i,R_i)}{i})$. 
    By Lemma~\ref{lem:dt_prod} we have 
    \small
    $$\jdgtt{\emptyTPEnv}{t'}{\flattensk(\type_1\times...\times\type_n)}.$$ 
    \normalsize
    Hence by Lemma~\ref{lem:sub_eval_n_err} we have $t'/T\neq\Err$. 
    Since $N$ is closed we have $\FV{\PRED}\subseteq\BV{T}=\dom{\Gamma'''}$. 
    By Lemma~\ref{lem:smaller_dom}, we have $\jdgtpred{\Gamma'''}{\PRED}$. 
    By Lemma~\ref{lem:sub_eval_n_err}, we have $\EVALT{\PRED(t'/T)}\neq\Err$. 
    We also have $\FV{t}\subseteq\BV{T}=\dom{\Gamma'''}$. 
    By Lemma~\ref{lem:smaller_dom}, we have 
    \small
    \begin{equation}\label{eq:sel_tp_t}
    \jdgtt{\Gamma'''}{t}{\type'}.  
    \end{equation}
    \normalsize
    By Lemma~\ref{lem:sub_eval_n_err}, we have $\EVALT{t(t'/T)}\sattsk \type'$, 
    which implies $\EVALT{t(t'/T)}\neq\Err$. 
    We can now assert that 
    \small
    $$
    N'=l_0::P[(I',R')/\TBV] ~\!||~\!l_1::(I_1,R_1)~\!||~\!...~\!||~\!l_k::(I_k,R_k), 
    $$
    \normalsize
    where
    $$I'=(\bot,\prodSK(\LST{\CALTB}, \LSTStruct{l_i::(I_i,R_i)}{i})\downarrow^{T}_{t})$$
    and 
    $$R'=\{\EVALT{t\sigma}\mid \exists t'\in \prodR(\LST{\CALTB},\LSTStruct{l_i::(I_i,R_i)}{i}): 
            t'/T=\sigma \land \EVALT{\PRED\sigma}=\TT \}.$$ 
    We proceed with showing $\jdgtP{\Gamma,\itfc}{P[(I',R')/\TBV]}$.
    This boils down to showing $\jdgtE{\emptyTPEnv,\itfc}{(I',R')}{\type'}$ by Lemma~\ref{lem:subst_aP}
    since we already have $\jdgtP{(\Gamma,\TBV:\type'),\itfc}{P}$, 
    $\BV{P}\cap\dom{\Gamma,\TBV:\type'}=\emptyset$ 
    (which is obvious since variables are renamed apart and $\BV{N}\cap\dom{\Gamma}\neq\emptyset$ holds). 

    Take an arbitrary tuple from $R'$ it must be expressible as $\EVALT{t(t'/T)}$
    for some $t'\in \prodR(\LST{\CALTB},\LSTStruct{l_i::(I_i,R_i)}{i})$ such that 
    $\EVALT{\PRED(t'/T)}=\TT$. 
    By previous result we have $\EVALT{t(t'/T)}\sattsk\type'$. 
    It is not difficult to see that $\EVALT{\EVALT{t(t'/T)}}=\EVALT{t(t'/T)}$; 
    hence $\EVALT{\EVALT{t(t'/T)}}\sattsk\type'$. 
    By Lemma~\ref{lem:tp_eval}, we have $\jdgtt{\emptyTPEnv}{\EVALT{t(t'/T)}}{\type'}$. 
    On the other hand, we have 
    $I'.\TBSK=\prodSK(\LST{\CALTB}, \LSTStruct{l_i::(I_i,R_i)}{i})\downarrow^{T}_{t}=
      \flattensk(\type_1\times...\times\type_n)\downarrow^{T}_{t}$. 
    By \eqref{eq:sel_tp_T}, \eqref{eq:sel_tp_t} and Lemma~\ref{lem:tp_proj} we have 
    $\flattensk(\type_1\times...\times\type_n)\downarrow^T_t=\type'$. 
    Hence $I'.\TBSK=\type'$ and $\jdgtt{\emptyTPEnv}{\EVALT{t(t'/T)}}{I'.\TBSK}$. 
    It can be derived that $\jdgtE{\emptyTPEnv,\itfc}{(I',R')}{\type'}$. 
    We can now establish 
    \small
    \begin{equation}\label{eq:sel_tp_subst_P}
      \jdgtP{\Gamma,\itfc}{P[(I',R')/\TBV]}
    \end{equation}
    \normalsize

    Using \eqref{eq:sel_tp_subst_P} we can now easily derive $\jdgtN{\Gamma,\itfc}{N'}$. 

  \item[\CASE $\REDUPD$:] 
    The net $N$ is $l_1::\UPDATE{\TBID}{T}{\psi}{t}{\llocConst_2}.P~\!||~\!l_2::(I,R)$. 
    By the typing of $N$ under $\Gamma$ and $\itfc$, we have 
    $\jdgtP{\Gamma,\itfc}{\UPDATE{\TBID}{T}{\psi}{t}{\llocConst_2}.P}$ and 
    $\jdgtC{\Gamma,\itfc}{(I,R)}$; 
    hence we have $\jdgta{\Gamma,\itfc}{\UPDATE{\TBID}{T}{\psi}{t}{\llocConst_2}}{\emptyTPEnv}$
    and 
    \small
    \begin{equation}\label{eq:upd_tp_P}
    \jdgtP{\Gamma,\itfc}{P}.   
    \end{equation}
    \normalsize
    By the typing of the update action, we have 
    $\jdgtT{\itfc(\TBID)}{T}{\Gamma''}$, for some $\Gamma''$ such that 
    $\jdgtpred{\Gamma,\Gamma''}{\PRED}$ and $\jdgtt{\Gamma,\Gamma''}{t}{\itfc(\TBID)}$, 
    and $\jdgtE{\Gamma}{l_2}{\loctp}$. 

    By reasoning similar to the case $\REDDEL$, we have $T\sattsk I.\TBSK$. 
    Pick an arbitrary $t'\in R$, again by reasoning similar to the case $\REDDEL$, 
    we can deduce $t'/T\neq\Err$ and $\EVALT{\PRED(t'/T)}\neq\Err$. 
    Using $\jdgtC{\Gamma,\itfc}{(I,R)}$ we have $\jdgtt{\Gamma}{t'}{I.\TBSK}$ and 
    $\itfc(I.\TBID)=I.\TBSK$. 
    Since the transition can take place, we have $\TBID=I.\TBID$.
    Hence $\jdgtt{\Gamma}{t'}{\itfc(\TBID)}$, and $\jdgtt{\emptyTPEnv}{t'}{\itfc(\TBID)}$ holds
    by Lemma~\ref{lem:smaller_dom} and the fact that $\FV{t'}=\emptyset$. 
    Since $N$ is closed, $\FV{t}\subseteq\BV{T}=\dom{\Gamma''}$. 
    Hence $\jdgtt{\Gamma''}{t}{\itfc(\TBID)}$. 
    By Lemma~\ref{lem:sub_eval_n_err} we have $\EVALT{t(t'/T)}\sattsk \itfc(\TBID)$, 
    which gives $\EVALT{t(t'/T)}\sattsk I.\TBSK$. 
    Hence it also holds that $\EVALT{t(t'/T)}\neq\Err$. 
    We can now assert that 
    \small
    $$N'=l_1::P~\!||~\!l_2::(I,R'_1\uplus R'_2),$$ 
    \normalsize
    where $R'_1=\{t'\in R\mid \EVALT{\PRED(t'/T)}\neq\TT\}$ and 
    $R'_2=\{\EVALT{t\sigma}\mid \exists t': t'\in R\land t'/T=\sigma\land \EVALT{\PRED\sigma}=\TT\}$. 

    We proceed with showing $\jdgtC{\Gamma,\itfc}{(I,R'_1\uplus R'_2)}$. 
    Pick arbitrary $t''\in R'_1\uplus R'_2$. 
    We show that $\jdgtt{\Gamma}{t''}{I.\TBSK}$ with a case analysis on $t''$. 
    \begin{mylist1}
    \item[\CASE $t''\in R'_1$:] 
      We have $t''\in R$; hence it is obvious that $\jdgtt{\Gamma}{t''}{I.\TBSK}$ holds. 
    \item[\CASE $t''\in R'_2$:]
      We have $t''=\EVALT{t(t'/T)}$ for some $t'\in R$ such that $\EVALT{\PRED(t'/T)}=\TT$. 
      By previous results we have $\EVALT{t(t'/T)}\sattsk I.\TBSK$. 
      Using Lemma~\ref{lem:tp_eval} we have $\jdgtt{\emptyTPEnv}{\EVALT{t(t'/T)}}{I.\TBSK}$. 
      Weakening the typing environment we obtain $\jdgtt{\Gamma}{\EVALT{t(t'/T)}}{I.\TBSK}$. 
    \end{mylist1}
    It is now straightforward to obtain 
    \small
    \begin{equation}\label{eq:upd_tp_tb'}
      \jdgtC{\Gamma,\itfc}{(I,R'_1\uplus R'_2)}
    \end{equation}
    \normalsize
    Using \eqref{eq:upd_tp_P} and \eqref{eq:upd_tp_tb'} we can easily derive $\jdgtN{\Gamma,\itfc}{N'}$. 
  \item[\CASE $\REDAGGR$:] 
    The net $N$ is 
    $\LOCATED{\llocConst_1}{}{\AGGR{\TBID}{T}{\psi}{f}{T'}{\llocConst_2}.P}~\!||~\!
     \LOCATED{\llocConst_2}{}{(I,R)}$. 
    By the typing of $N$ under $\Gamma$ and $\itfc$, 
    we have $\jdgtP{\Gamma,\itfc}{\AGGR{\TBID}{T}{\psi}{f}{T'}{\llocConst_2}.P}$ and 
    $\jdgtC{\Gamma,\itfc}{(I,R)}$; 
    hence we have $\jdgta{\Gamma,\itfc}{\AGGR{\TBID}{T}{\psi}{f}{T'}{\llocConst_2}}{\Gamma''}$ 
    for some $\Gamma''$ such that 
    \small
    \begin{equation}\label{eq:aggr_tp_P}
      \jdgtP{(\Gamma,\Gamma''),\itfc}{P}. 
    \end{equation}
    \normalsize
    By the typing of the aggregation action we have 
    $\jdgtT{\itfc(\TBID)}{T}{\Gamma'''}$,  for some $\Gamma'''$ such that 
    $\jdgtpred{\Gamma,\Gamma'''}{\PRED}$, and
    $f:\msett{\itfc(\TBID)}\rightarrow\type'$, for some $\type'$ such that 
    $\jdgtT{\type'}{T'}{\Gamma''}$, 
    and $\jdgtE{\Gamma}{l_2}{\loctp}$. 

    By reasoning similar to the cases $\REDDEL$ and $\REDUPD$, 
    we have $T\sattsk I.\TBSK$. 
    Pick arbitrary $t\in R$, we also have $t/T\neq\Err$ and $\EVALT{\PRED(t/T)}\neq\Err$ analogously
    to the two aforementioned cases. 
    Picking arbitrary $\type_1$ and $\type_2$ such that $f:\msett{\type_1}\rightarrow \type_2$,
    we have $\type_1=\itfc(\TBID)$ and $\type_2=\type'$. 
    Thus by $\jdgtC{\Gamma,\itfc}{(I,R)}$ we can deduce $\jdgtt{\Gamma}{t}{\type_1}$. 
    Since $\FV{t}=\emptyset$, by Lemma~\ref{lem:smaller_dom} we have $\jdgtt{\emptyTPEnv}{t}{\type_1}$. 
    By Lemma~\ref{lem:tp_eval} and the obvious fact $\EVALT{t}=t$, we have $t\sattsk \type_1$. 
    By Lemma~\ref{lem:tp_T_sat_sk} we have $T'\sattsk \type_2$. 
    We can now assert 
    \small
    $$
      N'=l_1::P(t'/T')~\!||~\! l_2::(I,R), 
    $$
    \normalsize
    where $t'=f(\{t\in R\mid \EVALT{\PRED(t/T)}=\TT\})$. 

    Since the range of $f$ is of type $\type_2$,
    we have $t'\sattsk \type_2$, or $\EVALT{t'}\sattsk \type_2$. 
    By Lemma~\ref{lem:tp_eval} we have $\jdgtt{\emptyTPEnv}{t'}{\type_2}$. 
    Since $t'$ is a function value, its components are values. 
    Suppose $t'=\val_1,...,\val_n$. 
    We have some $\msettp^{1}$, ..., $\msettp^{n}$ such that 
    $\jdgtE{\emptyTPEnv}{\val_1}{\msettp^{1}}$, ..., $\jdgtE{\emptyTPEnv}{\val_n}{\msettp^{n}}$. 
    By $\jdgtT{\type_2}{T'}{\Gamma''}$, we have 
    $T'=!\avar_1,...,!\avar_n$, where each $\avar_j$ is some $x_j$ or $u_j$, 
    $\jdgtT{\msettp^1}{!\avar_1}{[\avar_1:\msettp^1]}$, ..., 
    $\jdgtT{\msettp^n}{!\avar_n}{[\avar_n:\msettp^n]}$, 
    $\Gamma''=[\avar_1:\msettp^1,...,\avar_n:\msettp^n]$. 
    By Lemma~\ref{lem:sub_eval_n_err} we have $t'/T'\neq\Err$; 
    hence $t'/T'=[\val_1/\avar_1,...,\val_n/\avar_n]$. 
    By \eqref{eq:aggr_tp_P}, and Lemma~\ref{lem:subst_aP}, we have 
    \small
    \begin{equation}\label{eq:aggr_tp_P'}
      \jdgtP{\Gamma,\itfc}{P(t'/T')}
    \end{equation}
    \normalsize
    Using \eqref{eq:aggr_tp_P'} we can easily derive $\jdgtN{\Gamma,\itfc}{N'}$. 

  \item[\CASE $\REDCREATE$:] Trivial. 

  \item[\CASE $\REDDROP$:] Trivial. 

  \item[\CASE $\REDEVALP$:] 
    The net $N$ is $l_1::\EVAL{P}{l_2}.P'~\!||~\!l_2::\NILP$. 
    By the typing of $N$ under $\Gamma$ and $\itfc$, we have 
    $\jdgtP{\Gamma,\itfc}{\EVAL{P}{l_2}.P'}$ and $\jdgtP{\Gamma,\itfc}{\NILP}$; 
    hence we have $\jdgta{\Gamma,\itfc}{\EVAL{P}{l_2}}{\emptyTPEnv}$ and 
    $\jdgtP{\Gamma,\itfc}{P'}$. 
    By the typing of the {\sf eval} action we have $\jdgtP{\Gamma,\itfc}{P}$, 
    and $\jdgtE{\Gamma}{l_2}{\loctp}$. 
    We have $N'=l_1::P'~\!||~\! l_2::P$, and 
    we can easily derive $\jdgtN{\Gamma,\itfc}{N'}$. 

  \item[\CASE $\REDFORTT$:]
    The net $N$ is $l::\FOREACH{(I,R)\!}{T}{\psi}{\ordr}{\!P}$. 
    By the typing of $N$ under $\Gamma$ and $\itfc$ we have 
    $\jdgtE{\Gamma,\itfc}{(I,R)}{\type}$ for some $\type$ such that 
    $\jdgtT{\type}{T}{\Gamma''}$, for some $\Gamma''$ such that $\jdgtpred{\Gamma,\Gamma''}{\PRED}$
    and $\jdgtP{(\Gamma,\Gamma''),\itfc}{P}$. 

    On the other hand, according to the rule $\REDFORTT$ we have 
    \small
    $$N'=l::P(t_0/T); (\FOREACH{(I,R\setminus\{t_0\})\!}{T}{\PRED}{\ordr}{\! P}),$$
    \normalsize
    where $t_0\in\minimal(R,\ordr)$. 

    Since $t_0\in R$, by $\jdgtE{\Gamma,\itfc}{(I,R)}{\type}$ we have 
    $\jdgtt{\Gamma}{t_0}{I.\TBSK}$ and $\type=I.\TBSK$. 
    By reasoning similar to that in the case $\REDAGGR$, we can deduce 
    $\jdgtP{\Gamma,\itfc}{P(t_0/T)}$. 
    It is not difficult to deduce $\jdgtE{\Gamma,\itfc}{(I,R\setminus\{t_0\})}{\type}$; 
    hence $\jdgtP{\Gamma,\itfc}{\FOREACH{(I,R\setminus\{t_0\})\!}{T}{\PRED}{\ordr}{\! P}}$ can be established.
    Now it is straightforward to derive $\jdgtN{\Gamma,\itfc}{N'}$. 

  \item[\CASE $\REDFORFF$:]
    Similar to the case $\REDFORTT$, we have 
    $\jdgtE{\Gamma,\itfc}{(I,R)}{\type}$ for some $\type$ such that 
    $\jdgtT{\type}{T}{\Gamma''}$, for some $\Gamma''$ such that $\jdgtpred{\Gamma,\Gamma''}{\PRED}$
    and $\jdgtP{(\Gamma,\Gamma''),\itfc}{P}$. 

    Pick an arbitrary $t_0\in R$. 
    By $\jdgtE{\Gamma,\itfc}{(I,R)}{\type}$ we have $\jdgtt{\Gamma}{t_0}{I.\TBSK}$ and $\type=I.\TBSK$. 
    By reasoning analogous to that used in the case $\REDDEL$, 
    we can deduce that $t_0/T\neq \Err$ and $\EVALT{\PRED(t_0/T)}\neq \Err$. 
    Hence $N'=l::\NILP$ and it is trvial to establish $\jdgtN{\Gamma,\itfc}{N'}$. 

  \item[\CASE $\REDSEQTT$:] Straightforward. 

  \item[\CASE $\REDSEQFF$:] Straightforward. 

  \item[\CASE $(\mrm{CALL})$:]
    The net $N$ is $l::A(\LST{e})$. 
    By $\jdgtN{\Gamma,\itfc}{l::A(\LST{e})}$ we have 
    $\jdgtE{\Gamma}{e_1}{\type_1}$, ..., $\jdgtE{\Gamma}{e_n}{\type_n}$, 
    and $\jdgtP{(\avar_1:\type_1,...,\avar_n:\type_n),\itfc}{P}$. 
    Since $N$ is closed, none of $e_1$, ..., and $e_n$ contains variables. 
    By Lemma~\ref{lem:smaller_dom} we have 
    $\jdgtt{\emptyTPEnv}{e_1}{\type_1}$, ..., $\jdgtt{\emptyTPEnv}{e_n}{\type_n}$. 
    By Lemma~\ref{lem:tp_eval}, we have $\EVALT{e_1}\sattsk \type_1$, ..., $\EVALT{e_n}\sattsk \type_n$, 
    or $v_1\sattsk\type_1$, ..., $v_n\sattsk\type_n$. 
    Using Lemma~\ref{lem:tp_eval} again, we have 
    $\jdgtt{\emptyTPEnv}{v_1}{\type_1}$, ..., $\jdgtt{\emptyTPEnv}{v_n}{\type_n}$. 
    By Lemma~\ref{lem:subst_aP} we have $\jdgtP{\emptyTPEnv,\itfc}{P[v_1/\avar_1]...[v_n/\avar_n]}$. 
    By Lemma~\ref{lem:fv_in_dom} we have $\FV{P[v_1/\avar_1]...[v_n/\avar_n]}=\emptyset$. 
    Weakening the typing environment we get $\jdgtP{\Gamma,\itfc}{P[v_1/\avar_1]...[v_n/\avar_n]}$. 
    Hence it is straightforward to establish $\jdgtN{\Gamma,\itfc}{l::P[v_1/\avar_1]...[v_n/\avar_n]}$. 

  \item[\CASE $\REDPAR$:] 
    The net $N$ is $N_1||N_2$. 
    By $\jdgtN{\Gamma,\itfc}{N_1||N_2}$, 
    we have $\jdgtN{\Gamma,\itfc}{N_1}$ and $\jdgtN{\Gamma,\itfc}{N_2}$. 
    By the rule $\REDPAR$ we have $\STEP{N_2}{N_1}{N'_1}$. 
    By the \IH we have $\jdgtN{\Gamma,\itfc}{N'_1}$. 
    For $N'=N'_1||N_2$ it is straightforward to establish $\jdgtN{\Gamma,\itfc}{N'}$. 

  \item[\CASE $\REDRES$:] Trivial. 

  \item[\CASE $\REDEQUIV$:] Straightforward with the help of Lemma~\ref{lem:tp_equiv}. 

  \end{mylist1}
  This completes the proof. 
\end{proof}

We restate Theorem~\ref{thm:subj_red} and point out that 
it is a consequence of Lemma~\ref{lem:subj_red}. 
\begingroup
\def\thethm{\ref{thm:subj_red}}
\begin{thm}[Subject Reduction] 
For a closed net $N$, if $\BV{N}\cap\dom{\Gamma}=\emptyset$, 
$\jdgtN{\Gamma,\itfc}{N}$, and $\vdash N\rightarrow N'$, 
then $\jdgtN{\Gamma,\itfc}{N'}$. 
\end{thm}
\addtocounter{thm}{-1}
\endgroup
\begin{proof}
Theorem~\ref{thm:subj_red} can be directly obtained from Lemma~\ref{lem:subj_red} 
by instantiating the latter with $\envnet=\NILN$. 
\end{proof}

\begin{lem}\label{lem:efficiency}
Suppose the well-typedness of each procedure body has already been decided
with a typing environment that associates its formal parameters 
to their delcared types. 
With a given $\itfc$, 
and the assumptions that it takes $O(1)$ time to 
\begin{itemize}
  \item produce the types of all constant expressions, 
  \item decide the equality/inequality of table identifiers and types, 
  \item construct a singleton typing environment, 
  \item construct the extension $(\Gamma_1,\Gamma_2)$ of a typing environment $\Gamma_1$, and
  \item look up environments $\Gamma$ (resp. $\itfc$) for variables (resp. table identifiers), 
\end{itemize}
and that it takes $O(n)$ time to 
\begin{itemize}
  \item form an $n$-ary product type, and 
  \item perform the operation $\flattensk(\type_1\times...\times\type_n)$, 
\end{itemize}
the following results hold
\begin{enumerate}
  \item $\jdgtE{\Gamma}{e}{\type}$ can be decided in time $O(\sz{e})$
  \item $\jdgtpred{\Gamma}{\PRED}$ can be decided in time $O(\sz{\PRED})$
  \item $\jdgtt{\Gamma}{t}{\type}$ can be decided in time $O(\sz{t})$
  \item $\jdgtT{\type}{T}{\Gamma'}$ can be decided in time $O(\sz{T})$ 
  \item $\jdgtE{\Gamma,\itfc}{\CALTB}{\type}$ can be decided in time $O(\sz{\CALTB})$
  \item $\jdgta{\Gamma,\itfc}{a}{\Gamma'}$ can be decided in time $O(\sz{a})$, and \\
        $\jdgtP{\Gamma,\itfc}{P}$ can be decided in time $O(\sz{P})$
  \item $\jdgtC{\Gamma,\itfc}{C}$ can be decided in time $O(\sz{C})$
  \item $\jdgtC{\Gamma,\itfc}{N}$ can be decided in time $O(\sz{N})$
\end{enumerate}
\end{lem}

\begin{proofsketch} 
  A general observation is that the type system does not make use of subtyping and
  for each syntactical entity, there can only be one typing rule applicable. 
  On this basis, the proofs of (1)--(8) are by straightforward induction on 
  the structure of the corresponding syntactical entities. 
\end{proofsketch}

\begingroup
\def\thethm{\ref{thm:efficiency}}
\begin{thm}[Efficiency of Type Checking]
With a given $\itfc$, 
the time complexity of type checking net $N$ is linear in $\mi{size}(N)$, provided that 
it takes $O(1)$ time to 
\begin{itemize}
  \item determine the types of all constant expressions, 
  \item decide the equality/inequality of table identifiers and types, 
  \item construct a singleton typing environment, 
  \item construct the extension $(\Gamma_1,\Gamma_2)$ of a typing environment $\Gamma_1$ with $\Gamma_2$, and
  \item look up environments $\Gamma$ (resp. $\itfc$) for variables (resp. table identifiers), 
\end{itemize}
and that it takes $O(n)$ time to 
\begin{itemize}
  \item form an $n$-ary product type, and 
  \item perform the operation $\flattensk(\type_1\times...\times\type_n)$.
\end{itemize}
\end{thm}
\addtocounter{thm}{-1}
\endgroup
\proof
  We type check $N$ in two steps:
  \begin{enumerate}
    \item we type check all the bodies of the procedures defined, 
    with typing environments that associate their formal parameters to their declared types; 
    \item if one of these bodies fails to type check, we abort the task, declaring that 
    $N$ does not type check;
    otherwise we type check $N$. 
  \end{enumerate}
  In both steps, when a procedure call is encountered (in a procedure body or $N$), 
  we simply ignore the procedure body. 

  It takes $O(\sum_{A(\LSTStruct{\avar_i:\tau_i}{i})\triangleq P} \sz{P})$ time to perform step (1). 
  By Lemma~\ref{lem:efficiency}, it takes $O(\sz{N})$ time to perform step (2). 
  Therefore it is obvious that the time complexity of type checking $N$ is linear in $\mi{size}(N)$, 
  which is 
  \small
  $$\sz{N}+\sum_{A(\LSTStruct{\avar_i:\tau_i}{i})\triangleq P} \sz{P}.\eqno{\qEd}$$
  \normalsize

%%% Local Variables:
%%% mode: latex
%%% TeX-master: "paper"
%%% End:

\end{document}